\documentclass[12pt]{article}
\usepackage[utf8]{inputenc}
\usepackage[english]{babel}
\usepackage{amsmath}
\usepackage{amsthm}
\usepackage{color}
\usepackage{amssymb}
\usepackage{natbib}
\usepackage{graphicx}
\usepackage{multirow}
\usepackage{titling}
\usepackage{booktabs}
\usepackage{enumitem}
\usepackage{float}
\restylefloat{table}
\usepackage{dsfont}
\usepackage{bbm}
\usepackage{bm}
\usepackage[graphicx]{realboxes}

\setcitestyle{round}

\newcommand{\blind}{1}

\newcommand{\pr}{\prime}

\newcommand{\bs}{\boldsymbol{s}}

\newtheorem{prop}{Proposition}  [section]

\newtheorem{assumption}{Assumption}
\newtheorem{theorem}{Theorem}
\newtheorem{lemma}{Lemma}

\newtheorem{remark}{Remark}

\theoremstyle{definition}

\def\spacingset#1{\renewcommand{\baselinestretch}%
{#1}\small\normalsize} \spacingset{1}

\begin{document}

\if1\blind
{
  \title{\bf Preprocessing noisy functional data:\\ a multivariate perspective}
  \author{Siegfried H\"ormann\\
    Institute of Statistics, Graz University of Technology\\
    and \\
    Fatima Jammoul\thanks{Corresponding author.
Email: f.jammoul@tugraz.at}\hspace{.2cm} \\
    Institute of Statistics, Graz University of Technology}
  \maketitle
} \fi

\if0\blind
{
  \bigskip
  \bigskip
  \bigskip
  \begin{center}
    {\LARGE\bf Preprocessing functional data without smoothing}
\end{center}
  \medskip
} \fi

\bigskip
\begin{abstract}
We consider functional data which are measured on a discrete set of observation points. Often such data are measured with additional noise. We explore in this paper the factor structure underlying this type of data. We show that the latent signal can be attributed to the common components of a corresponding factor model and can be estimated accordingly, by borrowing methods from factor model literature. We also show that principal components, which play a key role in functional data analysis, can be accurately estimated after taking such a multivariate instead of a `functional' perspective. In addition to the estimation problem, we also address testing of the null-hypothesis of iid noise. While this assumption is largely prevailing in the literature, we believe that it is often unrealistic and not supported by a residual analysis.
\end{abstract}

\noindent%
\vfill
\newpage
\tableofcontents

\section{Introduction}

Functional data analysis (FDA) is concerned with the analysis of data that can naturally be described as curves. In mathematical terms data are modeled as random curves $(X(s)\colon s\in\mathcal{S})$, where $\mathcal{S}$ is some continuum. Examples where such data arise are very diverse, ranging from high frequency asset price curves over growth curves or pollution level curves, to 2D satellite images or fMRI scans.  For a simple presentation we assume without loss of generality that $\mathcal{S}=[0,1]$.  With technological advances recording and storing this type of data becomes more and more common and hence the corresponding FDA literature has seen a big upsurge over the past years. For an introduction to the topic we refer, for example, to the textbooks of \citet{ferratyvieu:2006}, \cite{horvath:kokoszka:2012} or \citet{ramsaysilverman05}.

In practice functional data are not fully observed, but sampled on a discrete set of time points. Consider  functional observations 
$(X_t(s)\colon 0\leq s\leq 1)$, $t\geq 1$, and assume we have measurements of it at time points $0\leq s_1< s_2<\ldots <s_p\leq 1$. A very common additional working hypothesis in FDA literature is that these measurements come with an additional error, so that we actually observe
\begin{equation}\label{signoise}
Y_t=(X_t(s_1),\ldots, X_t(s_p))^\prime+(U_{t1},\ldots, U_{tp})^\prime=:X_t(\bs)+U_t.
\end{equation}
(Henceforth we are going to write $\bs$ for $(s_1,\ldots, s_p)$ and use the convention that $g(\bs)$ denotes $(g(s_1),\ldots, g(s_p))^\pr$.)
The errors $U_t$ can, for example, be related to measurement errors. In this paper we focus on the setting where all data are observed at the same time points $s_i$. This is typically the case for machine recorded data. The goal then is to separate the errors from $X(\bs)$. %A distinctive feature between $x$ and $u$ is that the components of $x$ are correlated, while those of $u$ are not. 
Most papers (including those cited below) assume that the components $(U_{ti}\colon 1\leq i\leq p)$  are i.i.d.\ with zero mean and variance $\sigma^2_U>0$ and that $X_t$ and $U_t$ are independent. %We will later argue that this assumption can be relaxed. 
To recover the full curve $X_t$ or the discretisation $X_t(\bs)$  (henceforth we refer to both objects as the \emph{signal}) a variety of fitting techniques exist. The goal is to acquire an estimate $\hat X_t$ or $\hat X_t(\bs)$ that is close to the true latent signal. A very common technique is the basis expansion approach, explained thoroughly in  \citet{ramsaysilverman05}. Here, the fitted curve is a linear combination of suitable basis functions. 
%The required coefficients are obtained by regressing the raw measurements onto a design generated by discretely evaluated basis functions. 
Most popular choices are the Fourier basis or B-splines. By adding a roughness penalty to the least squares criterion the smoothness of the curves can be controlled. %We refer again to \citet{ramsaysilverman05}. 
Other approaches employ kernel smoothing (see e.g.,\ \citet{ramsaysilverman05,Stadtmueller2015}), with a detailed description in \citet{wandjones95}, where a kernel function is used to estimate the signal. Most prominent is the Nadaraya-Watson regression estimate as proposed in \citet{nadaraya64} and \citet{watson64}. 

Nevertheless, the recovery of the true signal remains a delicate problem when analysing real data.
The degree of smoothness of the latent curves, which is needed to choose the appropriate number of basis functions or the bandwidth of a kernel smoother, is typically unknown and then the result of the analysis is influenced by a non-verifiable working hypothesis. 

Cross-validation (CV) may look like an attractive route, since the parameter choices then become data driven. To illustrate that the problem is still challenging, we look at the synthetic example in Figure~\ref{fig2}. The two rows in the graph illustrate realisations from two random samples (four observations each). The dots represent the raw data $y$ and the green lines the underlying signals $X(s)$. In the first example (top row) the data generating process (DGP) is such that all intraday data have a higher variability at $s=0.7$, while in the second example (bottom row) the actual data form straight lines with one outlying measurement error in the second observation.  With curve-by-curve fitting techniques, it is not possible to distinguish between a measurement error or some systematic structure in the signal. We would need further measurements in the neighbourhood of $s=0.7$, where the signal is erratic. So the accuracy in recovering the signal is tied to its smoothness, and the relevant question is whether \emph{$p$ is  large enough relative to the degree of smoothness} to sufficiently justify a certain approach.

Suppose now that we want to estimate the signals for a sample $Y_1,\ldots, Y_T$. While $p\to\infty$ is favourable for curve-by-curve fitting techniques, i.e.,\ when $\hat X_t=\mathcal{F}_t(y_t)$ for some operator $\mathcal{F}_t$, say, \emph{no improvement in the fit can be expected if we let $T$ grow.} This suggests employing a fitting procedure which takes into account the entire sample, i.e.,\ $\hat X_t= \mathcal{G}_t(Y_1,\ldots, Y_T)$.
\begin{figure}[ht]
\begin{center}
\includegraphics[width=12cm]{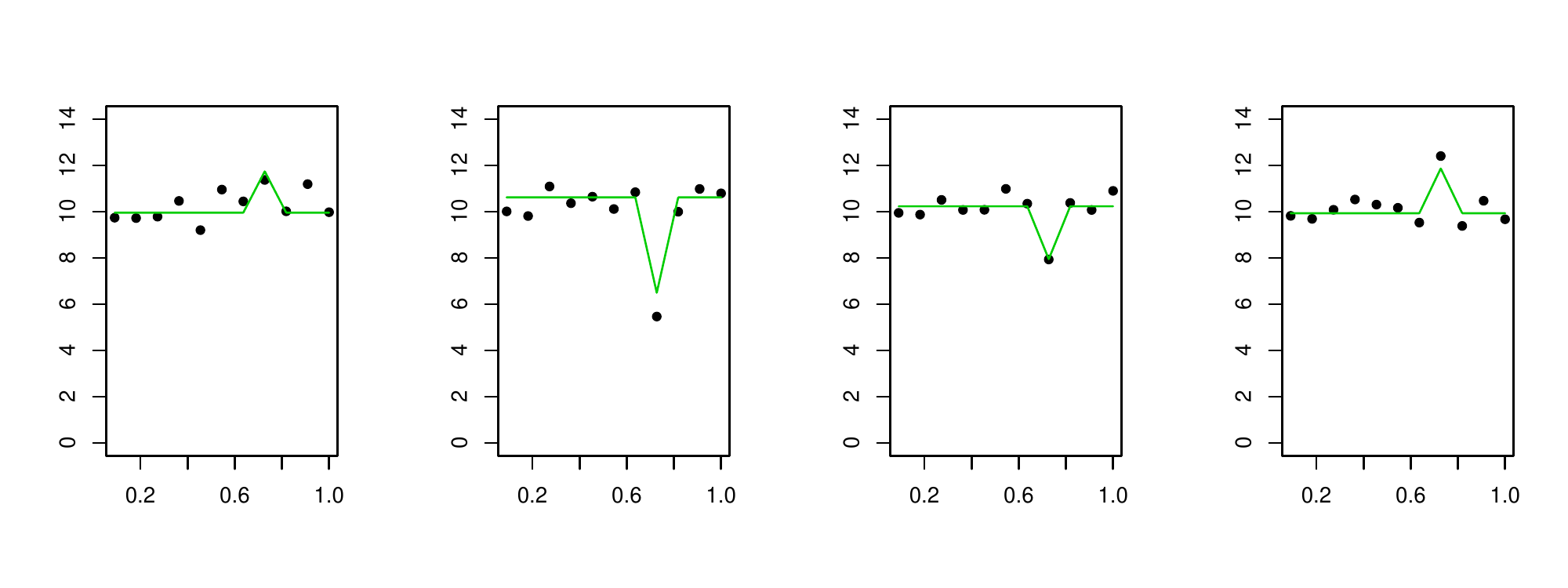}
\includegraphics[width=12cm]{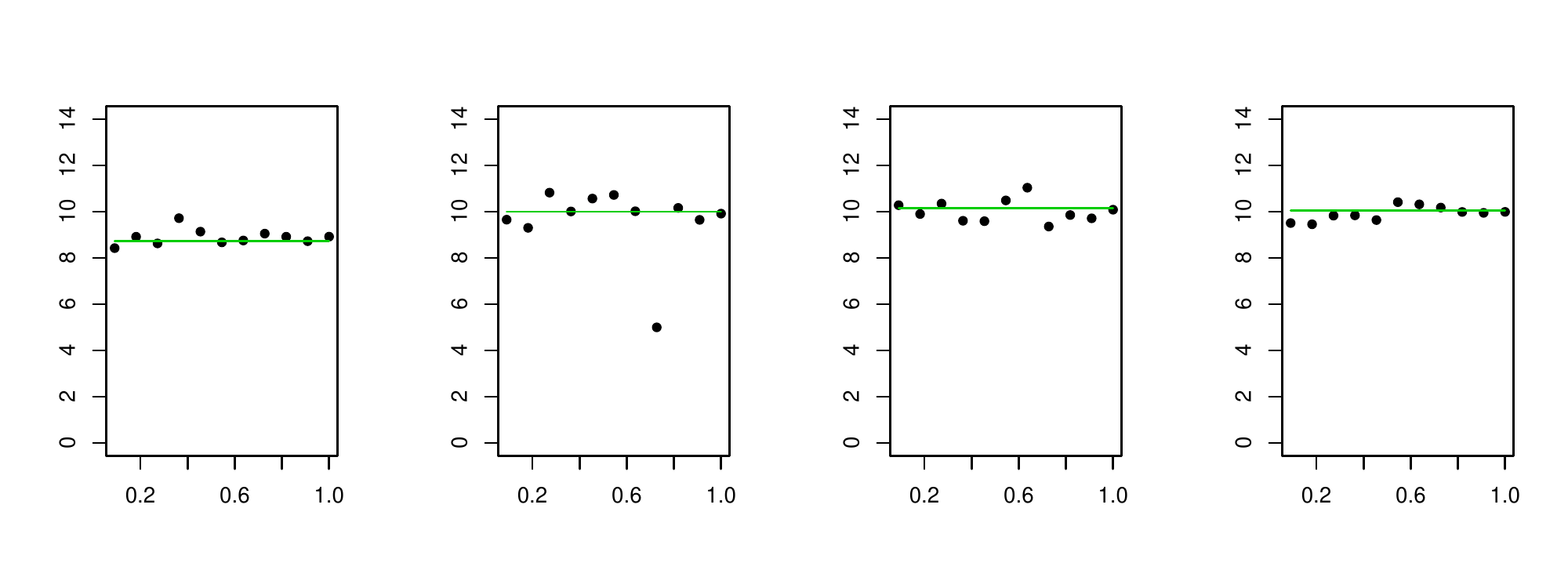}
\caption{In this illustrative example, the dots correspond to measurements $Y$ and the green solid lines to the actual signals $X(s)$.}
\label{fig2}
\end{center}
\end{figure}
\cite{Staniswalis:Lee:1998} proposed a method of this type. Under model \eqref{signoise}, say with iid errors, we have that $\mathrm{Var}(Y_t)=\mathrm{Var}(X_t(\bs))+\sigma^2 I_p$ (with $I_p$ being the identity matrix in $\mathbb{R}^{p\times p}$).  Hence, \begin{align*}
(\widehat{\mathrm{Var}}(Y))_{k,k'}=\frac{1}{T}\sum_{t=1}^N (Y_{tk}-\bar{Y}_{\cdot k})(Y_{tk'}-\bar{Y}_{\cdot k'})\approx \mathrm{Cov}(X(s_k),X(s_{k'}))+\sigma^2I\{s_k=s_{k'}\},
\end{align*}
where $Y_{tk}$ is the $k$-th component of $Y_t$ and $\bar{Y}_{\cdot k}=\frac{1}{T}\sum_{t=1}^T Y_{tk}$. Their idea is to estimate the covariance kernel $\mathrm{Cov}(X_t(s),X_t(s'))$ by smoothing the $p\times p$ matrix $\mathrm{Var}(Y_t)$.
The smoothing step is used to remove the discontinuity  
on the diagonal caused by the noise. They propose a kernel smoother, but other approaches such as spline smoothing can be used instead.  Denote the smooth covariance kernel estimator by $\widehat{\mathrm{Cov}}(X_t(s),X_t(s'))$.  In a second step the estimated (and also smooth) eigenfunctions $\hat\psi_i$ of $\widehat{\mathrm{Cov}}(X_t(s),X_t(s'))$, which form an orthonormal basis, are used to expand the $X_t$ along this basis. The resulting expansion provides an estimator of the signal. For further details we refer to \cite{Staniswalis:Lee:1998}.
 Several subsequent articles are based on variants of this approach. 
For example, the well known PACE algorithm established in \cite{yaoetal:2005} makes use of the idea described above. 
Other important references focus on the estimation of eigenfunctions and eigenvalues from the smoothed covariance and derive asymptotic results. We refer in particular to \cite{Hall:Mueller:Wang:2006} and \cite{Mueller:Stadtmueller:Yao:2006}. 

In our recent paper \citet{hormann:jammoul:2021} we have analysed yet another estimator for the signals $X_t(\bs)$, which is not based on any smoothing, but operates with the raw data. Similar as in \cite{Staniswalis:Lee:1998} it uses the entire sample $Y=(Y_1,\ldots, Y_T)$. The approach simply consists of projecting the observations on the (discrete) principal components related to $Y$. We have shown that this elementary and easy to implement technique leads under very mild assumptions to a uniformly  (over $t$ and $s_i$) consistent estimator, with explicit convergence rates. While the idea may seem  similar to the functional PCA based estimator of \cite{Staniswalis:Lee:1998}, it is inspired by a completely different perspective, which we will explain in Section~\ref{s:basics} of this paper.  We show in the next section that functional data sampled as in \eqref{signoise} follow some \emph{factor model} (see e.g.\ \citet{mardia:kent:bibby:1979}). The signal underlying the discretised observations is related to the \emph{common components} in the factor model and thus a natural strategy is to estimate these common components. The PCA technique used in \citet{hormann:jammoul:2021} is only one possible estimation scheme for factor models. Alternatively, one may resort to diverse likelihood techniques which exist in the literature.

While in \citet{hormann:jammoul:2021} only the discretized signals $X_t(\bs)$ are estimated, we can obtain a curve by diverse interpolation schemes. In Section~\ref{s:full} we show that under regularity conditions simple linear interpolation leads to a uniformly consistent estimator of the whole curve. In Section~\ref{s:eigenfun} we show that the related functional principal components also can be estimated without using on a smoothing step. Special attention will be given to the practical aspects of the implementation of the method and also to the analysis of the residual errors. Most papers impose iid assumptions on the error components $(U_{ti}\colon 1\leq i\leq p)$. However, the interpretation of the errors (aside from measurement errors) is broadly ignored and a thorough residual analysis, which is needed for corresponding model diagnostics, is also missing in most papers. We devote Section~\ref{s:fit} to adequate diagnostic tools. The key idea there also blazes a trail estimating the number of factors in our model, which will be discussed in some detail.  In Section~\ref{s:sim} we provide comprehensive simulation studies. For this purpose we create synthetic samples out of {\tt PM10} (particular matter) data, which allows us to have a controllable and at the same time realistic data base, as seen in Section \ref{s:smoothsim}.  In order to demonstrate that our method doesn't require specific smoothness conditions, we provide in Section \ref{s:roughsim} an additional simulation study considering signals which may contain a discontinuity. In Section~\ref{s:real} we analyse and compare two real data sets of Canadian temperature data.  We conclude in Section~\ref{s:outlook}. 

\section{Factor model representation}\label{s:basics}

We consider a set of functional data $X_1,\ldots, X_T$ defined on a common probability space.  Throughout we assume that observations are i.i.d.\ or form a general stationary functional process. The curves $(X_t(s)\colon s\in [0,1])$ are square integrable on $[0,1]$, and hence can  be expanded along a sequence of orthogonal basis functions $\{b_k(s)\colon k\geq 1\}$, (for example the Fourier basis). Then we have $X_t(s)=\sum_{k\geq 1}\langle X_t,b_k\rangle b_k(s)$, where $\langle a,b\rangle= \int_0^1 a(v) b(v) dv$. The convergence is in general only in $L^2$ sense, but under mild regularity conditions on path properties of $X_t$ we can also obtain pointwise or even uniform convergence. For example, if the covariance kernel $\Gamma^X(s,s^\prime):=\mathrm{Cov}(X_t(s),X_t(s^\prime))$ is continuous and we set $b_k=\varphi_k$, which denote the eigenfunctions of $\Gamma^X(s,s^\prime)$, then we obtain as a consequence of Mercer's theorem \citep{gohberg:goldberg:kaashoek:2000}, that
 \begin{equation}\label{e:mercer}
\sup_{s\in [0,1]}E\left|X_t(s)-\mu(s)-\sum_{\ell= 1}^L x_{t\ell}\varphi_\ell(s)\right|^2\to 0,\quad L\to\infty.
\end{equation}
The functions $\varphi_k$ are the so-called functional principal components, and define an optimal orthogonal basis system, in the sense of minimising the mean square error
$$
\int_0^1 E\left|X_t(s)-\mu(s)-\sum_{\ell= 1}^L \langle X_t,b_k\rangle b_k(s)\right|^2ds
$$
with respect to the basis functions $(b_k)$.
In typical applications the approximation error is already very close zero with small $L$ (say $L=5$) or at most moderate sized values of $L$ (say $L=20$), so that assuming a finite dimensional representation
\begin{equation}\label{finit}
X_t(s)= \sum_{k= 1}^L\langle X_t,b_k\rangle b_k(s),\quad \text{for some $L\geq 1$}
\end{equation}  
is no more than a theoretical restriction, which imposes no practical limitation of generality, if $L$ is allowed to be chosen large enough.

A basic requirement for our proposed method is that all curves are sampled at the same time points $0\leq s_1< s_2<\ldots <s_p\leq 1$. This is a very common setting for machine recorded data. We note that sampling points need not be equidistant though. We will assume throughout a general signal-plus-noise structure as in \eqref{signoise}.  
The following representation theorem for functional data observed as in \eqref{signoise} holds. 
\begin{prop}
Suppose \eqref{signoise} and \eqref{finit} hold. Let  $U_t$ be independent of $X_t$ and assume that $EU_t=0$ and $\mathrm{Var}(U_t)=\mathrm{diag}$. \emph{Then $Y_t$ follows an $L$-factor model}. 
\end{prop}

\begin{proof}
We show that there exists a matrix $B\in\mathbb{R}^{p\times L}$ such that
\begin{equation}\label{e:factmod}
Y_t=\mu(\bs)+B F_t+U_t,
\end{equation}
where $\mu(s)=EX_t(s)$, $EF_t=0$, $\mathrm{Var}(F_t)=I_L$ (the identity matrix in $\mathbb{R}^L$) and $\mathrm{Cov}(F_t,U_t)=0$.
We note that by the imposed stationarity the covariance kernel $\Gamma^X(s,s^\prime)$  doesn't depend on $t$.  
Using \eqref{finit} the Karhunen-Lo\`eve expansion gives
\begin{equation}\label{e:KLexp}
X_t(s)=\mu(s)+\sum_{\ell=1}^L x_{t\ell}\varphi_\ell(s),
\end{equation}
where  $\varphi_\ell(s)$ are the eigenfunctions of the covariance operator $\Gamma^X$ and $x_{t\ell}=\int_0^1(X_t(s)-\mu(s))\varphi_\ell(s) ds$. 
The scores $(x_{t\ell}\colon \ell\geq 1)$ are uncorrelated and  $\mathrm{Var}(x_{t\ell})=\lambda_\ell$, where $\lambda_\ell$ are the eigenvalues of $\Gamma^X$ (in decreasing order). See e.g.,\ \citet{bosq:2000} for details.

Define now the matrix
$
B:=(\sqrt{\lambda_1}\varphi_1(\bs),\ldots,\sqrt{\lambda_L}\varphi_L(\bs)).
$
Moreover, define $$F_t=(x_{t1}/\sqrt{\lambda_1},\ldots, x_{tL}/\sqrt{\lambda_L})^\prime.$$ This yields the desired representation. 
\end{proof}

In factor model language $B F_t$  are called the \emph{common components} of $Y_t$ and our problem is reduced to the estimation of these common components. For this purpose we can resort to a rich literature, especially from macroeconomics, where factor models are used to model markets with many assets. See, for example,  \citet{stockwatson2002a, stockwatson2002b} and \citet{forni:lippi}. In this context, \cite{chamberlain1983} have shown that it is useful to allow also for a certain degree of dependence in the \emph{idiosyncratic noise components} $U_{ti}$. This setting then refers to \emph{approximate factor models}. Some of the features employed in econometric applications are natural and useful in our context, too:
(1) The dimension $p$ of our sampling points $\bs$ is large and allowed to diverge with increasing sample size.
(2) The functional data $X_t$ may be time-dependent, i.e.,\ form a functional time series.
(3) In a realistic framework, the errors $(U_{ti}\colon 1\leq i\leq p)$ in \eqref{signoise} might be correlated at small lags.

Next to conceptual papers proposing different variants of factor models, there is also a profound literature on estimation theory for these models. 
In particular we refer to the papers of \citet{bai2003}, \cite{baili2012}, \citet{choi2012}, \citet{bailiao2012}, \citet{fanetal2013} and \citet{baili2016}. In context of dynamic factor models we refer to \citet{fornietal2000, fornietal2005}. How these methods may be used in the current context will be discussed in the next section.  

We conclude here by an important remark on the interpretation of errors.

\begin{remark}
It is common in FDA to smooth data, even if by their very nature they come without relevant measurement errors (e.g.,\ annual temperature curves generated from daily data, intraday stock prices, etc.). In this case it needs to be clarified how the residual noise is to be interpreted.  The translation of our problem into factor model language gives a mathematical/statistical meaning to the noise $U_t$ which goes beyond measurement errors. The $U_{t}$ define the \emph{ideosyncratic components} of $Y_{t}$, which are characterised by being uncorrelated or, more generally, be weakly correlated in a certain sense to be specified. The components of $U_t$ represent ``unsystematic'' fluctuations in our functional trajectories. 
\end{remark}

\subsection{Estimation approach}\label{s:generalapp}
 As mentioned above, the signal $X_t(\bs)$ is related to the common components of $Y_t$. The core idea of the algorithm that we pursue is simple and can be summarised as follows:\medskip

\noindent
\underline{{\bf Core algorithm:}}\\
1. Estimate $\mu(\bs)$ by $\hat\mu(\bs)=\frac{1}{T}(Y_1+\cdots+Y_T)$.\\
2. Center the data by $\hat\mu(\bs)$.\\
3. Choose an appropriate order $L$.\\
4. Approximate $X_t(\bs)-\mu(\bs)$ through the estimated \emph{common components}: $\hat{B} \hat F_t$.\\
5. Set $\hat X_t(\bs)=\hat\mu(\bs)+\hat B\hat F_t$.\medskip

Steps 3.\ and 4.\ can be carried out by many existing approaches for factor models. 
 \citet{baing2002} is a key reference for determining the dimension $L$. \citet{hallinliska2007} expanded the approach to dynamic factor models. \citet{onatski2010} proposes an approach that uses the empirical distribution (ED) of the eigenvalues of the sample covariance matrix. \citet{owenwang2016} use a Bi-Cross-Validation (BCV) technique to estimate the number of factors. Contrary to other approaches, they're not specifically interested in recovering the true number of factors, but rather the number of factors best-suited to recover the underlying signal.  In the process of our empirical work, we have investigated the behaviour of these estimators. We found that the BCV  and ED approaches work best in our FDA context. In Section~\ref{s:L} we will propose an empirical method to choose $L$.
 
Once $L$ is fixed, there are two main approaches for factor model estimation. One strategy is to utilize principal component analysis, e.g.,\ \citet{chamberlain1983} use this method.  PCA is particularly simple to implement and doesn't require numerically intense stochastic optimization methods. 
\citet{bai2003} investigated the asymptotic behaviour of both the factors as well as the factor loadings and---under technical conditions---proved consistency as well as robustness to mild correlation in the error terms.

The second popular strand is based on maximum likelihood. \citet{choi2012} expanded upon previous ideas by describing an efficient estimation for factor models, where the conditional distribution of $U_t \vert F_1, \ldots F_T$ is assumed to be normal with a covariance matrix that is not necessarily diagonal. \citet{baili2012, baili2016} provide a  method involving a quasi-maximum-likelihood approach. 

Let us discuss the PCA approach, which can be motivated as follows. Let $Y=(Y_1,\ldots, Y_T)$ and define  $U=(U_1,\ldots, U_T)$ and $F^\prime=(F_1,\ldots, F_T)$. Then we can write our model equation \eqref{e:factmod} in the compact matrix form
\begin{equation}\label{e:repmat}
Y=BF^\prime +U.
\end{equation}
In this notation, the objective is to estimate $BF^\prime$ through some estimator $\hat B\hat F^\prime$.
Suppose that $F$ is already known. Then 
$
Y^j=Fb_j+U^j,
$
which leads to the common least-squares estimator $\hat b_j^\text{LS}=(F^\prime F)^{-1}F^\prime Y^j.$ 
Here $b_j^\prime$ is the $j$-th row of $B$ and $Y^{j}$ and $U^{j}$ denote the $j$-th column of $Y^\pr$ and $U^\pr$, respectively. 
If our data are independent (or satisfy some appropriate weak dependence condition), it holds by the law of large numbers and orthogonality of principal components scores that
\begin{equation}\label{e:LLN}
\frac{1}{T}F^\prime F \stackrel{P}{\to} I_L.
\end{equation} 
This motivates 
$\hat B_{|F}:=\frac{1}{T} YF$ as estimator for $B$ conditional on $F$.
For $F$ in turn we use the empirical principal components and set $\hat F=\sqrt{T} \hat E$,
where $\hat E=(\hat e_1,\ldots, \hat e_L)$ are the eigenvectors of $\frac{1}{T}Y^\prime Y$  ($T\times T$) associated to the~$L$ largest eigenvalues $\hat\gamma_1\geq \ldots\geq \hat\gamma_L$. 
%The $j$-th row of $F^\prime$ is denoted by $F_j$.
 Then 
$\frac{1}{T}\hat F^\prime \hat F = I_L$. In summary
$
\hat F=\sqrt{T}\hat E\quad \text{and}\quad \hat B=\frac{1}{T}Y\hat F,
$
which implies that 
\begin{equation}\label{e:pcaapproach}
(\hat X_1(\bs),\ldots, \hat X_T(\bs))=\widehat{BF^\prime}:=\hat B\hat F^\prime=Y\hat E\hat E^\prime.
\end{equation}
We have analysed this estimator for the signal in \cite{hormann:jammoul:2021} and have shown that under mild technical conditions (see Assumptions~\ref{a:noise}-\ref{a:pcs} in the Appendix) this estimator converges uniformly, i.e.
$$
\sup_{1\leq t\leq T}\sup_{1\leq i\leq p}|X_t(s_i)-\hat X_t(s_i)|\to 0
$$
in probability and explicit convergence rates can be obtained.

In applications the user is free to choose any estimation method that leads to satisfactory and plausible results. 
(See Section~\ref{s:fit}.)

\begin{remark}\label{r:1}
A well known problem in factor model theory is that factor loadings and the factor scores are not unique. If $O\in\mathbb{R}^{L \times L}$ is some orthogonal matrix, then
$
BF=(B O)(O^\pr F),
$
and $\mathrm{Var}(O^\pr F)=I_L$. This identification issue is not a problem here, because we are primarily interested in the common components $BF$, which remain well identified.
\end{remark}

\subsection{Estimating the full curve}\label{s:full}
The factor approach doesn't return a full curve, but an estimate of the noise-free curves at the points $0\leq s_1< s_2<\ldots<s_p\leq 1$.  If the goal is to work with full curves, then it is up to the experimenter to choose a discrete-to-function transformation which is designated for noise-free data.  The simplest approach, namely linear interpolation, will be considered in this section.   
\begin{theorem}\label{thm:cont}
Let $\hat X_t(s_i)$ be the factor model estimates for the underlying signal at the points  $0\leq s_1< s_2<\ldots<s_p\leq 1$ as defined in \eqref{e:pcaapproach} and let $\hat X_t(s)$, $s\in [0,1]$, be the interpolation of these estimates $\hat{X}_t(s_{i})$. Denote $\delta=\max_{1\leq i \leq p-1} |s_{i+1}-s_i|$. Assume that for some $\alpha\in (0,1 \rbrack$ we have a random variable $M_t$ such that
\begin{equation}\label{e:lip}
| X_t(s) - X_t(u) | \leq M_t  | s - u |^\alpha
\end{equation} holds, where $EM_t=m<\infty$. Then under Assumptions~\ref{a:noise}--\ref{a:pcs} in the Appendix we have
\begin{equation}\label{e:liprate}
\sup_{s \in \lbrack 0,1 \rbrack} | X_t(s) - \hat{X}_t(s) | = O_P\left( \frac{1}{T^{1/4}} + \frac{T^{1/4}}{\sqrt{p}} + \delta^\alpha \right).
\end{equation}
\end{theorem}

In order to extend our results to the full sample paths we require the Lipschitz condition~\eqref{e:lip}, which has for example been previously considered in \citet[p.169]{bosq:2000} or in \citet{kallenberg:2002}. Prominent examples of processes that fulfill this property include the Brownian and fractional Brownian Motion, hence also processes which are by no means smooth. Note that under these assumptions, the observation points need not be equidistant in order to control size of the modulus of continuity, but merely the largest distance between two knots needs to become small. It is, however, natural to assume that $\delta = O(p^{-1})$ holds and according rates may be easily derived from \eqref{e:liprate}.

\subsection{Estimation of eigenfunctions}\label{s:eigenfun}

Functional principal components take a central role in FDA literature, see for example \citet{ramsaysilverman05}. When data are fully observed, the estimation theory is well established \citep{kleffe:1973, dauxoisetal:1982, hallhosseininasab:2006}. When data are discretely observed and subject to measurement errors, then obviously estimation theory has to be adapted. The most common strategy is to first estimate the curves using techniques described in the introduction and then to estimate principal components from the empirical covariance operator of the fitted data. Alternatively, one may use eigenfunctions of the non-parametric estimates of the covariance kernel as suggested in \cite{Staniswalis:Lee:1998} or \cite{yaoetal:2005}.

We now show that functional principal components can be estimated quite well from discretely observed and noisy data. \emph{Unlike the procedures mentioned before, this does not involve a smoothing step. }
 Let us begin by noting that Lemma~1 in \citet{hormann:jammoul:2021} shows that under some mild technical assumptions the $\ell$-th eigenvalue $\lambda_\ell$ of the covariance kernel $\Gamma^X(s,s^\prime) = \text{Cov}(X_t(s), X_t(s^\prime))$ may be consistently estimated by $\hat\gamma_\ell^Y / p$, which denotes $p$-th fraction of the $\ell$-th eigenvalue of the raw covariance matrix $T^{-1}Y^\prime Y$ or equivalently of $\hat\Sigma^Y:=T^{-1}YY^\prime$ ($\in \mathbb{R}^{p\times p}$). A similar result has been obtained in \citet{Benko:Haerdle:Kneip:2009}. These authors also work with the raw data when estimating the eigenvalues. For estimation of eigenfunctions they do, however, use a smoothing step. 
To formulate our result, we  denote the eigenvectors associated to the eigenvalues $\hat{\gamma}_\ell^Y$ by $\hat{\psi}_\ell^Y$. In order to properly describe the relationship between the function $\varphi_\ell$ and the vector $\hat\psi_\ell^Y$ we define  $\tilde\varphi_\ell(s)=\sqrt{p}[\hat\psi_\ell^Y]_i$ if $s\in [s_i,s_{i+1})$, where $[v]_i$ denotes the $i$-th component of a vector $v$. The step-function $\tilde{\varphi}_\ell$ is the proposed estimator for the eigenfunction $\varphi_\ell$. Note that the scaling ensures that $\|\tilde\varphi_\ell\|^2:=\int_0^1\tilde\varphi_\ell^2(s)=1$.

\begin{remark}
Eigenfunctions and eigenvectors are of course uniquely defined only up to the sign. In order to ensure that $\tilde{\varphi}_\ell$ indeed is the estimate for $\varphi_\ell$, we assume that $\langle \varphi_\ell, \tilde{\varphi}_\ell \rangle \geq 0$ holds. To lighten the notation, we henceforth assume in the proof of Theorem~\ref{thm:eigenfun} that the inner product is nonnegative for any pair of eigenfunctions and eigenvectors whose difference is being investigated.
\end{remark}

\begin{theorem}\label{thm:eigenfun}
Let Assumptions \ref{a:noise} and \ref{a:signal}~(a) and (b) hold. Assume that the sampling points $s_i$ are equidistant and that  
\begin{equation}\label{a:cont}
\sup_{s\in\lbrack0,1\rbrack} \text{E}| X(s+h) - X(s)|^2 = O(h)\quad(h\to 0).\end{equation} 
Then if $\alpha_\ell=\min\{\lambda_\ell-\lambda_{\ell+1},\lambda_{\ell-1}-\lambda_\ell\}\neq 0$, we have
$$\|\varphi_\ell-\tilde\varphi_\ell\| = O_P\left(\frac{1}{\sqrt{p}} + \frac{1}{\sqrt{T}}\right),\quad\ell \geq 1.$$
\end{theorem}

\citet{Benko:Haerdle:Kneip:2009} have compared their eigenfunction estimators from discretely observed and noisy data to the empirical eigenfunctions $\hat\varphi_\ell$ from fully observed data. They show that the error is of smaller order of magnitude than the error between $\hat\varphi_\ell$ and $\varphi_\ell$. Since their result is pointwise in $s$, it is not directly comparable to our $L^2$ distance. From a technical point of view both results have advantages and disadvantages. Our result holds under milder smoothness conditions. We merely need Assumption~\eqref{a:cont}, while they request second order derivatives with a uniformly bounded fourth order moment. Furthermore, we allow for dependence in both the errors and the observations. \citet{Benko:Haerdle:Kneip:2009} focus on the i.i.d.\ setup. On the other hand, they allow for more general errors with 8 moments and don't request a regular sampling design.

\section{Model diagnostics}\label{s:fit}
A simple diagnostic tool which may help to discern inadequate signal extraction is the inspection of the covariance of the residuals. Consider the fits $\hat X_1(\bs),\ldots, \hat X_T(\bs)$ and denote by $\hat U_t=Y_t-\hat X_t(\bs)$ the residual vectors.  Each residual vector $\hat U_t$ defines a time series $\hat U_{t1},\ldots \hat U_{tp}$. For example, if $(U_{ti}\colon 1\leq i\leq p)$ is assumed to be white noise, then this should be reflected in the empirical autocorrelation functions (acf's) 
\begin{equation}\label{acfcheck}
\hat\gamma_{\hat U_t}(h)=\frac{1}{p}\sum_{i=1}^{p-|h|}(\hat U_{t,i+h}-\bar{\hat U}_{t\cdot})(\hat U_{ti}-\bar{\hat U}_{t\cdot}).
\end{equation}  Since we have replicates, we may also conclude that
\begin{equation}\label{diagcov}
\hat\Gamma^{\hat U}:=\frac{1}{T}\sum_{t=1}^T(\hat U_t-\bar{\hat U})(\hat U_t-\bar{\hat U})^\prime\approx\mathrm{Var}(U_t),
\end{equation}
where $\bar{\hat U}$ is the grand mean of $\hat U_1,\ldots, \hat U_T$.
If there is doubt that the noise components are stationary (e.g.,\ if the homogeneous variance assumption is likely to be violated) analysing $\hat\Gamma^{\hat U}$ may be preferable over investigating the acf's $\hat\gamma_{\hat U_t}$. 
If the residual covariances are not conform with the assumptions on the noise variables (e.g.\ iid noise), this indicates that either these assumptions were incorrect, or that the transformation from discrete to functional data introduced some bias. 

In our real data examples (Section~\ref{s:real}) we investigate daily mean temperatures and corresponding annual temperature curves from Canada. Following \citet{ramsay09}, the daily data were transformed with 65 basis functions and a roughness penalty to annual curves. 
In Figure~\ref{fig:stmargintro} we show the acf's \eqref{acfcheck} of the residual vectors of this penalized B-spline approach at a weather station in St.~Margaret's Bay, Nova Scotia, in the year 1993. We also show the  heat map representing the righthand side in \eqref{diagcov}. For better visibility, the heat map is restricted to the first 2 months of the year. 
 \begin{figure}[!ht]
\begin{center}
\includegraphics[width=7.5cm]{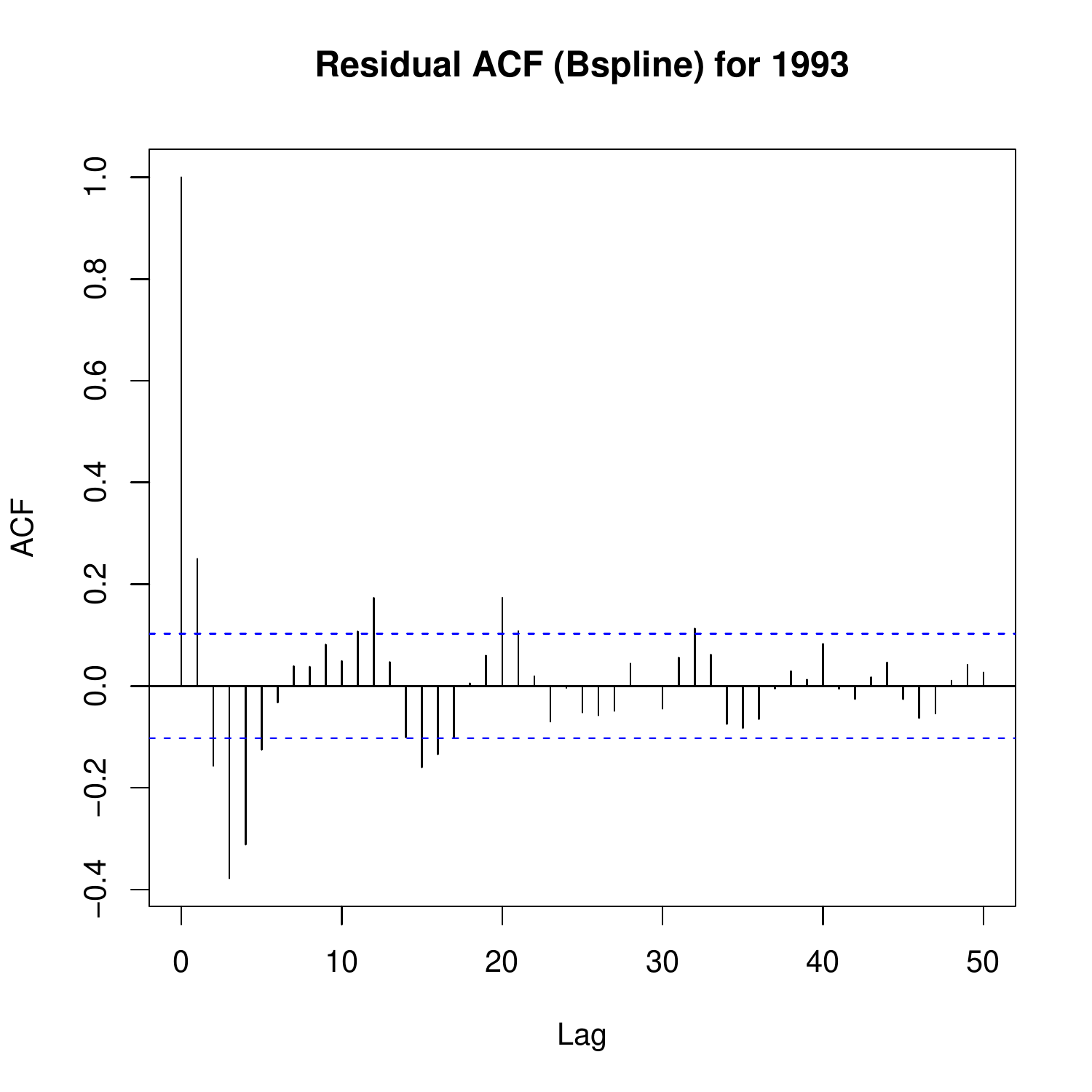}
\includegraphics[width=7.5cm]{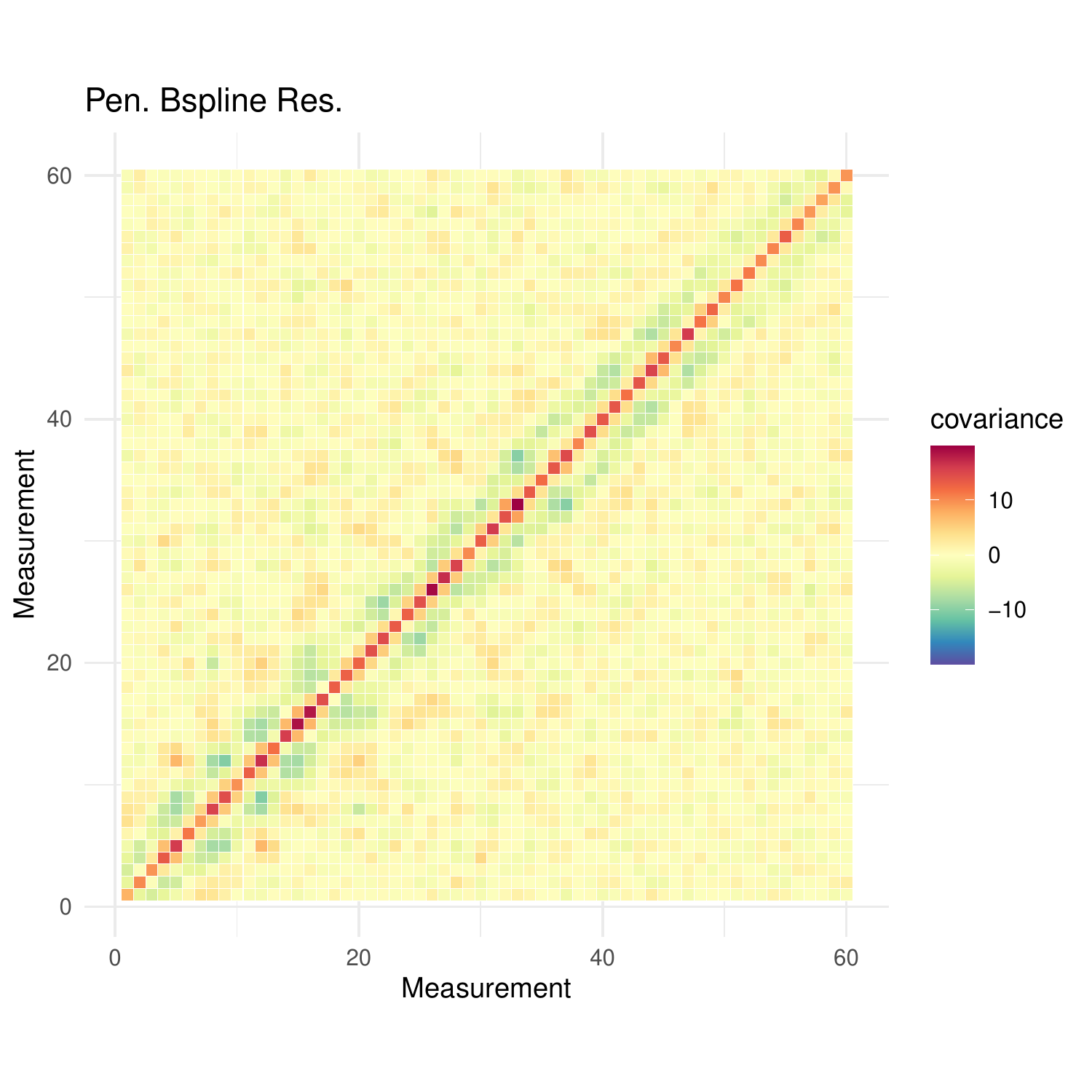}
\caption{Autocorrelation function of residual vectors for Penalized B-Splines in the \textit{St. Margaret's Bay Data} example for the year 1993 and the heat map corresponding to the lefthand side in \eqref{diagcov}.}
\label{fig:stmargintro}
\end{center}
\end{figure}
Details on the data and the implementation will be given in Section~\ref{s:real}. At this stage, we want to draw the readers attention to the  spurious oszillation in the acf. If the components of the error vectors were iid---as it is commonly assumed---then the acf should be zero for all lags $\neq 0$. In a slightly more realistic setting we would expect some moderate positive correlation of the errors, which tapers to zero with increasing lag. 

\subsection{Testing for independent errors}\label{s:test}
\newcommand{\ii}{\mathrm{i}}

Within the FDA literature the iid assumption for the error components $U_{t1},\ldots, U_{tp}$ is strongly prevailing. Below we refer to this assumption as the null hypothesis $\mathcal{H}_0$. Surprisingly, however, on real data this assumption is typically used without providing empirical evidence. To the best of our knowledge, no specific  statistical tests have been developed for this problem. Of course, a straight forward strategy is to employ some of the existing  white noise tests individually to each residual vector and then to aggregate the information from the resulting $T$ tests. Below we propose a tailor-made test statistic for our setting. 
To this end we introduce some further notation. We assume throughout that the error vectors $U_t$, $1\leq t\leq T$, are iid. The components will also be iid or stationary, depending on whether we operate under the null hypothesis or the alternative hypothesis.

For some generic random vector $Z=(Z_1,\ldots, Z_p)^\prime$  we denote the empirical variance of the components of $Z$ by $S_Z^2=\frac{1}{p-1}\sum_{k=1}^p (Z_k-\bar Z)^2$ . The periodogram is defined as $$I_{Z}(\theta)=\frac{1}{p}\left|\sum_{k=1}^p Z_ke^{-\ii k\theta}\right|^2.$$ Here $\ii=\sqrt{-1}$ and $|z|$ is the modulus of a complex number $z$. We refer to the frequencies $\theta_\ell=\frac{2\pi \ell}{q}$, $1\leq \ell\leq q:=\lfloor p/2\rfloor$ as the fundamental frequencies.  

Now we choose a subset of fundamental frequencies $\bm \theta=\{\theta_\ell$, $\ell \in \mathcal{F} \subset \lbrace 1, \ldots, q\rbrace\}$ and denote $f := \vert \mathcal{F} \vert$. We allow $\mathcal{F}$ (and hence $f$) and $p$ to depend on $T$ and our asymptotic statements below are then for $T\to\infty$. 

Set $\xi=T^{-1} \sum_{t=1}^T I_{U_t}(\bm\theta)$ and note that $\xi$ is an estimator of the spectral density of the $U_t$ at the fundamental frequencies contained in $\bm\theta$. If the components $U_{ti}$ are iid, then the spectral density is constant and the components of $\xi$ will be roughly constant as well. Our test statistic is thus based on the empirical variance $S_\xi^2$, which under $\mathcal{H}_0$ shall be accordingly small. Proposition~\ref{p:test} below establishes the essential asymptotic result related to the proposed test under the null.

\begin{prop}\label{p:test}  Assume that $\mathcal{H}_0$ holds and that $EU_{11}=0$, $EU_{11}^2=\sigma^2$ and $EU_{11}^4<\infty$. 
Let $\hat\sigma^2$ denote a consistent estimator of the variance.
Then $$\Lambda_\emph{fin}:=(f-1)TS_\xi^2/\hat\sigma^4\stackrel{\mathrm{d}}{\to} \chi^2_{f-1}.$$
If additionally    $EU_{11}^8<\infty$ and $f\to\infty$, and $f/T\to 0$, then
$$
\Lambda_\emph{inf}:=\left(TS_\xi^2/\hat\sigma^4 - 1\right)\sqrt{(f-1)/2}\stackrel{\mathrm{d}}{\to} N(0,1).
$$
\end{prop} 
In practice the $U_{ti}$ are latent and the test will be applied to the residuals $\hat U_{ti}=Y_{ti}-\hat X_t(s_i)$. This gives then rise to the test statistics $\hat\Lambda_\text{fin}$ and $\hat\Lambda_\text{inf}$. 
 If $p$ is of the same order or of a bigger order of magnitude than $T$ (e.g.\ this is the case in our real data in Section~\ref{s:real}), then taking $\mathcal{F}=\{1,\ldots, q\}$ is not theoretically justified by the proposition, since then $f\approx p/2$ and hence $f/T\not\to 0$. We can overcome this problem by thinning out the frequencies $\mathcal{F}$, i.e.\ we choose some large enough $m$ and only take every $m$-th frequency. Then $f\approx\frac{p}{2mT}$.

In our next result we want to show the proposed test is consistent under the following alternative:
\begin{assumption}\label{a:consistencytest}[Alternative Hyptothesis]
We assume that the process $U:=\lbrace U_{ti}: i\geq 1\rbrace$ is stationary with absolutely summable autocovariance $\gamma_U$ function and spectral density 
$$g(\theta):=\sum_{h\in\mathbb{Z}}\gamma_U(h)e^{ih\theta}.$$ 
Additionally we assume that
$\text{Var}(I_{U_1}(\theta_\ell))$ is uniformly bounded for all $\ell\in\mathcal{F}$ and all dimensions $p$.  Finally, denoting $\bar{g}=\frac{1}{f}\sum_{\ell\in\mathcal{F}}g_U(\theta_\ell)$ we assume that there is some $\delta>0$ such that
 \begin{equation}\label{e:spectnonconst}
\frac{1}{f-1} \sum_{\ell \in\mathcal{F}} (g(\theta_{\ell}) - \bar{g})^2 > \delta.
  \end{equation} 
  When $f$ diverges, \eqref{e:spectnonconst} should hold uniformly in $f$.
\end{assumption}

Besides mild technical moment assumptions (which hold e.g.\ for certain linear processes), our basic requirement under the alternative is that the noise is correlated and hence that the spectral density is not constant. In order to detect such a non-constant spectral density, we have to assure that it varies at the frequencies we have incorporated in our test statistic. This is assured by \eqref{e:spectnonconst}.  Note that the term in \eqref{e:spectnonconst} does not just depend on $f$ but also on the choice of frequencies.
If we select the frequencies $\theta_\ell$ on a regular grid and $f\to\infty$, then we can replace our condition by $\int_0^\pi \left( g(\theta)-\int_0^\pi g(s)ds\right)^2 d\theta>\delta$.

\begin{prop}\label{p:alt}  Consider the setting of Proposition~\ref{p:test} and assume that \eqref{a:consistencytest} holds. Let $0\leq \psi_T=o(T)$. Then $\Lambda_\emph{fin}/\psi_T\to\infty\quad (T\to\infty).$
If additionally $f=f(T)\to\infty$, $f/T\to 0$ and if $0\leq \psi_T=o(T\sqrt{f})$ then
$
\Lambda_\emph{inf}/\psi_T\to\infty\quad (T\to\infty).$
\end{prop}

We can see from Proposition~\ref{p:alt} that our proposed test statistic is rather powerful. Hence, we need to account for the effect of the estimation error $\hat U_{ti}-U_{ti}=\hat X_t(s_i)-X_t(s_i)$. The estimated residuals $\hat U_{ti}$ will not be perfectly iid, even if the $U_{ti}$ are, no matter which estimator is applied. We have experienced in simulations that frequencies close to zero seem to get eliminated in the time series $\{\hat U_{ti}\colon 1\leq i\leq p\}$ and hence we do not accurately estimate the spectral densities at very low frequencies. We illustrate this in Figure~\ref{fig:perplot1}. The same phenomenon holds true for other estimation methods we have explored in this paper.  We overcome the problem in our real data example by excluding $\ell\in\mathcal{F}$ if $\ell<0.1\times q$.

\begin{figure}[ht]
\begin{center}
\includegraphics[width=7.5cm]{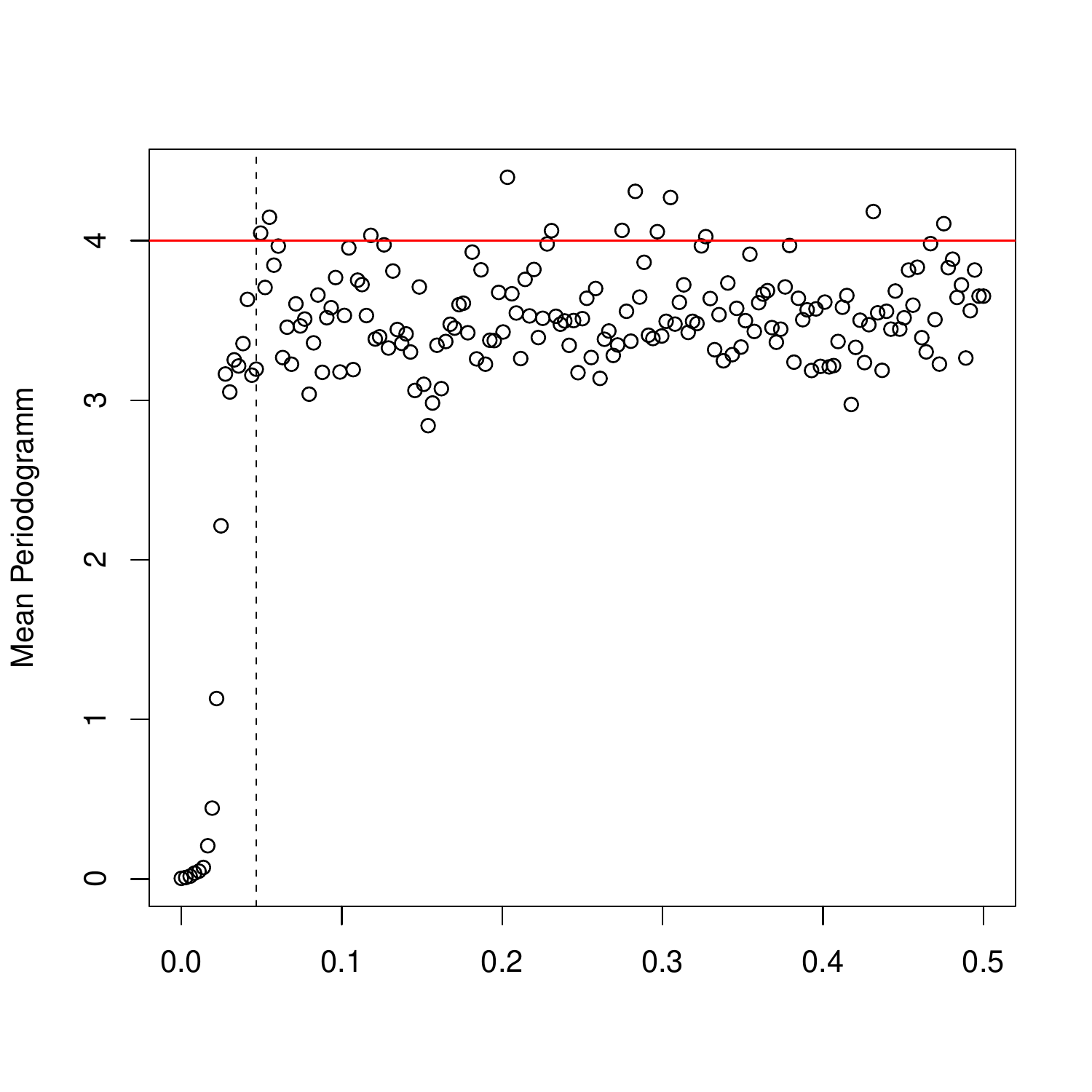}\\
\caption{Components of the averaged periodograms $\hat \xi=T^{-1} \sum_{t=1}^T I_{\hat U_t}(\bm\theta)$. We use the simulation setting of Section~\ref{s:smoothsim} with $U_{ti}\stackrel{\text{iid}}{\sim} N(0,4)$,  $p=365$ and $T=200$. Dotted line indicates a cutoff of the first $10\%$ of $q$ available frequencies.}
\label{fig:perplot1}
\end{center}
\end{figure}

\subsection{A variant of the scree plot}\label{s:L}
Determining the number of factors is a difficult problem. 
 As previously mentioned in Section~\ref{s:generalapp}, among the existing approaches the methods by \citet{onatski2010} and \citet{owenwang2016} were the most accurate in our context. In this section we would like to propose an empirical approach, which is a visual tool similar to the widely used scree plot \citep{cattell:1966}.  We recall that the classical scree plot is based on the eigenvalues $\hat\gamma_1,\hat\gamma_2,\ldots$ of the empirical covariance matrix $\frac{1}{T}YY^\prime$. It shows the eigenvalues in descending order. A kink (or an `elbow') in the graph, where the rate of descent becomes small, indicates the true number of factors.

Instead of eigenvalues we propose to use the values of our test-statistics $\hat\Lambda_\text{inf}$ established in Section~\ref{s:test}.  %To this end we remind the reader that \eqref{signoise} may be written in terms of the $L$-factor model \eqref{e:factmod}, where the underlying signal to be estimated refers to the common components and the error process is the idiosyncratic component. 
The logic behind is as follows: assume that the errors $(U_{ti}\colon 1\leq i\leq p)$ are iid and suppose we fit a factor model with $\ell < L $ factors. Then, a certain amount of cross-sectional dependence still prevails in the residuals $(\hat{U}_{ti}\colon 1\leq i\leq p)$, since the estimator does not yet fully account for the common component. Hence, underestimating $L$ is likely to result in a large value for~$\hat\Lambda_\text{inf}$. When increasing the number $\ell$ of factors included in the model, the cross-sectional dependence is expected to diminish and finally to drop to a baseline level, when $\ell$ surpasses the true $L$, i.e.\ when in principle we move from a dependent to an independent sequence. Since our estimators $\hat X_t(\bs)$ are robust to overestimation of $L$, we expect the test statistics to approximately remain constant for   $\ell \geq L$. 
The method can be theoretically justified if the noise variables are iid, e.g.\ when we know that the noise can be related to measurement errors. In practice we may use it in a more general context. There we move from a long-range type dependence to weak dependence, which is likely to be reflected by a corresponding change in the decay rate of the test values.

We illustrate this approach in Figure~\ref{fig:screeplot1}, where we show plots of $\hat\Lambda_\text{inf}$ (figure on the left) and $\hat\gamma_\ell$ (figure on the right) against the chosen number of factors $\ell$. Details of the related data is again provided in Section~\ref{s:smoothsim}. We have chose $T=200$ and $p=365$. These parameters are comparable to our real data example in Section~\ref{s:real}. Due to the very large dimension, we are thinning out frequencies with $m=3$ and as suggested in Section~\ref{s:test} and we also drop 10$\%$ of the lowest frequencies, so that  $f/T\approx 0.27$.  In this example the true number of factors is $L=21$ (marked by the dashed vertical line). Our variant of the scree plot remains approximately constant after $\ell=20$. This would then be our suggested estimator $\hat L$. For the standard scree plot the rule is to choose the position of the kink as the number of factors, which here is at $\ell=8$.

\begin{figure}[ht]
\begin{center}
\includegraphics[width=13.5cm]{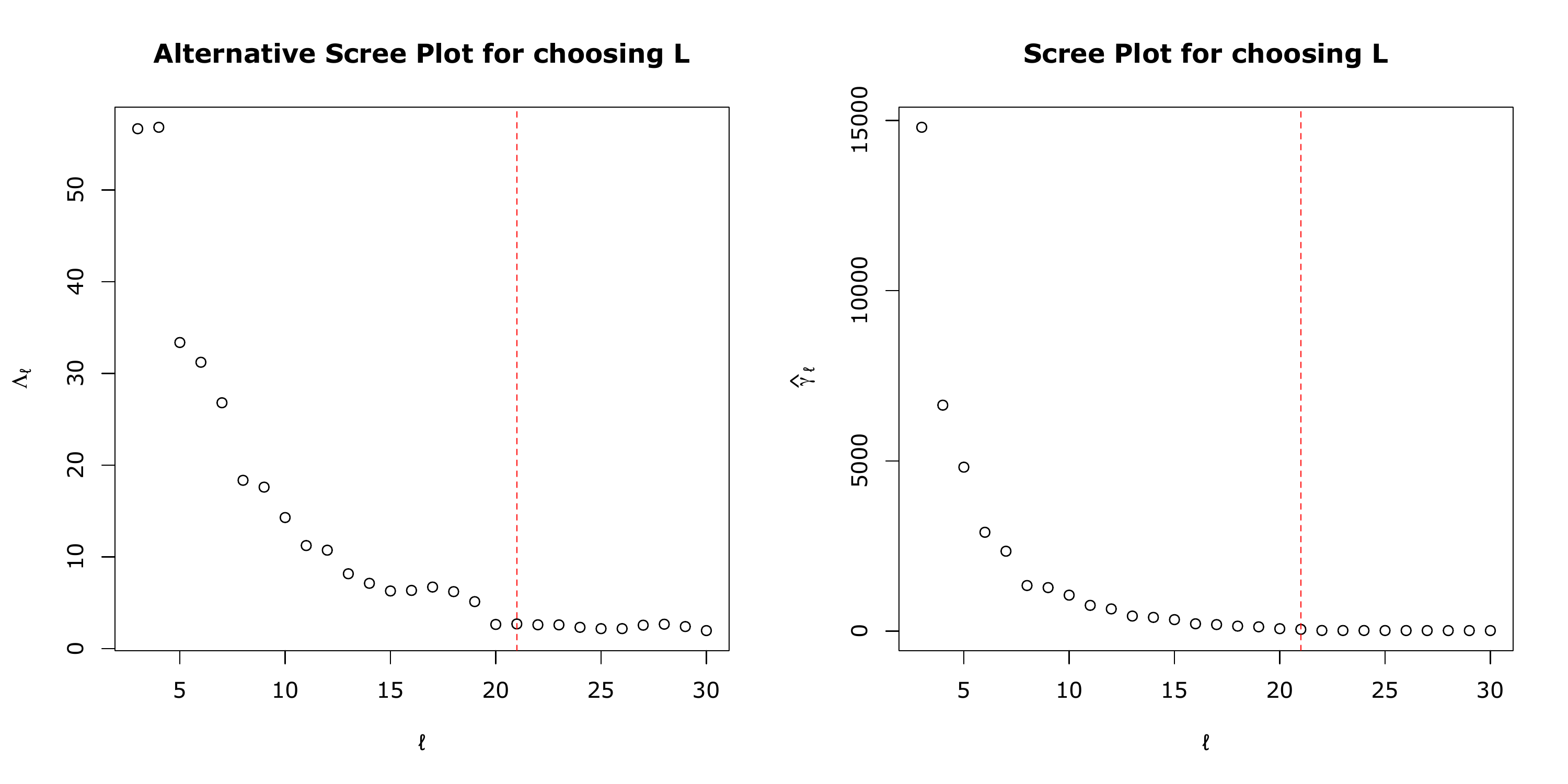}\\
\caption{Scree plots for choosing $L$ in a simulation setting of Section~\ref{s:smoothsim}. We use  $U_{ti}\stackrel{\text{iid}}{\sim} N(0,4)$,  $p=365$ and $T=200$. The true number of factors $L=21$ is indicated by the dotted red line.}
\label{fig:screeplot1}
\end{center}
\end{figure}

In Figure~\ref{fig:screeplot2} below we consider another setting, where $p=48$ is relatively small compared to $T=500$. In this example we simply use $\mathcal{F}=\{1,\ldots, q\}$. Since we get huge values for small $\ell$ we plot $\log \hat\Lambda_\text{inf}^2$.
We can see that the `scree' in our approach is much steeper and levels off near the true value of $L$.
For the standard eigenvalues-based scree plot no accentuated kink can be spotted at $L=21$. According to the `elbow-rule', we would again chose $L=8$. 

\begin{remark}
The proposed method is based on the assumption of independent noise.  An important message of our paper is that in several real data examples the errors are not necessarily related to measurement errors and a certain degree of dependence is well expected. This is also the case for the data we consider in  Section~\ref{s:real}. A modification of the approach which allows for weakly dependent errors would be interesting, but is out of the scope of this paper and will be subject of future research.
\end{remark}

\begin{figure}[ht]
\begin{center}
\includegraphics[width=13.5cm]{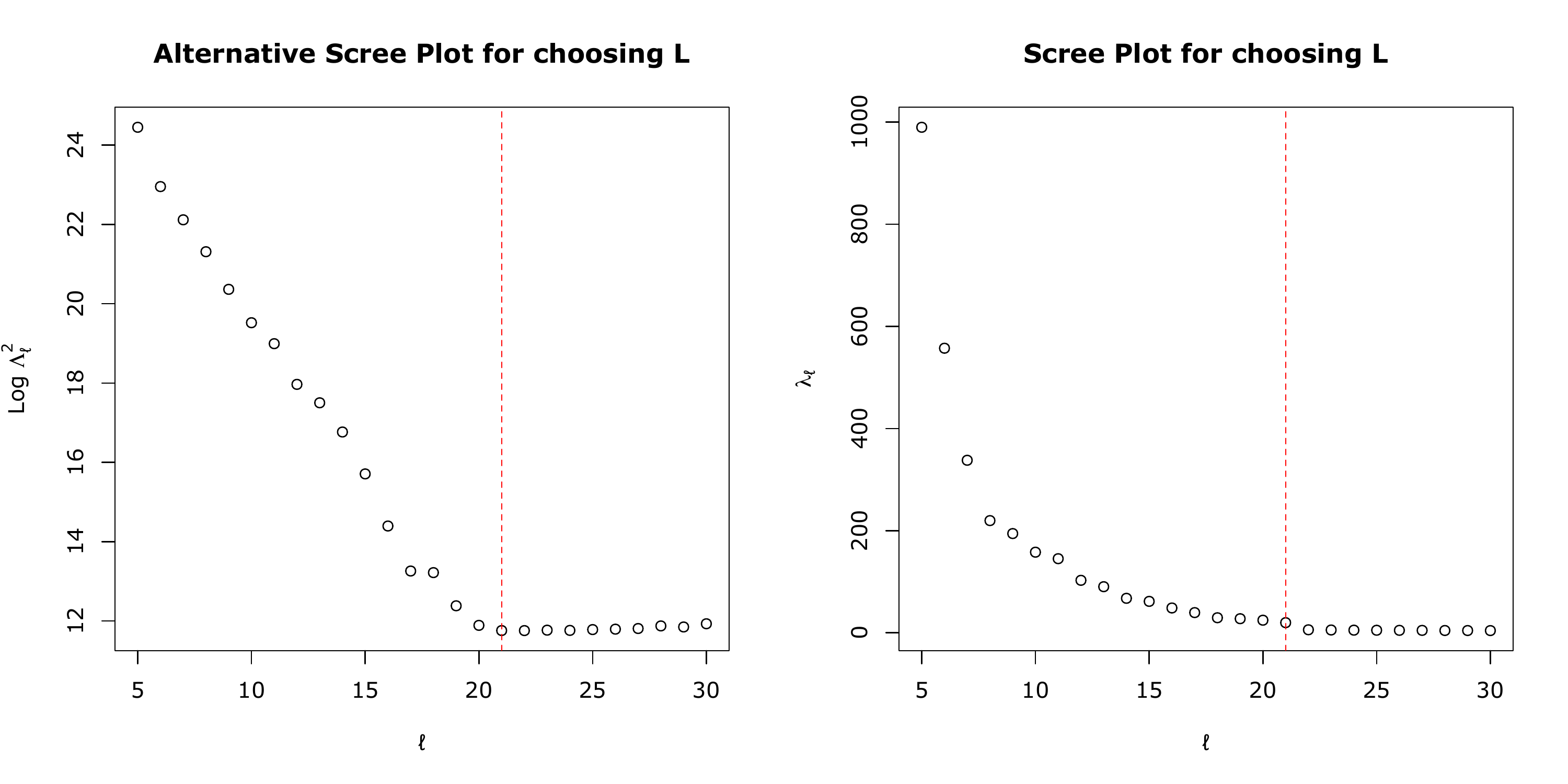}
\caption{Scree plots for choosing $L$ in a simulation setting of Section~\ref{s:smoothsim}. We use $U_{ti}\stackrel{\text{iid}}{\sim} N(0,4)$, $p=48$ and $T=500$. The true number of factors $L=21$ is indicated by the dotted red line.}
\label{fig:screeplot2}
\end{center}
\end{figure}

\section{Simulation experiments}\label{s:sim}

In this section we investigate the performance of our methods on simulated data examples. We have performed extensive simulation studies that can be separated into two types: smooth data (Section \ref{s:smoothsim}) and data where the underlying signal and its derivative contains discontinuities (Section \ref{s:roughsim}). The following simulations were performed in \texttt{R version 4.0.3}. 

\subsection{Recovering smooth signals}\label{s:smoothsim}

In order to have a realistic data generating process, we chose to adapt some real data for our purposes. We consider bi-hourly measurements of particulate matter \texttt{PM10} in Graz from October 1st 2010 to March 31st 2011. Thus, we have 48 observations per day over the course of 182 days. To have control over the actual structure of the data, we generated synthetic curves by the following four steps:
(1) transform the raw data to functional data;
(2) create a bootstrap sample of size $T$;
(3) evaluate the resulting sample on a grid of intraday time points;
(4) add noise as in \eqref{signoise}.
In Step~(1) we chose to do a least squares fit using $21$ cubic B-splines. Then $T = 50, 100, 200, 500$ curves were obtained by the bootstrapping in Step~(2). These curves are considered as our signals $X_1(s),\ldots, X_T(s)$.  The signals in turn were evaluated at $p = 24, 48, 96$ equidistant points in $\lbrack 0, 1 \rbrack$, giving rise to $X_t(\bs)$ (Step~(3)). In the final step we generated $Y_t=X_t(\bs)+U_t$ with $U_t\stackrel{\text{i.i.d.}}{\sim} N_p(0,\Omega)$. For $\Omega$ we chose the covariance of a sample $(\varepsilon_1,\ldots, \varepsilon_p)^\prime$ from the stationary AR(1) process
$
\varepsilon_k=\theta \varepsilon_{k-1}+\xi_k,
$
where $(\xi_k)$ is white noise with zero mean and variance $\sigma^2$. Altogether we show four parameter settings (S1--S4). 
These are listed in Table~\ref{tab:param}.
 We remark that under settings S1 and S2 the noise components are i.i.d.\ $N(0,\sigma^2)$. In S3 and S4 we take into account potential autocorrelation in the intraday noise, which we believe is a realistic assumption in many applications.
\begin{table}[ht]
\begin{center}
\footnotesize
\begin{tabular}{c|c|ccc}
& & & $\sigma$  &\\
\hline
& &1& 2& 4\\
\hline
& 0 &   & S1 & S2\\
$\theta$ & 0.4& S3  &   & \\
&0.8 & S4 &   &  \\
\end{tabular}
\caption{Parameter settings for the noise.}
\label{tab:param}
\end{center}
\end{table}

Our goal is now to recover the $X_t(\bs)$. First we compare our proposed method with a B-spline (B) and penalized B-spline smoothing approach (PB). Since the actual signal in this simulation setting is already contained in a space spanned by B-splines, we consider in fact a setup which is favourable for these competitors. 
The B-spline smooth was computed using $p/3$ basis functions and methods from the \texttt{fda} package in \texttt{R}. When $p=48$ this yields a number which is comparable to the actual number of B-splines used to create the signal, otherwise it is bigger. This is in line with \citet{wood2017}, who suggests using more basis functions than one believes necessary and then using a penalization approach to smooth the result. We use penalized B-splines with a roughness penalty of the form $\int X''(s)^2ds$. The penalty is added to the regular least squares equation and weighted with a parameter $\lambda$, which needs to be chosen. This has been done using a GCV (generalized cross validation) technique as described in \citet{ramsay09}. 

Furthermore, we compare our approach to the functional principal components (FPC) approach as motivated in \citet{Staniswalis:Lee:1998}. To this end, we have used the function \texttt{fpca.sc} from the \texttt{refund} package in \texttt{R}, which smooths the empirical covariance prior to obtaining an estimate for the functional scores and subsequently, the estimated signal. The number of principal components was automatically chosen to be large enough to explain $99 \%$ of the variance. Note that in this approach, the smoothing of the covariance operator is done via penalized splines, which is in line with a suggestion in \citet{Di:etal:2009}. The number of basis functions we used in this smoothing is $p/3$ as well. In our exploration we found that increasing the number of splines in this function requires immense computational effort while giving little improvement. 

For the factor analysis, we used two different approaches. First, we used the PCA driven approach, as described in \citet{fanetal2013} and explained in our Section~\ref{s:generalapp} (PCA). We used the package \texttt{POET} in \texttt{R} in order to obtain the described estimates for the factor scores and loadings. Second, we use a Maximum-Likelihood approach (ML) with the EM algorithm as described in \citet{baili2012} and implemented in the package \texttt{cate}. As for choosing the number of factors, we used the methods BCV and ED, which are described in Section~\ref{s:generalapp}. We note that the method we proposed in Section~\ref{s:L} provides a powerful visual tool, but choosing $L$ in this way for hundreds of simulation runs is not practically feasible. 

Implementation of BCV and ED can also be found in the package \texttt{cate}. 
Note that for the implementation, a maximum number of factors \texttt{rmax} to be considered can be selected. We have chosen \texttt{rmax} = 23.
Estimates tend to be robust to the overestimation of the dimension $L$, but sensitive to too small $L$, see for example \citet{fanetal2013}. This is intuitive, as a too small choice of $L$ will result in important information being excluded from the fit, whereas we only add potentially ``insignificant" information if $L$ is chosen too large. Thus, we have used $\hat{L} = \max(\hat{L}_{\text{BCV}}, \hat{L}_{\text{ED}})$. Practically we experience that in most settings $\hat{L}_{\text{BCV}}\geq \hat{L}_{\text{ED}}$. Hence the results remain basically unchanged if $\hat{L}=\hat{L}_{\text{BCV}}$ is used. 

In order to evaluate the quality of the respective approaches we are interested in the error $X_t(\bs)-\hat X_t(\bs)$.
While for real data $X_t(\bs)$ is not observable, the signals are known in our simulation setting and we can hence 
define 
\begin{equation}\label{eq:sseapprox}
\text{SSE}^{\text{appr}} = \frac{1}{pT}\sum_{i=1}^p \sum_{t=1}^T (X_{t}(s_i) - \hat{X}_{t}(s_i))^2.
\end{equation}

The results of our Monte Carlo study with 250 iterations can be found in Tables~\ref{tab:tabiidtest} and ~\ref{tab:tabartest}. Methods that produce the minimal $\text{SSE}^{\text{appr}}$ in each instance are bold.

\begin{table}[ht]
\centering
\begingroup
\footnotesize
\begin{tabular}{cc|cccccc|cccccc}
  \toprule \multicolumn{2}{c}{Dimensions} & \multicolumn{1}{|c}{} & \multicolumn{5}{c}{$\text{SSE}^{\text{appr}}$ ($\sigma$ = 2)} & \multicolumn{1}{|c}{} & \multicolumn{5}{c}{$\text{SSE}^{\text{appr}}$ ($\sigma$ = 4)}\\$p$ & $T$ & $\hat{L}$ & B & PB & FPC & ML & PCA & $\hat{L}$ & B & PB & FPC & ML & PCA \\ 
  \hline
24 & 50 & 8 & 39.50 & 41.76 & 39.45 & 22.07 & \textbf{16.59} & 7 & 43.07 & 45.62 & 42.58 & 31.64 & \textbf{27.07} \\ 
  24 & 100 & 11 & 39.55 & 42.11 & 39.94 & 15.02 & \textbf{10.33} & 9 & 42.61 & 45.85 & 42.14 & 25.04 & \textbf{21.18} \\ 
  24 & 200 & 14 & 39.79 & 42.39 & 39.82 & 11.16 & \textbf{7.32} & 11 & 43.54 & 47 & 42.85 & 22.63 & \textbf{18.4} \\ 
  24 & 500 & 18 & 39.07 & 41.51 & 39.95 & 6.07 & \textbf{4.34} & 13 & 43.30 & 47.44 & 42.91 & 19.28 & \textbf{14.93} \\ 
     &  &  &  &  &  &  &  &  &  &  &  &  &  \\ 
48 & 50 & 12 & 6.87 & 6.92 & 14.32 & 7.09 & \textbf{5.51} & 10 & 10.68 & \textbf{10.66} & 16.95 & 15.34 & 14.12 \\ 
  48 & 100 & 21 & 6.74 & 6.75 & 14.40 & 3.02 & \textbf{2.62} & 13 & 10.83 & 10.9 & 17.21 & 11.70 & \textbf{10.82} \\ 
  48 & 200 & 21 & 6.76 & 6.77 & 13.91 & 2.14 & \textbf{2.1} & 16 & 10.77 & 10.91 & 16.85 & 9.82 & \textbf{9.08} \\ 
  48 & 500 & 22 & 6.77 & 6.79 & 14.06 & 1.99 & \textbf{1.96} & 21 & 10.80 & 11.01 & 16.82 & 8.05 & \textbf{7.87} \\ 
     &  &  &  &  &  &  &  &  &  &  &  &  &  \\ 
96 & 50 & 16 & 1.33 & \textbf{1.19} & 6.91 & 2.95 & 2.69 & 12 & 5.24 & \textbf{4.13} & 9.14 & 9.88 & 9.67 \\ 
  96 & 100 & 21 & 1.32 & \textbf{1.18} & 7.02 & 1.72 & 1.68 & 16 & 5.24 & \textbf{4.14} & 9.03 & 7.08 & 6.96 \\ 
  96 & 200 & 22 & 1.32 & \textbf{1.18} & 7.21 & 1.32 & 1.33 & 19 & 5.24 & \textbf{4.17} & 9.00 & 5.53 & 5.42 \\ 
  96 & 500 & 22 & 1.32 & 1.18 & 7.16 & 1.10 & \textbf{1.09} & 21 & 5.24 & \textbf{4.2} & 8.78 & 4.30 & 4.34 \\ 
   \bottomrule
  \end{tabular}
\endgroup
\caption{Simulation results ($\text{SSE}^{\text{appr}}$) for the synthetic \texttt{PM10} data with iid noise. Here $\hat L$ is median value of the estimates $\max(\hat{L}_{\text{BCV}}, \hat{L}_{\text{ED}})$.} 
\label{tab:tabiidtest}
\end{table}

\begin{table}[ht]
\centering
\begingroup
\footnotesize
\begin{tabular}{cc|cccccc|cccccc}
  \toprule \multicolumn{2}{c}{Dimensions} & \multicolumn{1}{|c}{} & \multicolumn{5}{c}{$\text{SSE}^{\text{appr}}$ ($\theta$ = 0.4)} & \multicolumn{1}{|c}{} & \multicolumn{5}{c}{$\text{SSE}^{\text{appr}}$ ($\theta$ = 0.8)}\\$p$ & $T$ & $\hat{L}$ & B & PB & FPC & ML & PCA & $\hat{L}$ & B & PB & FPC & ML & PCA \\ 
  \hline
24 & 50 & 8 & 38.50 & 40.46 & 38.99 & 20.7 & \textbf{14.61} & 8 & 39.51 & 41.56 & 39.95 & 21.07 & \textbf{15.15} \\ 
  24 & 100 & 12 & 38.63 & 40.71 & 39.22 & 12.37 & \textbf{7.34} & 12 & 39.71 & 41.78 & 40.07 & 13.44 & \textbf{8.51} \\ 
  24 & 200 & 18 & 38.42 & 40.49 & 39.14 & 4.52 & \textbf{2.88} & 18 & 40.35 & 42.17 & 40.63 & 5.76 & \textbf{4.19} \\ 
  24 & 500 & 19 & 38.88 & 40.69 & 39.62 & 2.23 & \textbf{1.63} & 19 & 40.09 & 41.66 & 40.51 & 4.31 & \textbf{3.47} \\ 
   &  &  &  &  &  &  &  &  &  &  &  &  &  \\ 
48 & 50 & 21 & 6.05 & 6.12 & 13.55 & 1.21 & \textbf{1.08} & 21 & 8.00 & 8.06 & 14.83 & 2.71 & \textbf{2.66} \\ 
  48 & 100 & 21 & 6.18 & 6.25 & 14.00 & 0.92 & \textbf{0.91} & 21 & 7.91 & 7.98 & 15.71 & 2.6 & \textbf{2.58} \\ 
  48 & 200 & 22 & 6.13 & 6.21 & 13.65 & 0.88 & \textbf{0.87} & 22 & 7.90 & 7.97 & 15.41 & 2.58 & \textbf{2.56} \\ 
  48 & 500 & 22 & 6.16 & 6.23 & 13.71 & 0.86 & \textbf{0.85} & 22 & 7.87 & 7.94 & 15.08 & 2.57 & \textbf{2.56} \\ 
   &  &  &  &  &  &  &  &  &  &  &  &  &  \\ 
96 & 50 & 21 & 0.70 & \textbf{0.69} & 6.55 & 0.89 & 0.84 & 21 & 2.43 & \textbf{2.42} & 7.73 & 2.58 & 2.5 \\ 
  96 & 100 & 22 & 0.70 & \textbf{0.69} & 6.58 & 0.71 & 0.69 & 22 & 2.42 & 2.42 & 7.94 & 2.4 & \textbf{2.36} \\ 
  96 & 200 & 22 & 0.70 & 0.69 & 6.79 & 0.64 & \textbf{0.62} & 23 & 2.42 & 2.41 & 7.93 & 2.37 & \textbf{2.32} \\ 
  96 & 500 & 22 & 0.70 & 0.69 & 6.65 & \textbf{0.56} & \textbf{0.56} & 23 & 2.43 & 2.42 & 8.03 & 2.35 & \textbf{2.29} \\ 
   \bottomrule
  \end{tabular}
\endgroup
\caption{Simulation results for the synthetic \texttt{PM10} data with AR(1) noise.} 
\label{tab:tabartest}
\end{table}

The most important observations are summarised below:
\begin{enumerate}
\item
The factor model approach outperforms the B-splines largely when $p$ is growing slower than $T$. For example, when $p=24$, $T=500$ then with setting S3 the $\text{SSE}^{\text{appr}}$ produced by the B-splines and penalized B-splines are  more than $20$ times bigger compared to the PCA factor model approach. The penalized B-splines work best if $p$ is very large and $T$ is small. In this case the noise can be very well smoothed on a local level.
\item
For the B-splines based approaches the $\text{SSE}^{\text{appr}}$ doesn't decrease with growing sample size only with increasing $p$. Against our expectations, the FPC method doesn't improve in practice with increasing $T$ either, though theoretically it should (see the results in \citet{Mueller:Stadtmueller:Yao:2006}). It seems that the eigenfunctions from the smoothed covariances are oversmoothing the data and then local features of the data cannot be accurately recovered. In contrast, for both factor model estimators $\text{SSE}^{\text{appr}}$ decreases significantly with increasing $T$ as well as increasing $p$.
\item The PCA approach gave better results than the MLE approach.
\end{enumerate}

We have also experimented with further simulations settings. Not surprisingly, by further increasing $\sigma$, $\text{SSE}^{\text{appr}}$ increases for all methods.
Nevertheless we observe that in comparison to each other the methods behave similarly as in the settings described. The combination large $\sigma$, large $p$ and very small $T$ (e.g.,\ $T=10$) favours our competitors, while our proposed approach improves considerably with growing $T$ in all instances. For only mildly larger $T$ and much larger $p$ (e.g., $p=96, 192$ and $T=30$) we immediately obtain estimates that are competitive with the other approaches. 

Since the signals in our simulations are relatively smooth, the good performance of smoothing methods is not surprising for large $p$. For curves with rough signal, smoothing approaches are not able to recover specific features of the signal due to oversmoothing. This phenomenon can be observed in the following section.

\subsection{Recovering signals with discontinuities}\label{s:roughsim}

 We check in the following simulation setting the practical impact of ``rough" signals on the respective methods. More specifically, the signals $X_t(s)$ are defined on $\lbrack 0,1 \rbrack$ and are constructed as follows:
$$X_t(s) = \sum_{k=1}^3 \xi_{tk} \varphi_k(s), $$
where $\varphi_1(s) = \mathds{1}_{ \lbrace s>1/3 \rbrace }$, $\varphi_2(s) = (-1)^{\kappa}4(0.2 - \vert s-0.5 \vert)\mathds{1}_{ \lbrace s \in \lbrack 1/3, 2/3 \rbrack \rbrace }$ and $\varphi_3(s) = \cos{6 \pi s}$, where $\kappa = \mathds{1}_{ \lbrace s \in ( 1/2, 2/3 \rbrack \rbrace }$. Here, $\mathds{1}_{ \lbrace s \in \lbrack a, b \rbrack \rbrace }$ denotes the characteristic function, that is $\mathds{1}_{ \lbrace s \in \lbrack a, b \rbrack \rbrace } = 1$ if $s \in \lbrack a, b \rbrack$ and $0$ otherwise. The associated scores are independent and normally distributed $\xi_{tk} \sim N(0, 2^{-2(k-1)})$ for $k=1,2,3$. The noisy observations are obtained via $y_{ti} = X_t(s_i) + u_{ti}$, where  $u_{ti} \stackrel{iid}{\sim} N(0, \sigma^2)$. We consider equidistant observation points $s_i = (i - 0.5)/p$ for $i=1, \ldots ,p$. 
Thus, the signals may be disrupted at $s=1/3$ (through $\varphi_1$), %and have a peak in $\lbrack 1/3, 2/3 \rbrack$ (through $\varphi_2$).
and they have a discontinuous derivative at $s=1/2$ (through $\varphi_2$).
Figure~\ref{fig:roughsimest} shows two sample curves (black line) and the corresponding noisy observations (circles).

We consider the configurations $p = 20, 50, 70$ and $T = 50, 100, 200, 400$ as well as $\sigma^2 = 0.01, 0.05, 0.1$.
This gives rise to a total of 36 different settings, which have been repeated 200 times each. The signal is estimated by the methods PCA, PB and FPC. The rest of the procedure is the same as in Section~\ref{s:smoothsim}. 
The results have been summarized for each configuration and approach in Table~\ref{tab:simrough}. We see that the factor model approach (PCA) nearly always outperforms its competitors. It is evident that the penalized B-spline as well as the FPCA approach both fail to accurately estimate the signal at the discontinuity $s=1/3$; see Figures~\ref{fig:roughsimest} and \ref{fig:roughsimcov}. The factor model approach on the other hand is entirely unperturbed by this discontinuity and  outperforms its competitors in many settings by a huge margin.

We also mention that $\hat{L}$ can be overestimated as can be seen in the case of $\sigma^2 = 0.01$. Despite the mild overestimation of the required number of factors, we see no negative impact on the recovery of the signal.

We note that the function $\psi_1(s) = \sqrt{3/2}\varphi_1(s)$ is an eigenfunction of this process. As outlined in Section~\ref{s:eigenfun} we may estimate this eigenfunction from the raw data. Our estimate is subsequently compared to the first functional principal component obtained using the method motivated by \citet{Staniswalis:Lee:1998} and implemented in the package \texttt{refund}. Furthermore, we compare our result to the principal components obtained from using the penalized B-spline model, using the package \texttt{fda}. The resulting estimates are shown in Figure~\ref{fig:roughsimcov}. We can see that both the FPCA and PB approach cannot appropriately recover the jump around $s=1/3$. Our suggested approach recovers this particular feature very accurately.

\begin{figure}[ht]
\begin{center}
\includegraphics[width=7.5cm]{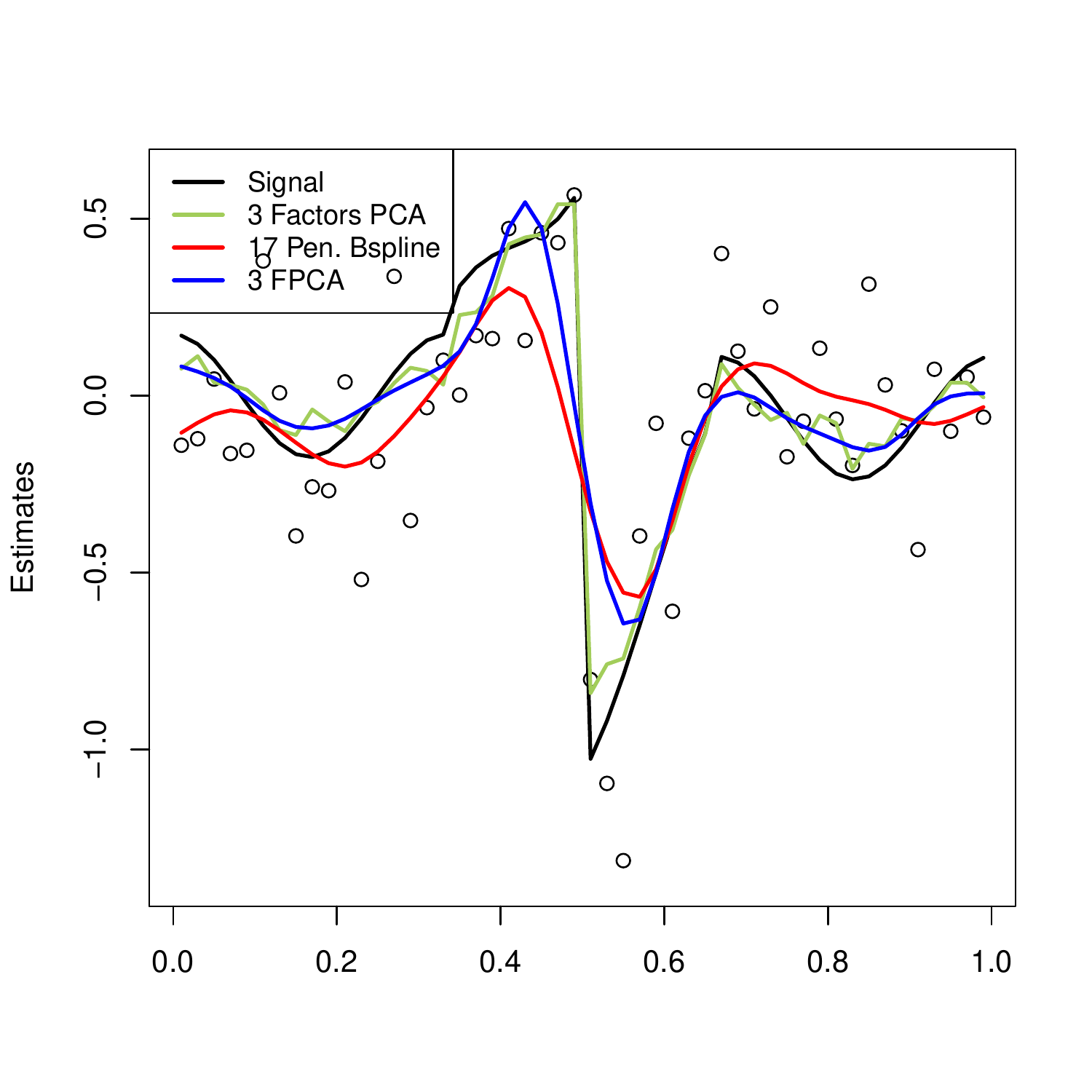}
\includegraphics[width=7.5cm]{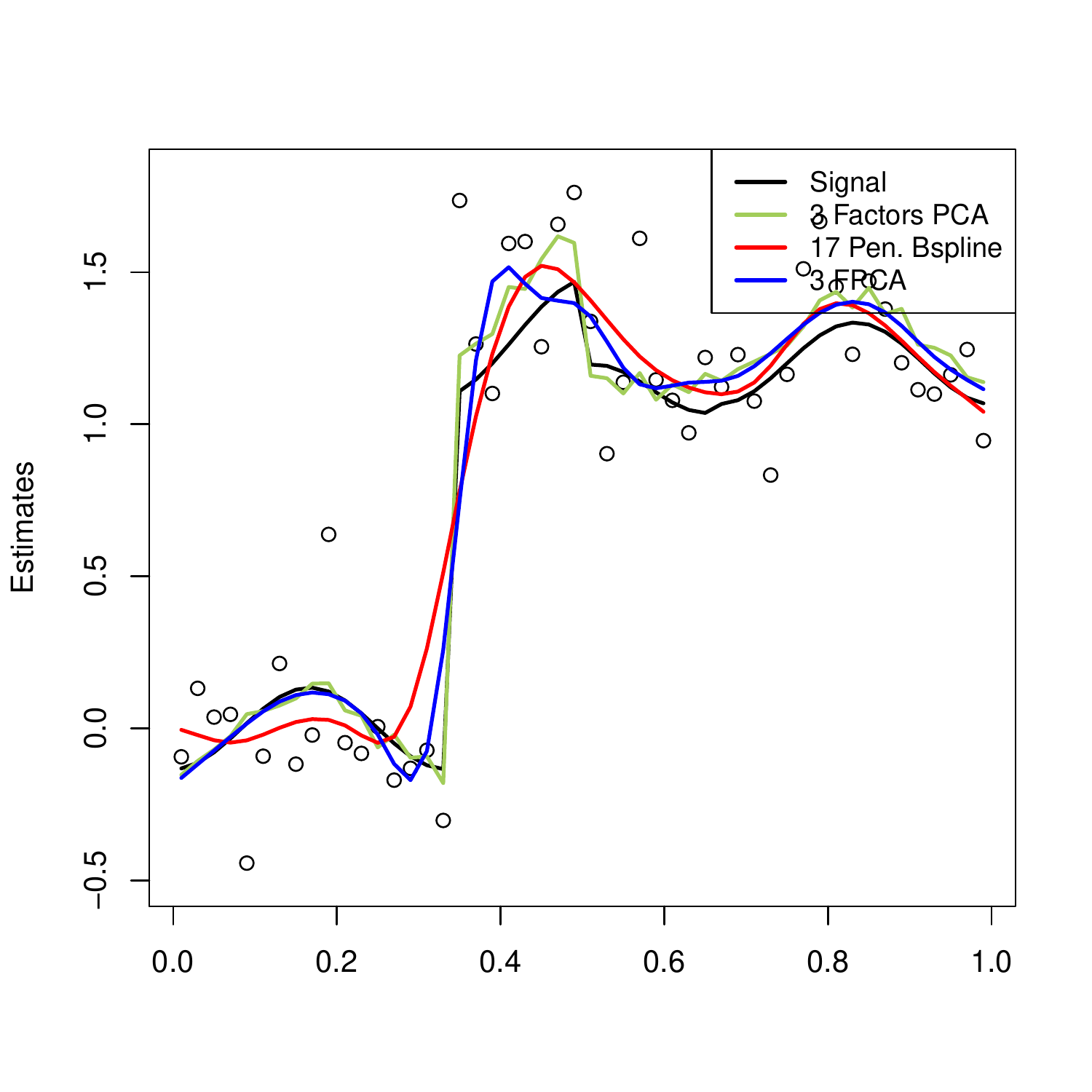}	\caption{Estimates for the rough signal simulation for $p=50, T=200, \sigma^2 = 0.05$. Dots represent the noisy observations.}
\label{fig:roughsimest}
\end{center}
\end{figure}

\begin{figure}[!ht]
\begin{center}
\includegraphics[width=7.5cm]{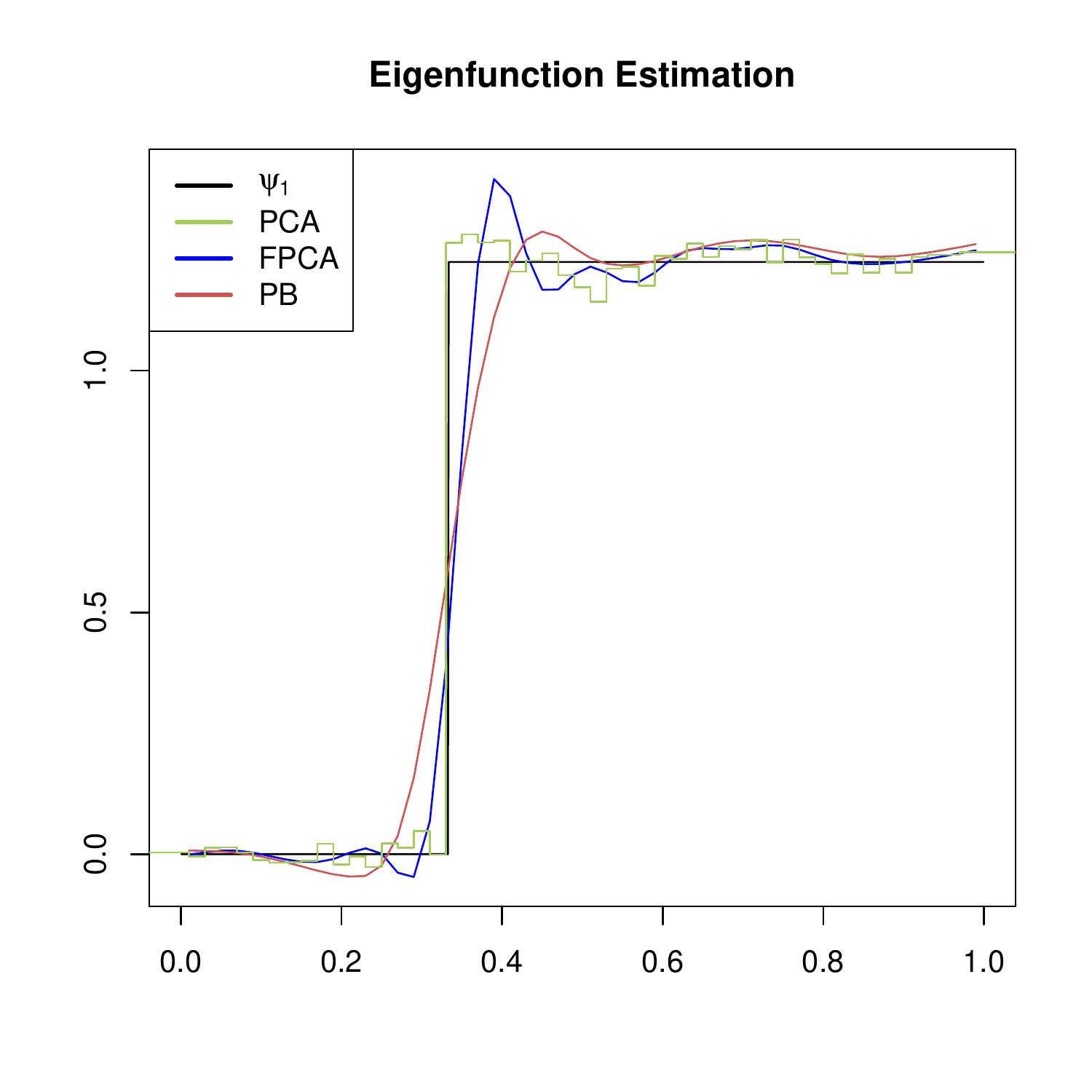}
\includegraphics[width=7.5cm]{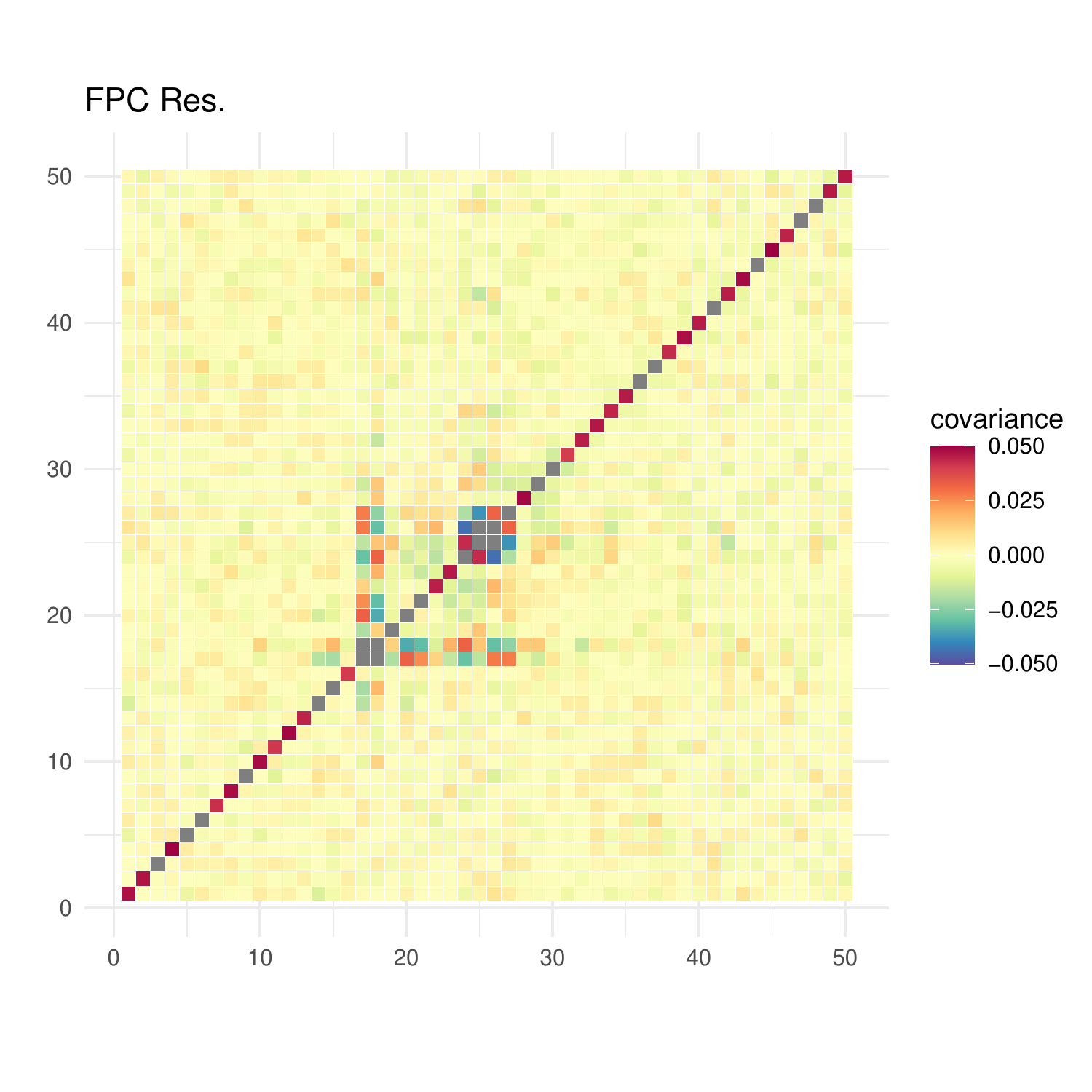}
\includegraphics[width=7.5cm]{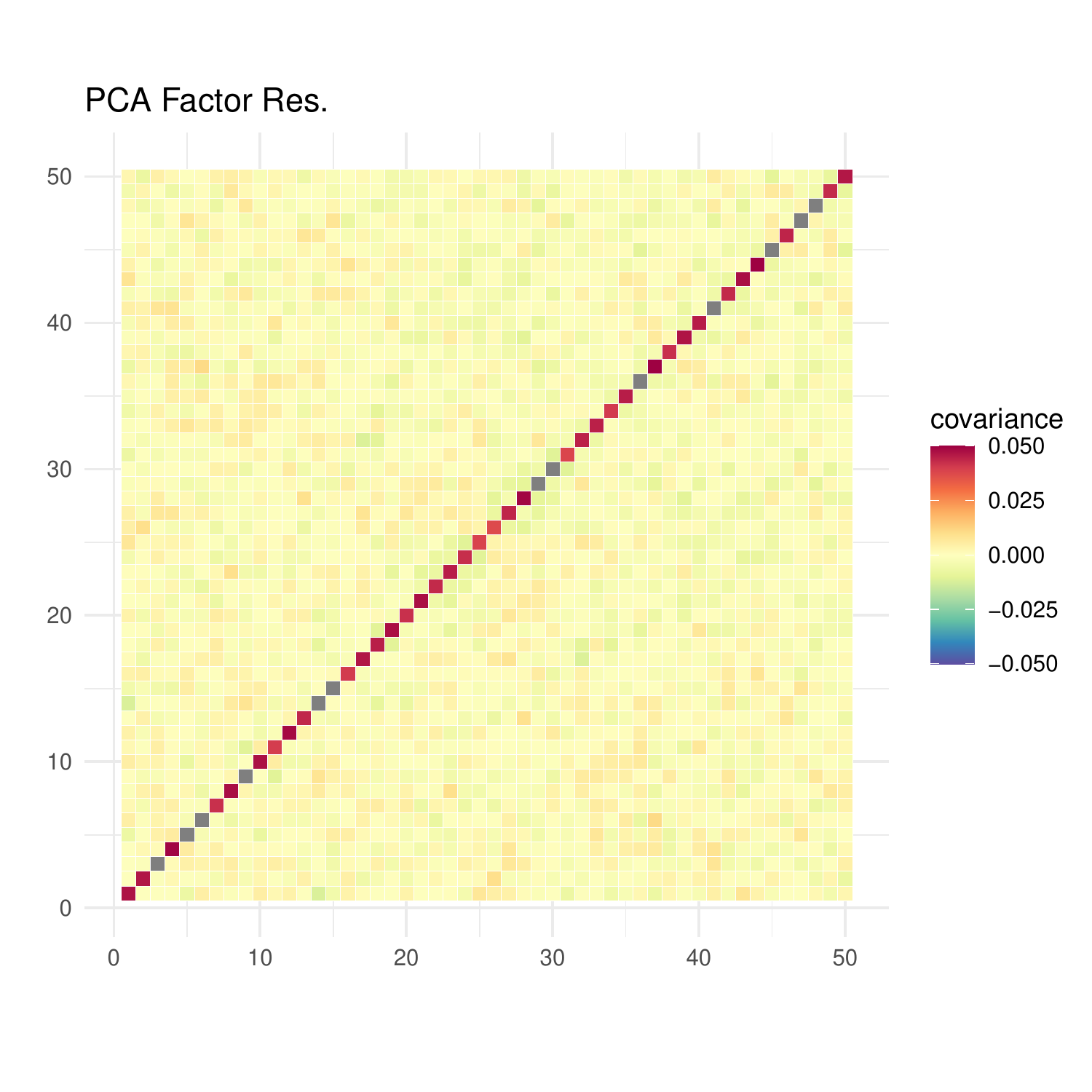}
\includegraphics[width=7.5cm]{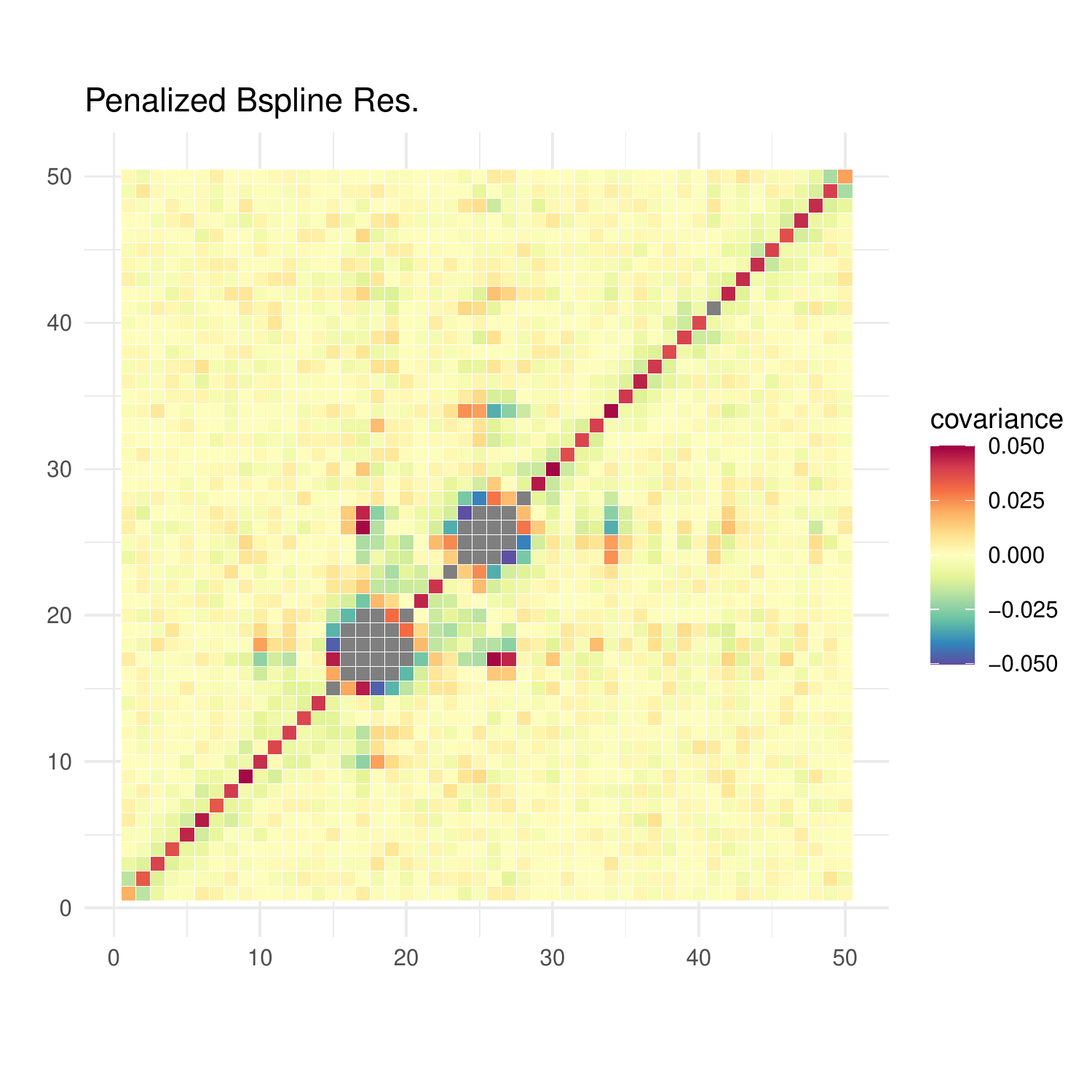}	\caption{Estimates for the first eigenfunction $\psi_1(s)$ (top left) and heat maps of the empirical covariance matrix of the residuals for the FPCA (top right), PCA (bottom left) and PB approach (bottom right) for the rough signal simulation.}
\label{fig:roughsimcov}
\end{center}
\end{figure}

\begin{table}[h]
\centering
\begingroup
\footnotesize
\begin{tabular}{cc|cccc|cccc|cccc}
  \toprule \multicolumn{2}{c}{Dimensions} & \multicolumn{1}{|c}{} & \multicolumn{3}{c}{$\text{SSE}^{\text{appr}}$ ($\sigma_\epsilon^2$ = 0.01)} & \multicolumn{1}{|c}{} & \multicolumn{3}{c}{$\text{SSE}^{\text{appr}}$ ($\sigma_\epsilon^2$ = 0.05)} & \multicolumn{1}{|c}{} & \multicolumn{3}{c}{$\text{SSE}^{\text{appr}}$ ($\sigma_\epsilon^2$ = 0.1)}\\$p$ & $T$ & $\hat{L}$ & PB & FPC & PCA & $\hat{L}$ & PB & FPC & PCA & $\hat{L}$ & PB & FPC & PCA \\ 
  \hline
20 & 50 & 4 & 0.056 & 0.046 & \textbf{0.004} & 3 & 0.072 & 0.05 & \textbf{0.012} & 3 & 0.094 & 0.056 & \textbf{0.025} \\ 
  20 & 100 & 5 & 0.057 & 0.046 & \textbf{0.004} & 3 & 0.072 & 0.051 & \textbf{0.01} & 3 & 0.097 & 0.055 & \textbf{0.02} \\ 
  20 & 200 & 5 & 0.057 & 0.046 & \textbf{0.003} & 3 & 0.069 & 0.049 & \textbf{0.01} & 3 & 0.097 & 0.055 & \textbf{0.018} \\ 
  20 & 400 & 5 & 0.057 & 0.045 & \textbf{0.003} & 3 & 0.069 & 0.049 & \textbf{0.01} & 3 & 0.098 & 0.054 & \textbf{0.017} \\ 
    & & & & & & & & & & & & &  \\ 
50 & 50 & 5 & 0.015 & 0.011 & \textbf{0.003} & 3 & 0.026 & 0.014 & \textbf{0.008} & 3 & 0.035 & 0.018 & \textbf{0.015} \\ 
  50 & 100 & 5 & 0.015 & 0.011 & \textbf{0.002} & 3 & 0.026 & 0.013 & \textbf{0.006} & 3 & 0.034 & 0.016 & \textbf{0.011} \\ 
  50 & 200 & 3 & 0.015 & 0.01 & \textbf{0.001} & 3 & 0.026 & 0.013 & \textbf{0.004} & 3 & 0.034 & 0.016 & \textbf{0.008} \\ 
  50 & 400 & 3 & 0.015 & 0.01 & \textbf{0.001} & 3 & 0.026 & 0.013 & \textbf{0.004} & 3 & 0.034 & 0.016 & \textbf{0.007} \\ 
    & & & & & & & & & & & & &  \\ 
70 & 50 & 5 & 0.011 & 0.006 & \textbf{0.002} & 3 & 0.020 & 0.009 & \textbf{0.007} & 3 & 0.028 & \textbf{0.012} & 0.014 \\ 
  70 & 100 & 3 & 0.011 & 0.006 & \textbf{0.001} & 3 & 0.020 & 0.008 & \textbf{0.004} & 3 & 0.028 & 0.01 & \textbf{0.009} \\ 
  70 & 200 & 3 & 0.011 & 0.006 & \textbf{0.001} & 3 & 0.020 & 0.008 & \textbf{0.003} & 3 & 0.028 & 0.01 & \textbf{0.006} \\ 
  70 & 400 & 3 & 0.011 & 0.006 & \textbf{0.001} & 3 & 0.020 & 0.007 & \textbf{0.003} & 3 & 0.028 & 0.01 & \textbf{0.005} \\ 
   \bottomrule
  \end{tabular}
\endgroup
\caption{Simulation results for the discontinuous signal.} 
\label{tab:simrough}
\end{table}

\subsection{Testing for independent noise}\label{s:testingpractical}
In this section, we empirically investigate the size of the test developed in Section~\ref{s:test}. To this end, we consider the setting of Section~\ref{s:smoothsim}, with $p=365$ observations per curve for $T=200$ curves. These number compare to the real data settings we consider in the next section. The error variance is set to $\sigma^2_U = 4$. We then fit a factor model using the true number of factors $L=21$ and test whether the model residuals are i.i.d. As mentioned in Section~\ref{s:test}, we note that in the case of model residuals, the spectral densities are not well-estimated at very low frequencies. Thus, considering all available frequencies $\lbrace \theta_\ell, \ell \in \lbrace 1, \ldots, q \rbrace \rbrace$ in the test is not appropriate, as the test is too powerful and will reject too often. This may be mitigated by a cutoff $c$, such that we only consider frequencies $\theta_\ell$ with $\ell > cq$. Furthermore, our theoretical results only support the case of $f/T\longrightarrow0$, where $f=\vert \mathcal{F} \vert$ is the size of the set of frequencies considered. In order to remain in this setting, we thin out the observations, taking only every $m$--th frequency into account. We expect this to improve the size of the test. As a point of comparison we have also applied the testing procedure to i.i.d. errors $u_{\text{Norm}} \sim N(0,4)$ and $u_{\text{Exp}} \sim \text{Exp}(2)$ of the same dimension $p$ and $T$ as above. We estimate the variance $\sigma_U^2$ in each instance using $\hat{\sigma}_U^2 := T^{-1}\sum_{t=1} \lbrack 6(p-2) \rbrack^{-1} \sum_{j=2}^{p-1} \lbrack u_{t,j+1} + u_{t,j-1} - 2u_{t,j}\rbrack^2$ in accordance to \cite{gasseretal:1986}. This estimate has proven to work very well in this framework. Each setting has been repeated $1000$ times and we check how often the test rejects $\mathcal{H}_0$ at levels $1, 5, 10 \%$. The results are displayed in Table~\ref{tab:emptestsize}. We can see clearly that in the case of i.i.d. variables the size matches the level quite well, indicating that our test works well in these instances. For the residual errors, we notice that the size improves steadily with a bigger cutoff and more thinning. Still, we see that $\mathcal{H}_0$ is too often rejected. This is not a surprise. Naturally, any estimator comes with some error and there is still structure left over that causes the testing procedure to reject slightly too often. In practice it is hence advisable to inspect the corresponding $p$-values.

\begin{table}[h]
\centering
\begingroup\small
\begin{tabular}{ccc|ccc}
  & cutoff & $m$ & 0.01 & 0.05 & 0.1 \\ 
  \hline
 & 0.05 & 3 & 0.090 & 0.187 & 0.270 \\ 
   & 0.05 & 5 & 0.060 & 0.126 & 0.177 \\ 
   & 0.05 & 10 & 0.044 & 0.106 & 0.157 \\ 
   & & & & & \\ 
 & 0.10 & 3 & 0.059 & 0.130 & 0.194 \\ 
   & 0.10 & 5 & 0.045 & 0.102 & 0.160 \\ 
   & 0.10 & 10 & 0.045 & 0.104 & 0.146 \\ 
   & & & & & \\ 
 & 0.20 & 3 & 0.051 & 0.112 & 0.167 \\ 
   & 0.20 & 5 & 0.047 & 0.097 & 0.156 \\ 
   & 0.20 & 10 & 0.046 & 0.081 & 0.125 \\ 
   \hline
 $u_{\text{Norm}}$ & 0.10 & 3 & 0.013 & 0.047 & 0.102 \\ 
  $u_{\text{Exp}}$ & 0.10 & 3 & 0.014 & 0.054 & 0.097 \\ 
  \end{tabular}
\endgroup
\caption{Simulation results for the empirical testsize.} 
\label{tab:emptestsize}
\end{table}

\section{Real data illustrations}\label{s:real}

In the following subsections we analyse annual temperature curves from Canada. We differentiate between two settings: On the one hand, we analyse a \textit{temporal setting} in the sense that we consider curves over a few years in one location. On the other hand, we consider a \textit{spatial setting}, where we investigate the same year for different weather stations. Our objective is to transform daily mean temperature data throughout a year into annual temperature curves. The data was acquired from https://climate.weather.gc.ca/ and curves with more than 10$\%$ missing observations were discarded in both the temporal and spatial setting. Remaining missing observations were imputed by using interpolation.

We fit factor models (PCA) and give comparisons to the basis function and the FCPA approach. \citet{ramsay09} have smoothed this type of data with 65 Fourier basis functions and a penalization term. We follow this route, but instead use a B-spline basis (also with 65 basis functions) and a roughness penalty of the form $\int (f''(x))^2dx$. The tuning parameter controlling the size of the penalization term is chosen with generalized cross validation techniques as in \citet{ramsay09}.

\subsection{Temporal Data: St. Margaret's Bay}\label{s:temporal}

In this section, we consider temporal data from St. Margaret's Bay in Nova Scotia, Canada. This weather station has a long history of recorded data, from which we will use a selection of $91$ yearly curves ranging from $1923$ to $2020$. The available curves in this dataset from $2000$ to $2020$ can be seen in Figure~\ref{fig:stmargdata}. We observe a typical seasonal shape with very cold winters and rather mild summers. The data is rather noisy and the goal is to separate the underlying signal from the unsystematic noise.

\begin{figure}[!ht]
\begin{center}
\includegraphics[width=12cm]{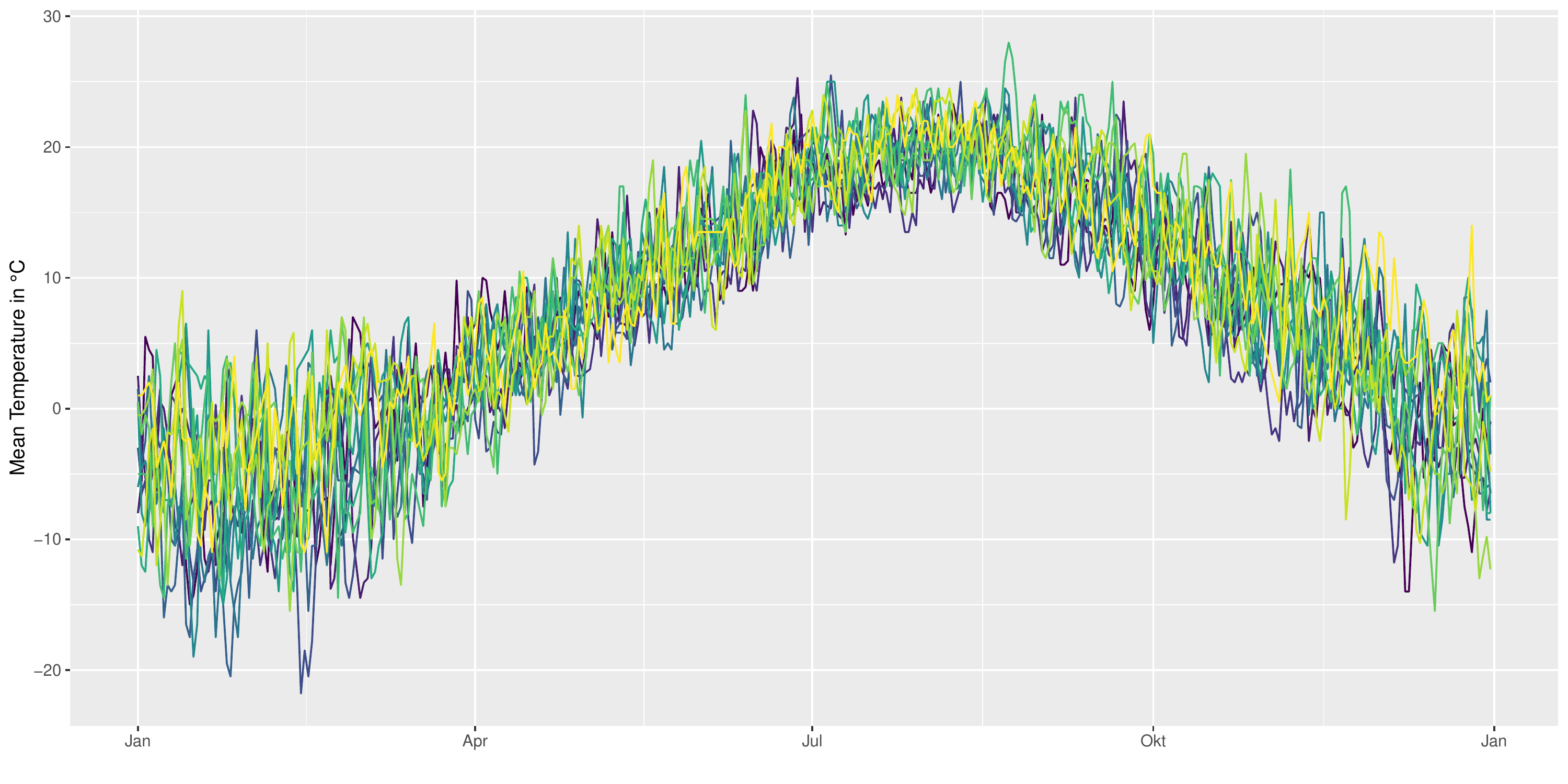}	
\caption{Yearly Temperature Data in St. Margaret's Bay, Nova Scotia, from 2000-2020.}
\label{fig:stmargdata}
\end{center}
\end{figure}

In order to apply the factor analysis approach, we first need to choose the factors for the model. As mentioned in Section~\ref{s:sim}, factor model estimation is robust to overestimation of $L$, yet rather volatile when it comes to underestimating $L$. We refer to Figure~\ref{fig:stmargLhat}, where we show the alternative Scree Plot as introduced in Section~\ref{s:L} on the left and the Classic Scree Plot on the right. For our alternative Scree Plot we have chosen a cutoff of $10\%$ and have thinned out the frequencies to a third. We can see in Figure~\ref{fig:stmargLhat} that the Scree, ED and BCV choices for the estimate of $L$ all agree on $\hat{L}=3$. The alternative Scree Plot is not particularly conclusive in this example. The values of the test statistic $\Lambda_\ell$ remain very large for all choices of $L$. This indicates that the residuals are not iid and our method for choosing $L$ is not applicable here. We hence choose $\hat{L}=3$ in accordance with the other methods. 

\begin{figure}[!ht]
\begin{center}
\includegraphics[width=7.5cm]{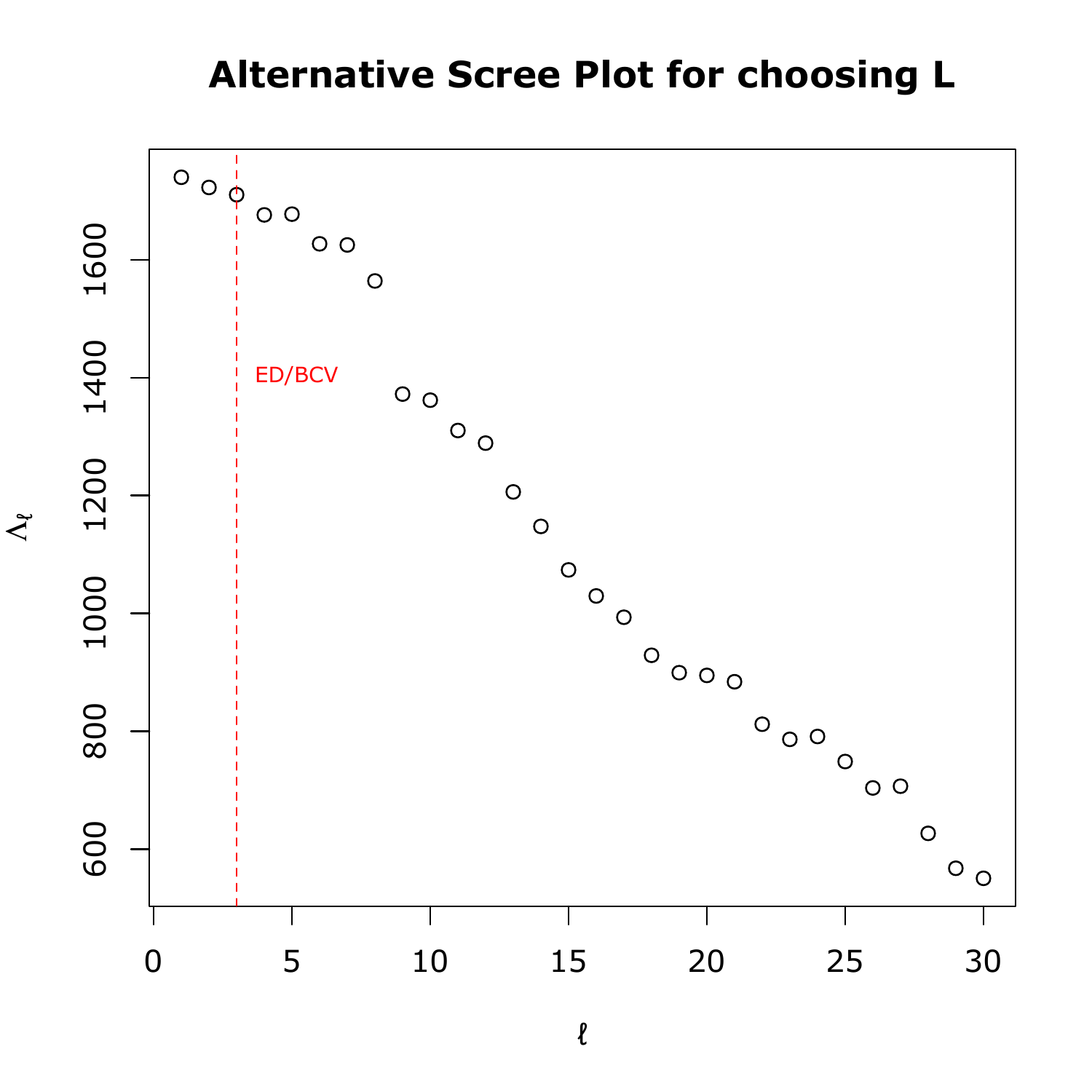}	
\includegraphics[width=7.5cm]{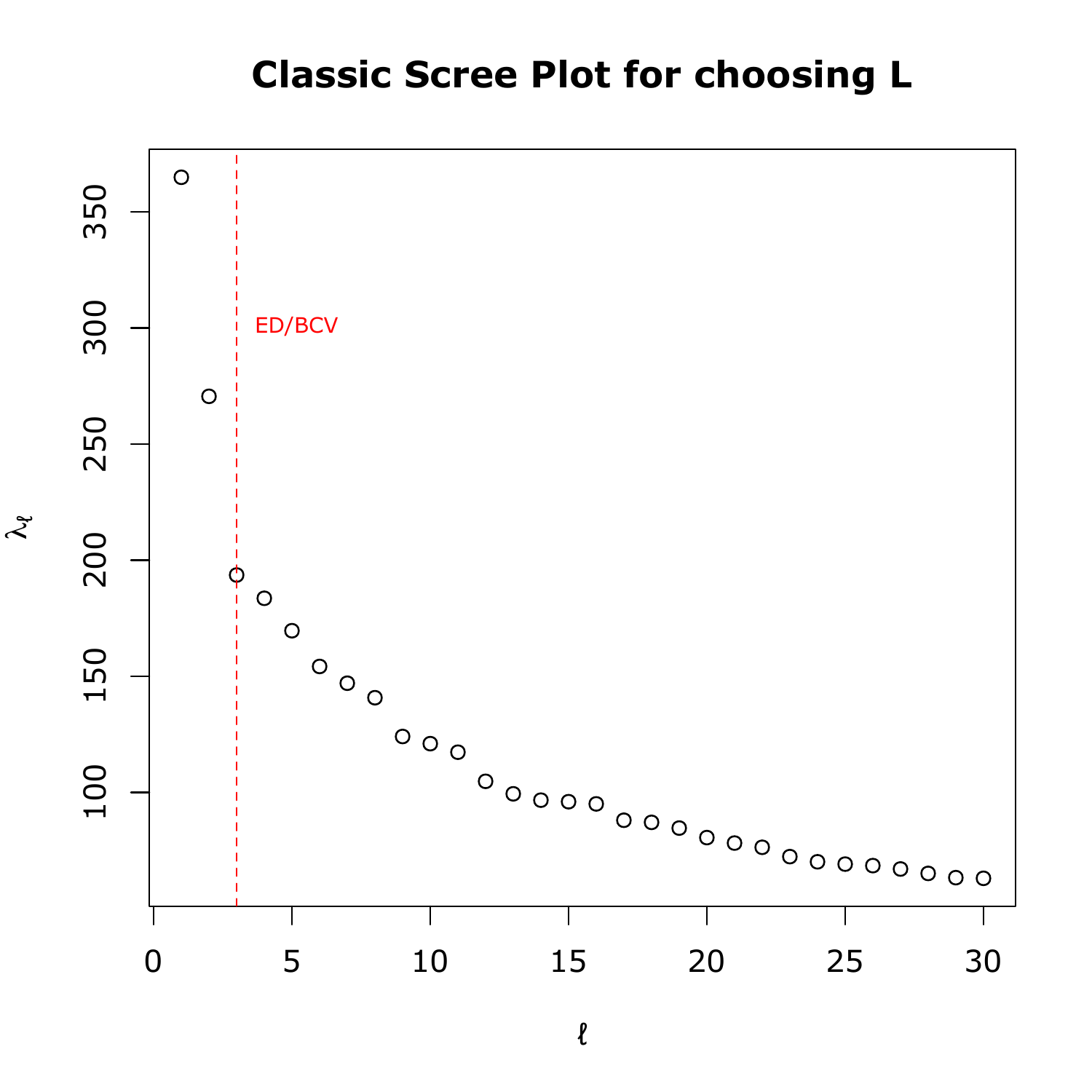}	
\caption{Alternative (left) and classic (right) Scree plots to estimate the number of factors for the temporal data example. The estimates $\hat{L}_{\text{ED}}$ and $\hat{L}_{\text{BCV}}$ are indicated in red.}
\label{fig:stmargLhat}
\end{center}
\end{figure}

To keep this analysis concise, we consider for the factor model approach only the classic principal components (PCA) and for the splines approach only the penalized B-Splines (PB). Using maximum-likelihood instead of PCA or Fourier functions instead of B-splines gives almost identical results.  While our factor model estimate looks quite rough it is much less wiggly than what we get with FPC and PB. Our approach leads to the largest variance in the residuals. See Figure~\ref{fig:stmargfit}. This is in entire contrast to the spatial data discussed in the next section.
\begin{figure}[!ht]
\begin{center}
\includegraphics[width=11cm]{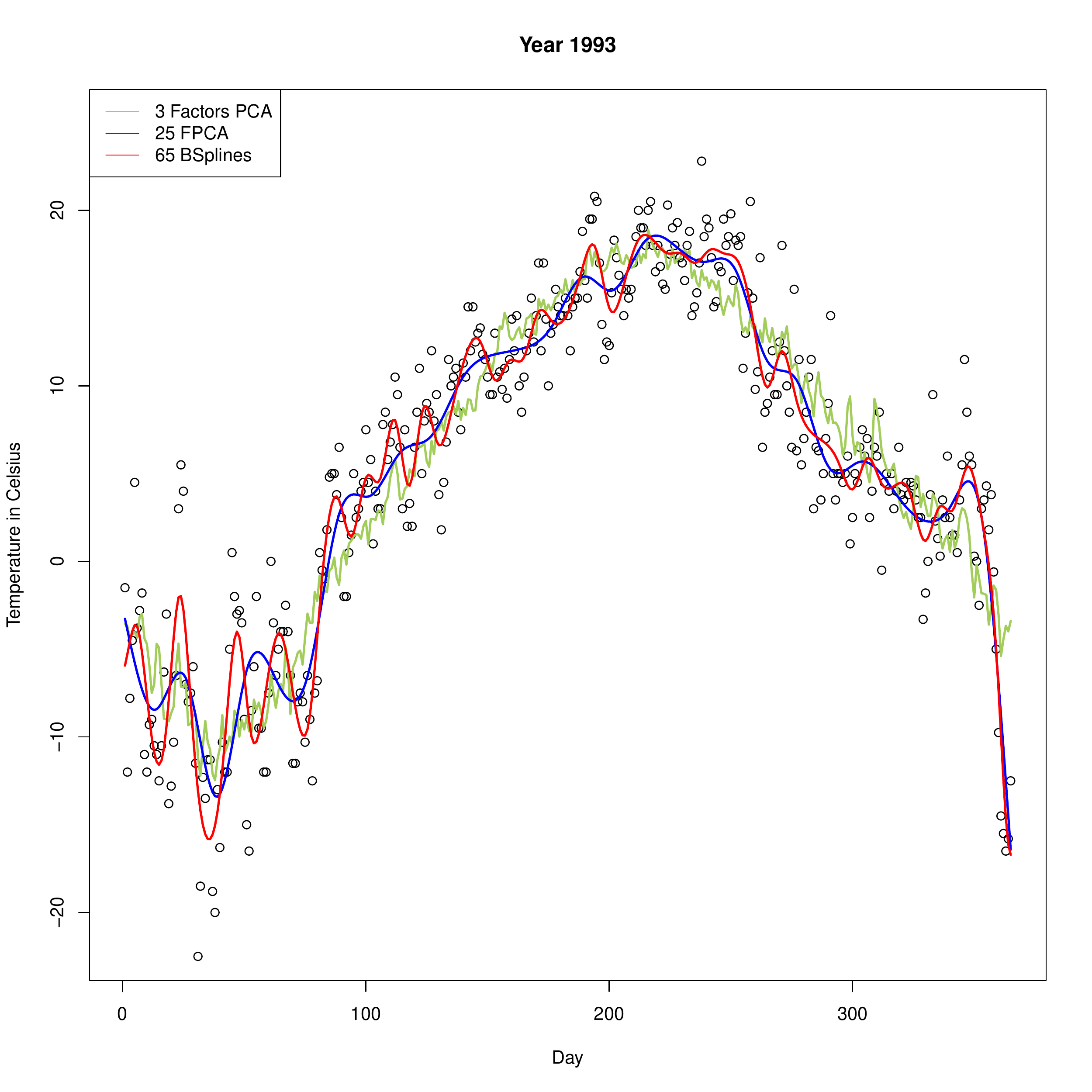}	
\caption{St. Margaret's Bay temperature curve fits for the year 1993. Dots represent noisy data.}
\label{fig:stmargfit}
\end{center}
\end{figure}

From Figures~\ref{fig:stmargintro} and \ref{fig:stmargfitcovs} we see some interesting structure in the residual covariance when we work with PB and FPC. The heatmaps of the covariance and correlation matrices show alternating negative and positive bands parallel to the diagonal, which translate into oscillating autocovariances.

\begin{figure}[!ht]
\begin{center}
\includegraphics[width=7cm]{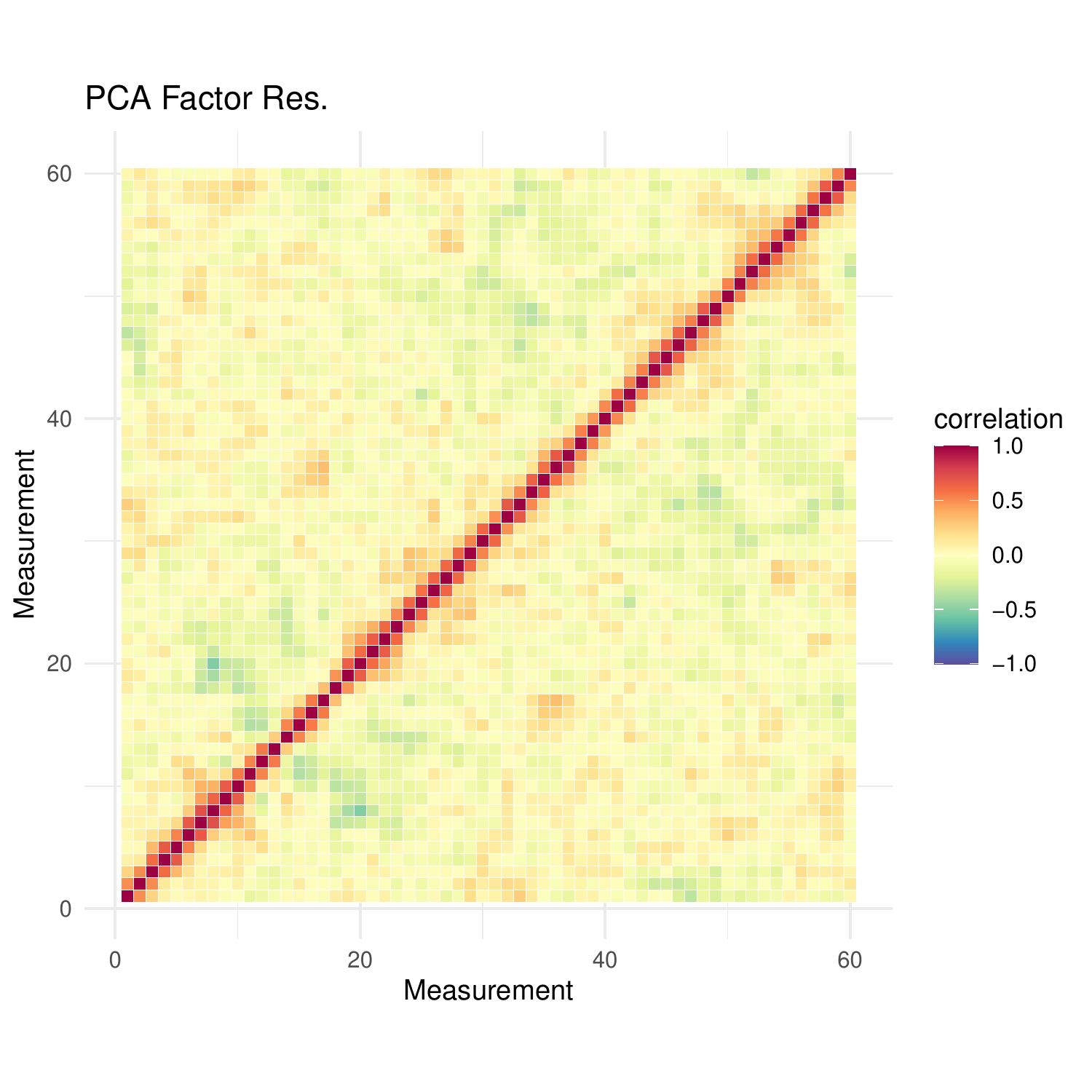}	
\includegraphics[width=7cm]{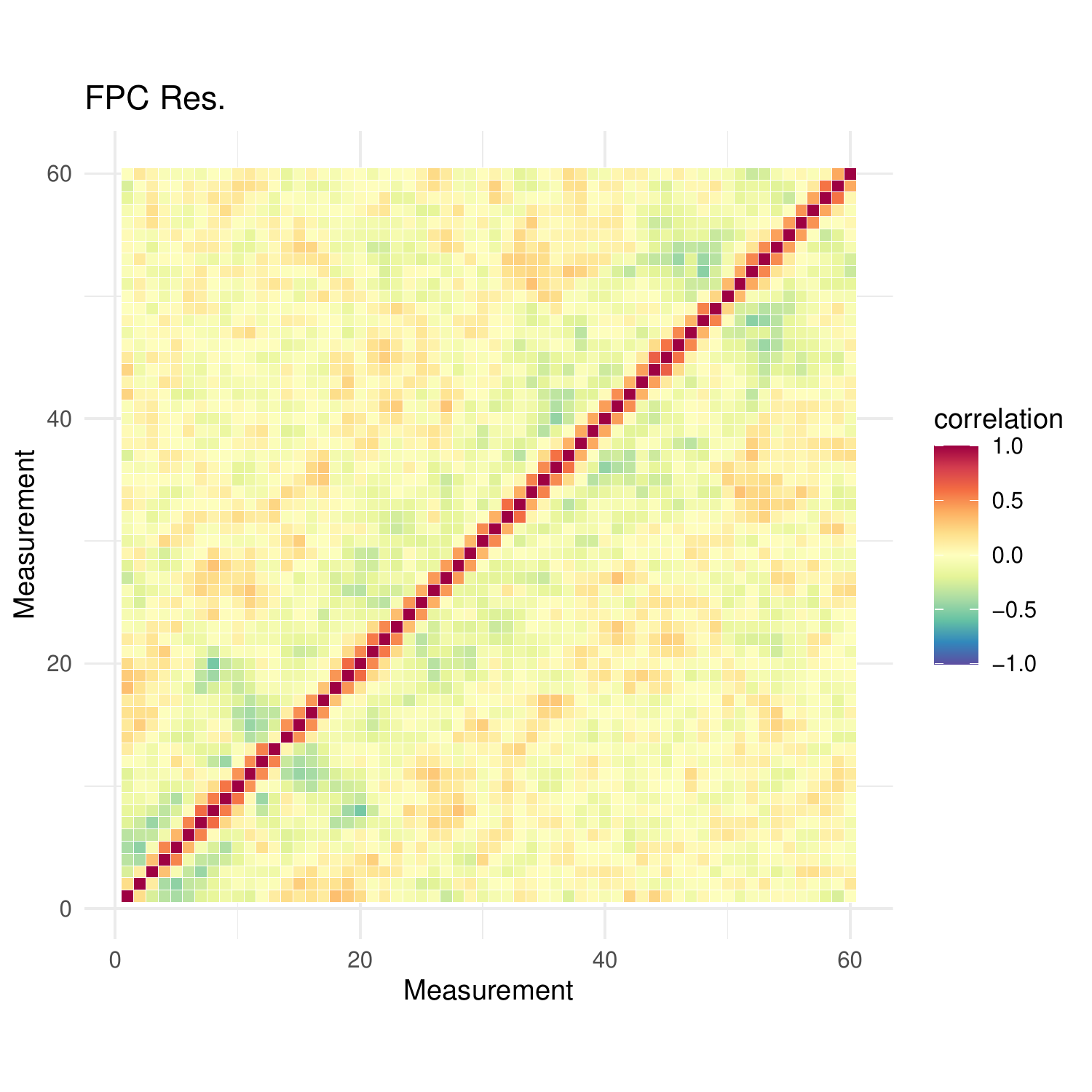}	
\includegraphics[width=6.5cm]{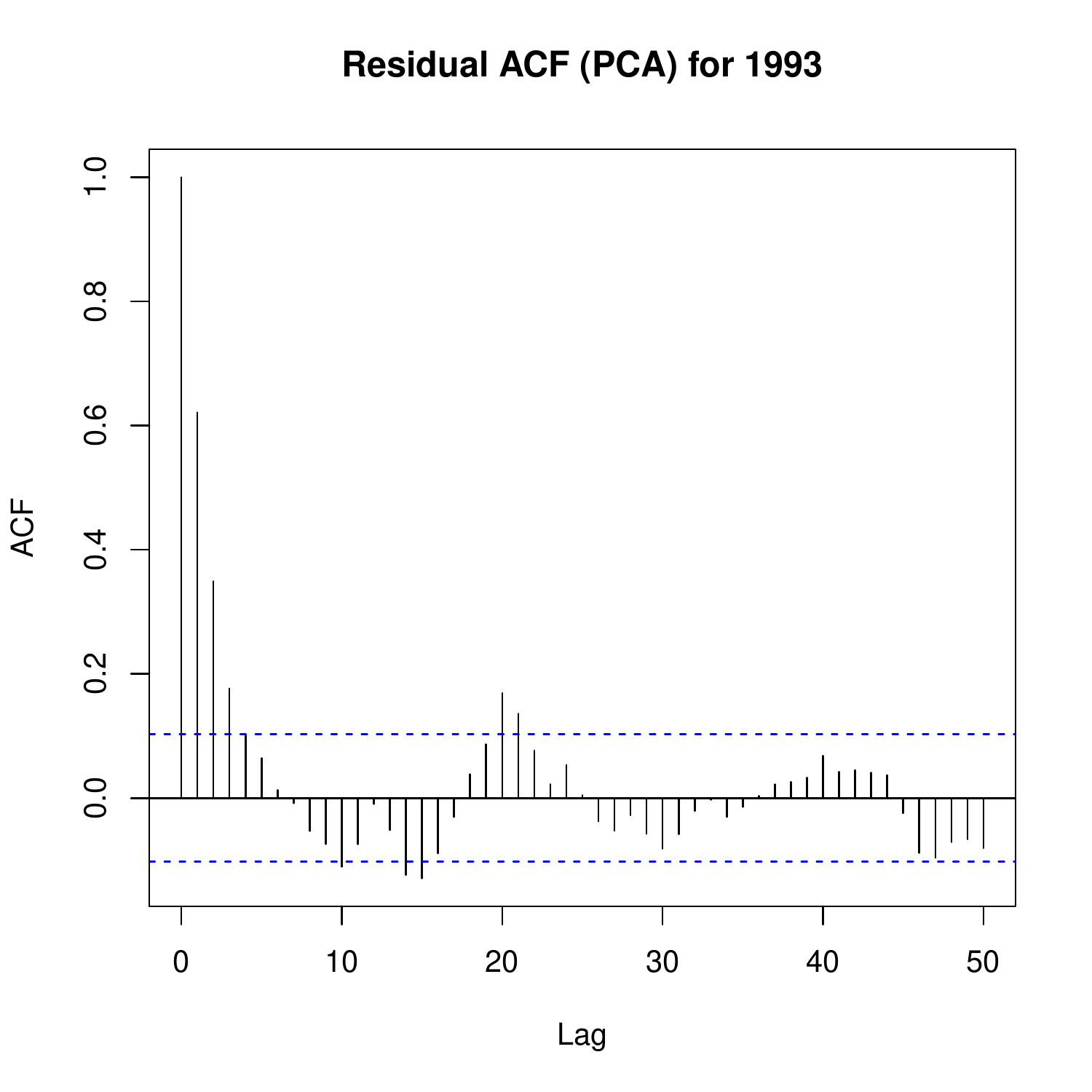}	
\includegraphics[width=6.5cm]{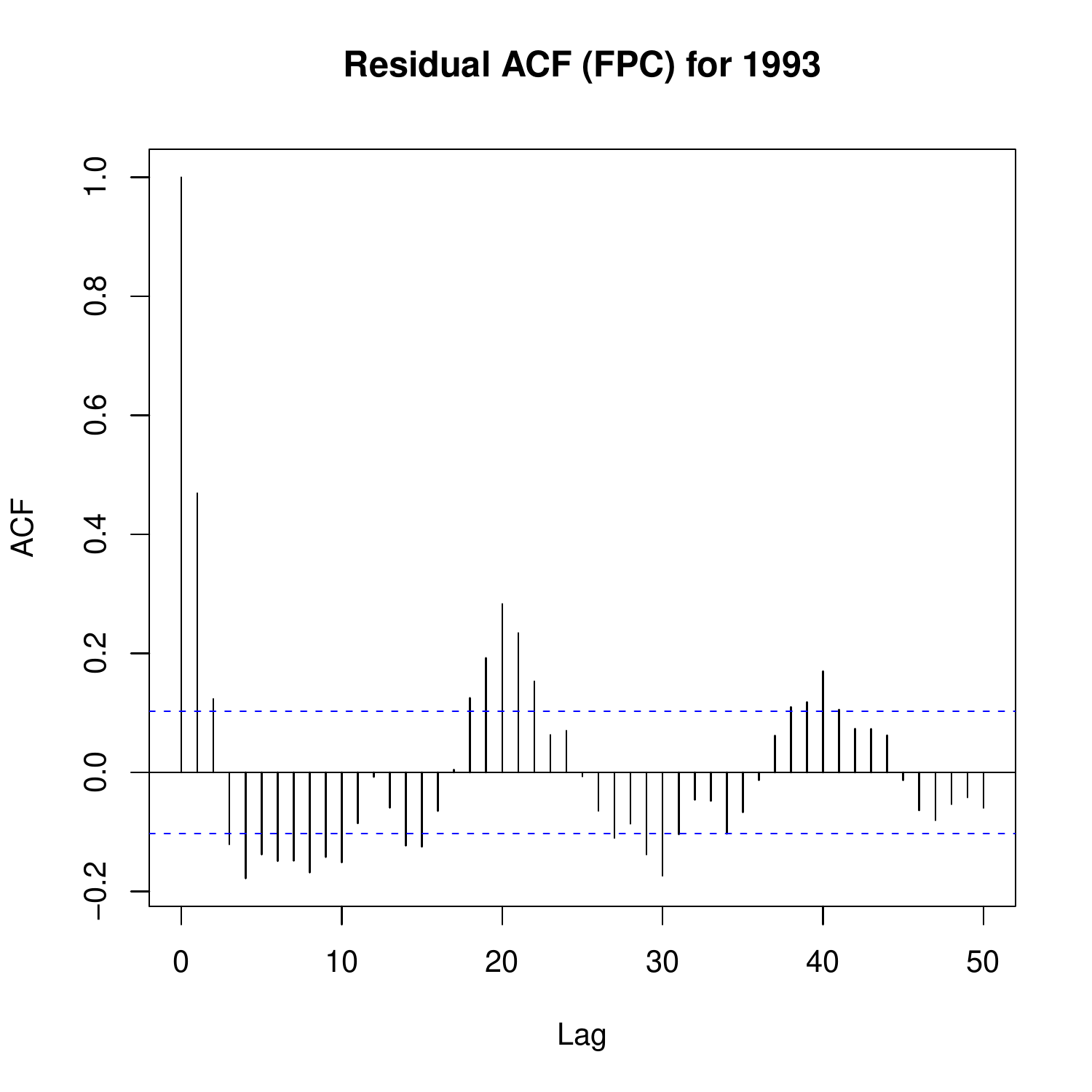}	
\caption{Heat maps of the empirical residual correlation matrices (restricted to the first two months a year) for the PCA (left) and FPC (right) approach for the St. Margaret's Bay data and associated estimated Residual ACF for the year 1993 for the PCA (bottom left) and FPC (bottom right) approach.
}
\label{fig:stmargfitcovs}
\end{center}
\end{figure}

Although the residuals clearly are not iid in this example, we perform our test for independent errors as discussed in Section~\ref{s:test} for illustrative purposes. The values of the test statistics are given in Table~\ref{tab:teststmarg}. The tests were performed with a $10\%$ cutoff and using every third frequency. In all three cases the values of $\hat\Lambda_{\text{inf}}$ are way beyond the critical values from a standard normal distribution. 

The fact that we are far away from iid errors can also be seen from the averaged periodogram ordinates $\frac{1}{T}\sum_{t=1}^T I_{\hat U}(\boldsymbol{\theta})$ which we compute for errors obtained from the three investigated methods (see Figure~\ref{fig:stmargperiodogramm}). These averages are estimators for the respective spectral densities of the errors. Observe that we have a strong bias towards zero at frequencies close to 0 for all methods, which indicates that they filter out low frequencies. This phenomenon is most pronounced for the penalized B-splines. Otherwise the shape  of the spectral density estimator is reminiscent of autoregressive errors. This seems like a reasonable assumption for a dataset of this type and allows for mild temporal dependence. 

For illustrative purposes, we fit AR(2) processes to the residual vector components. For the factor model the resulting fits show that in general, $\hat\phi_1$ lies at around $0.6$ and $\hat\phi_2$ is typically at around $-0.15$, see Figure~\ref{fig:stmargar2coef} for the associated boxplots. In Figure~\ref{fig:stmargperiodogramm}, we have plotted the corresponding spectral densities of ten such estimated processes AR(2) processes.
A periodic form of the acf as we have seen it now mainly for approaches PB and FPC would amount to AR(2) processes, where the roots of the characteristic polynomial are complex (see e.g.\ \citet{brockwell:davis:1991}). Investigating this line of thought we find that out of the $T=91$ residual time series, $56$ of the associated characteristic polynomials have complex roots when we use the factor model approach. For the approaches PB and FPC, all of the associated characteristic polynomials have complex roots. 

\begin{table}[H]
\centering
\begin{tabular}{r|c|c|c}
  & PCA & PB & FPC \\ 
  \hline
$\hat\Lambda_{\text{inf}}$ & 1710.62 & 840.04 & 1662.23 \\ 
  \end{tabular}
\caption{Teststatistic $\hat\Lambda_{\text{inf}}$ for the residuals for St. Margaret's temperature curves.} 
\label{tab:teststmarg}
\end{table}

\begin{figure}[!ht]
\begin{center}
\includegraphics[width=4.5cm]{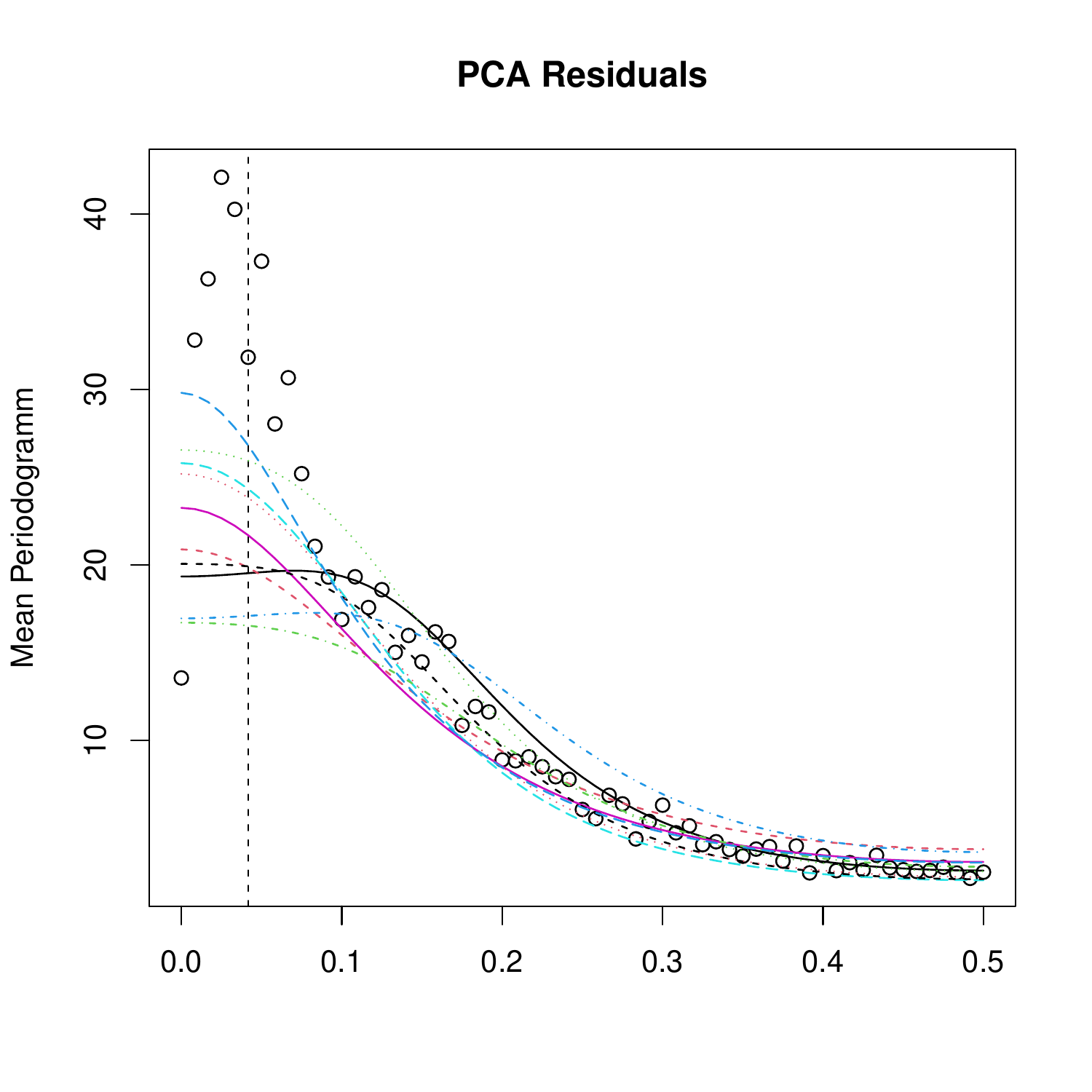}	
\includegraphics[width=4.5cm]{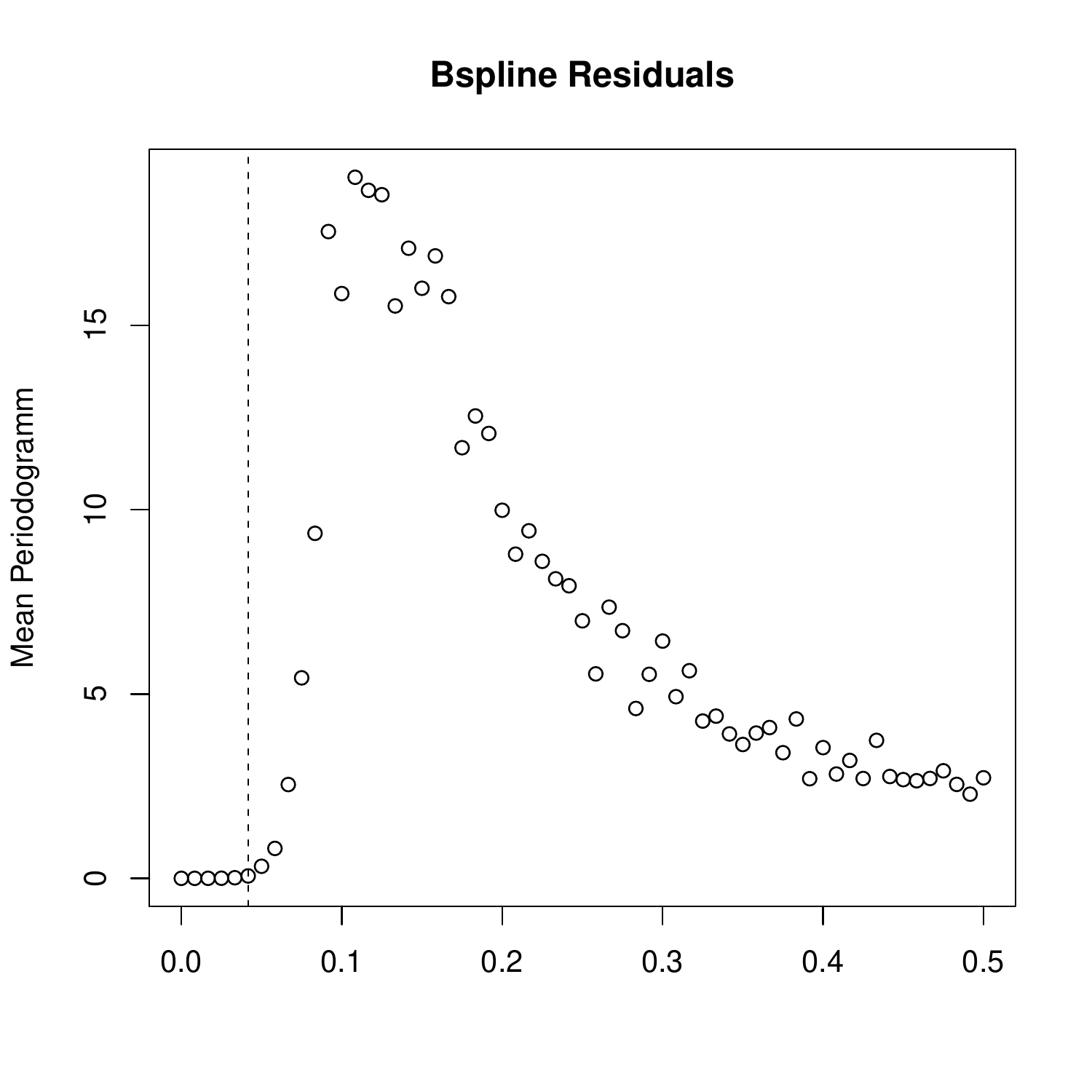}
\includegraphics[width=4.5cm]{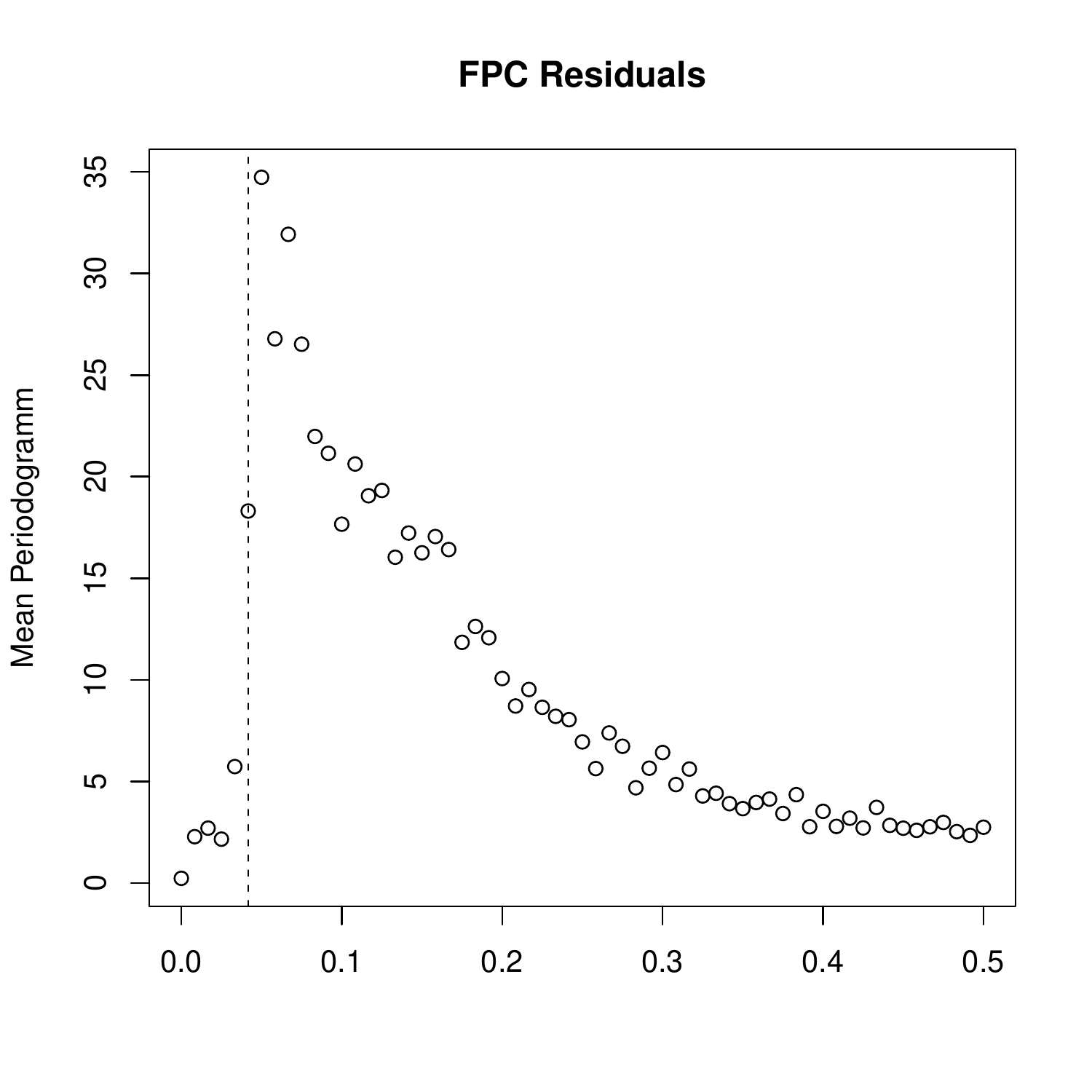}
\caption{Mean Periodogramm $\xi=T^{-1} \sum_{t=1}^T I_{\hat{U}_t}(\bm\theta)$ for the PCA (left), PB (middle) and FPCA (right) residuals in St. Margaret's Bay Data. Vertical dotted line indicates the $10\%$ cutoff. Dotted and colorful lines on the left indicate the estimated spectral density of estimated AR(2) processes for $10$ factor model residual curves.}
\label{fig:stmargperiodogramm}
\end{center}
\end{figure}

\begin{figure}[!ht]
\begin{center}
\includegraphics[width=4.5cm]{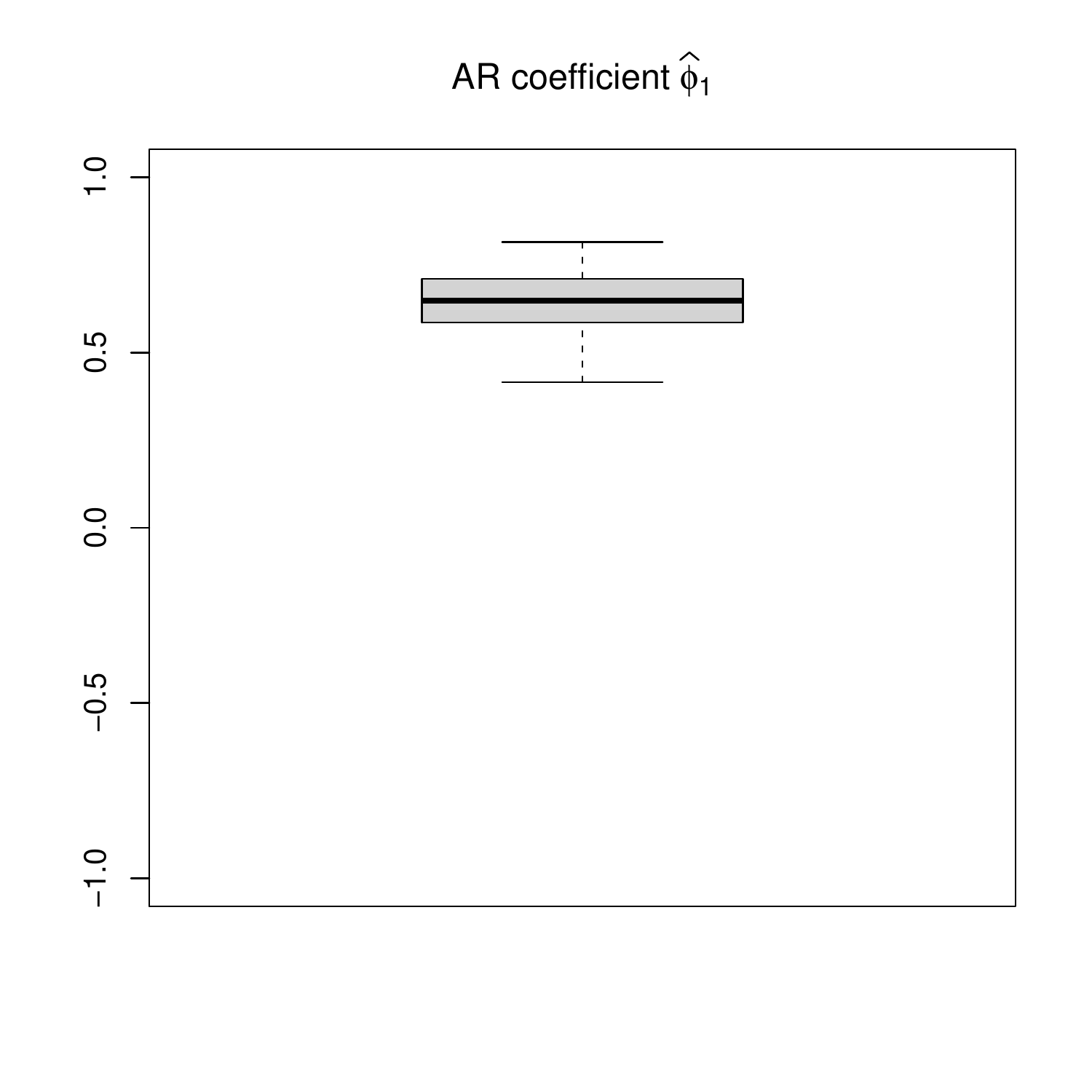}	
\includegraphics[width=4.5cm]{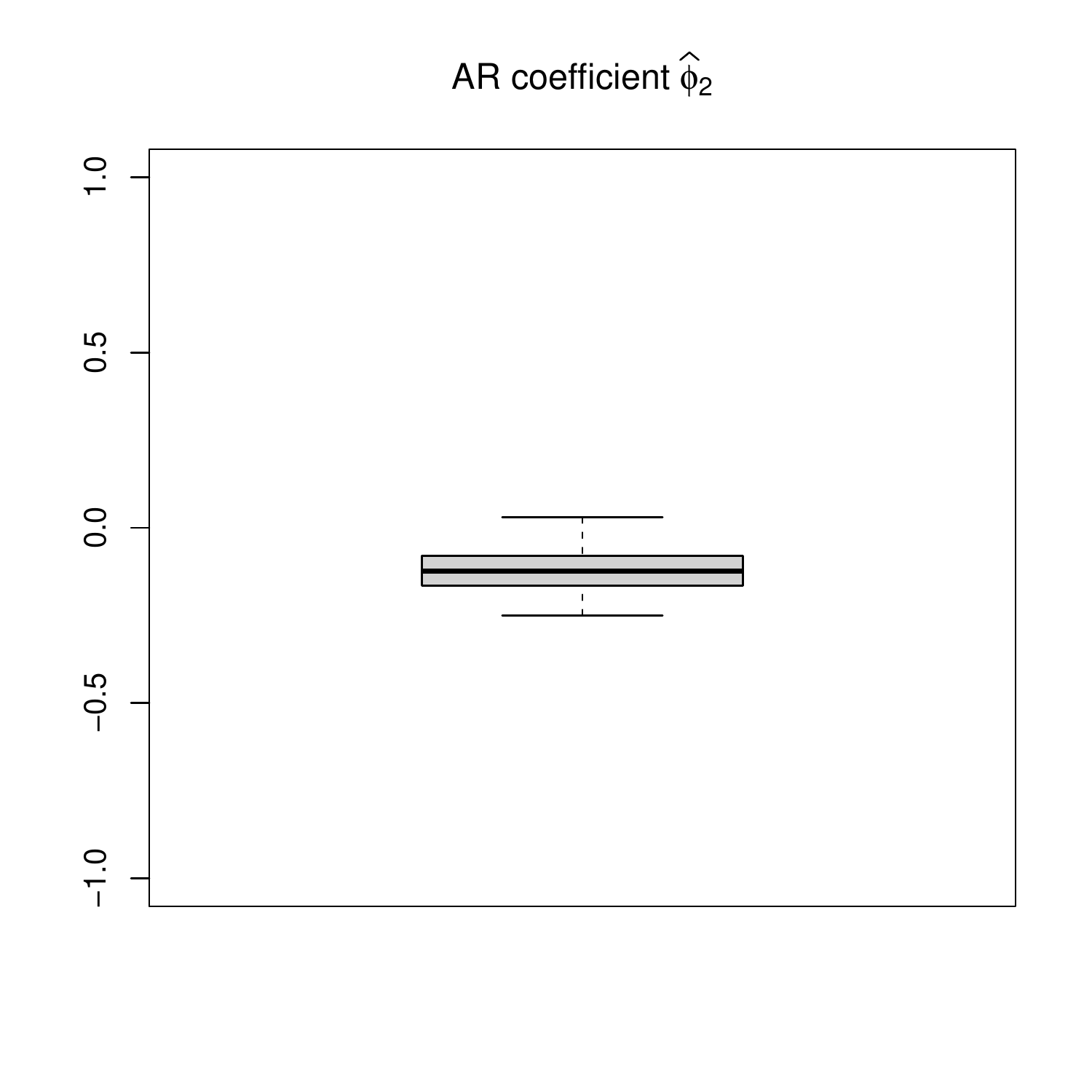}
\includegraphics[width=4.5cm]{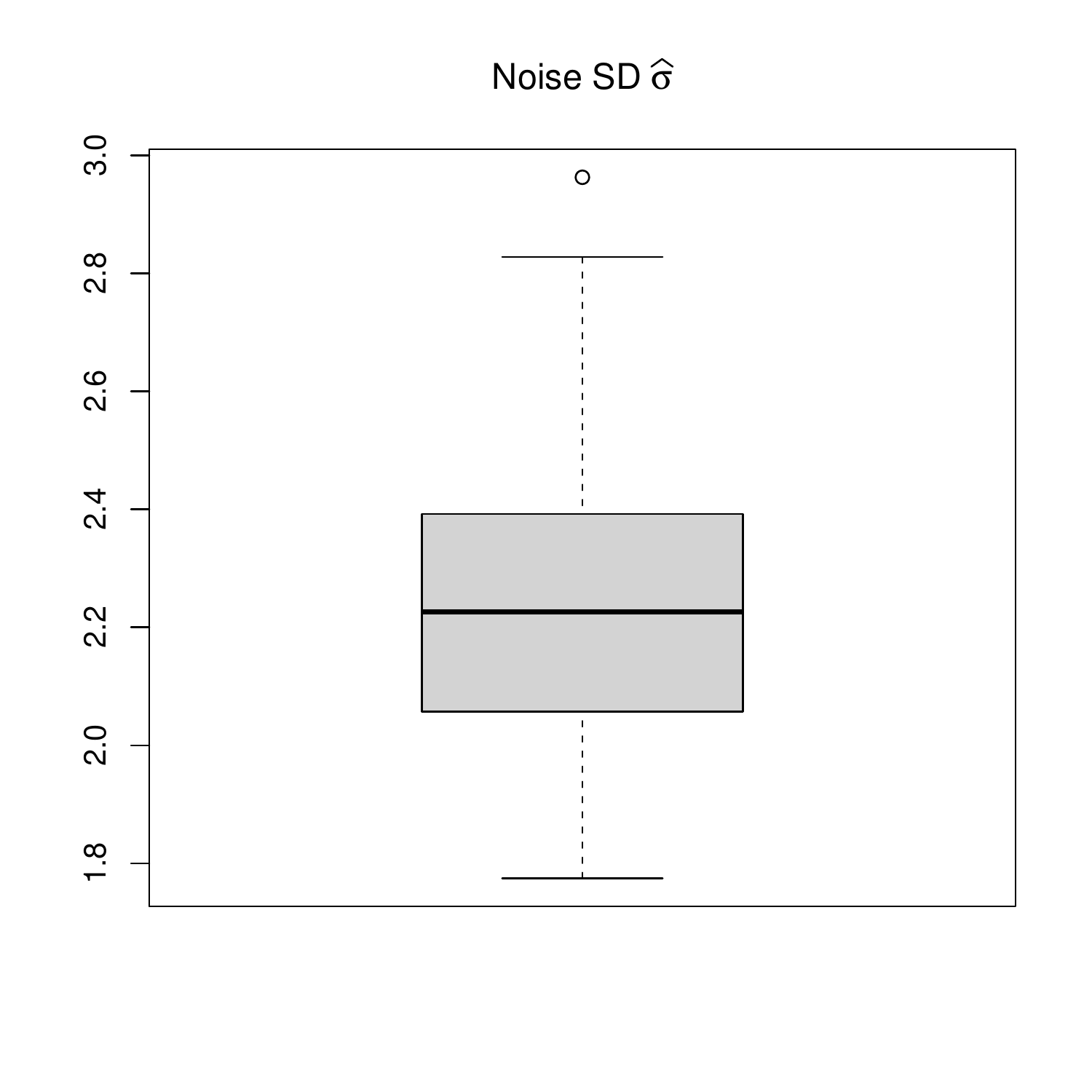}
\caption{Boxplots of the estimated coefficients and noise standard variance of the residual AR(2) processes for the factor model (PCA).}
\label{fig:stmargar2coef}
\end{center}
\end{figure}

\subsection{Spatial Data: Canadian Weather Stations}\label{s:spatial}

Now we consider a spatial setting, where each of the annual curves corresponds to a weather station. To this end, we have compiled the data from weather stations in the provinces of Quebec and Ontario with daily mean temperature measurements available in 2013. After imputing scarcely scattered missing values and removing stations with too much missing data, we have $T=213$ curves left. For illustrative purposes, we show in Figure~\ref{fig:canadadata} the curves associated to the first $100$ stations in alphabetical order. While the same general structure as in the temporal setup can be observed, we expect here a different residual behavior. The idiosyncratic components now describe a station-specific error. Since periods of too warm or too cold temperatures are likely to occur across several stations, we expect a close co-movement resulting in much smaller idiosyncratic noise terms. 

\begin{figure}[!ht]
\begin{center}
\includegraphics[width=12cm]{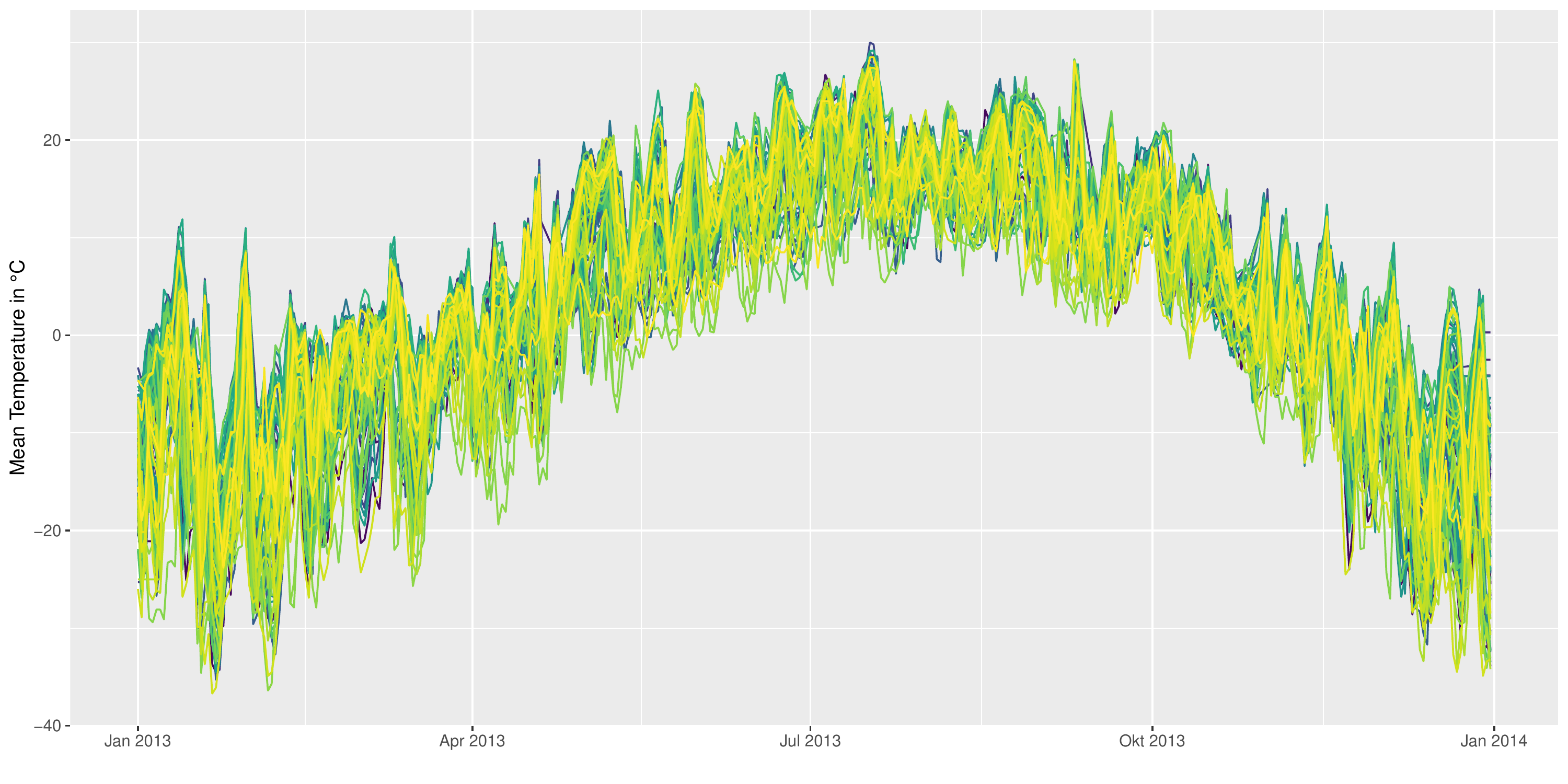}	
\caption{Annual Temperature Data in 2013 for 100 Stations in Quebec and Ontario, Canada.}
\label{fig:canadadata}
\end{center}
\end{figure}

We begin by choosing the number of factors. The Scree plot (see Figure~\ref{fig:canadaLhat}) indicates  $\hat{L}=6$, which coincides with the choice by the ED criterion, whereas BCV sets $\hat L=27$. The latter is close to $\hat L =30$, which we deduce from our alternative Scree plot and which we finally select for this data. 

\begin{figure}[!ht]
\begin{center}
\includegraphics[width=7.5cm]{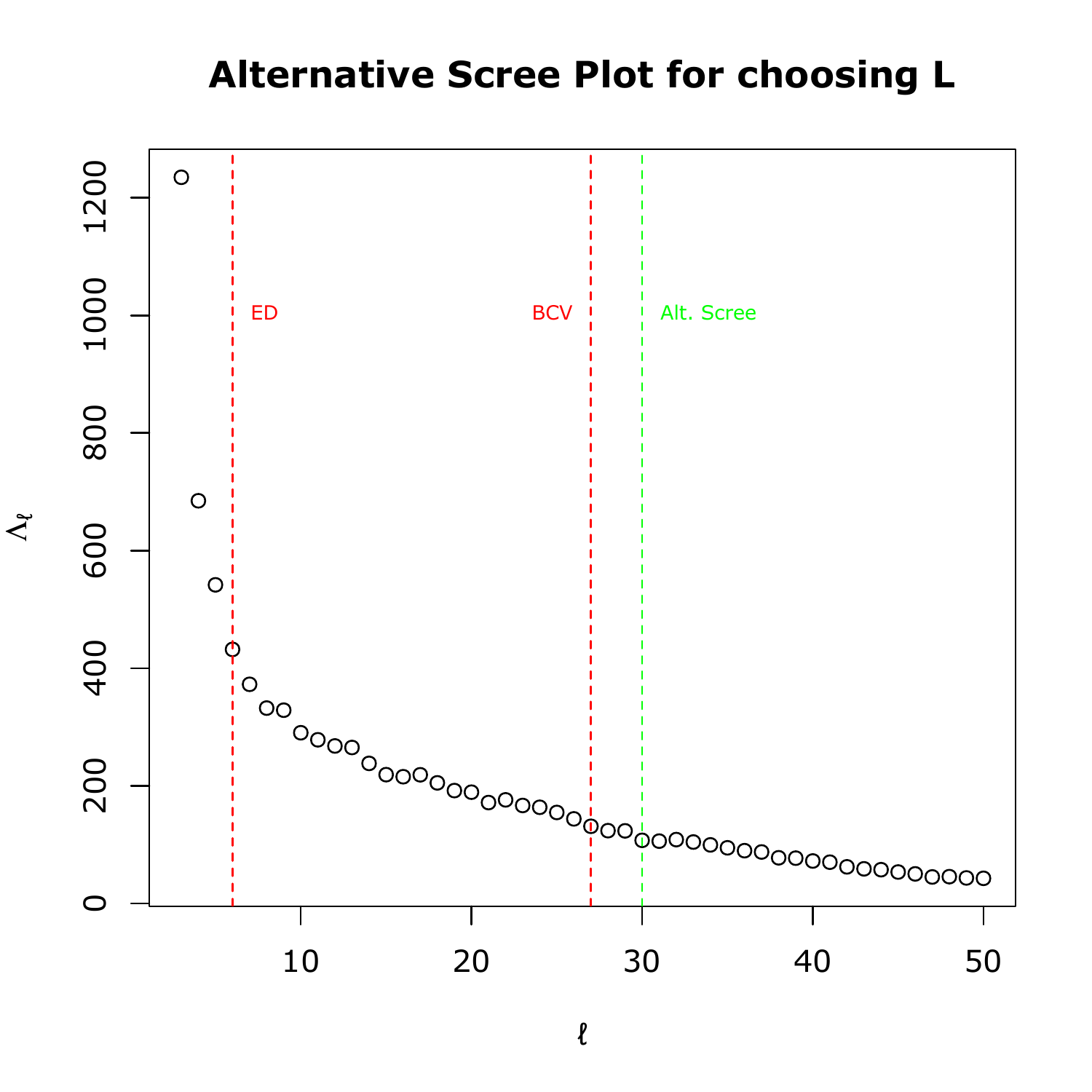}	
\includegraphics[width=7.5cm]{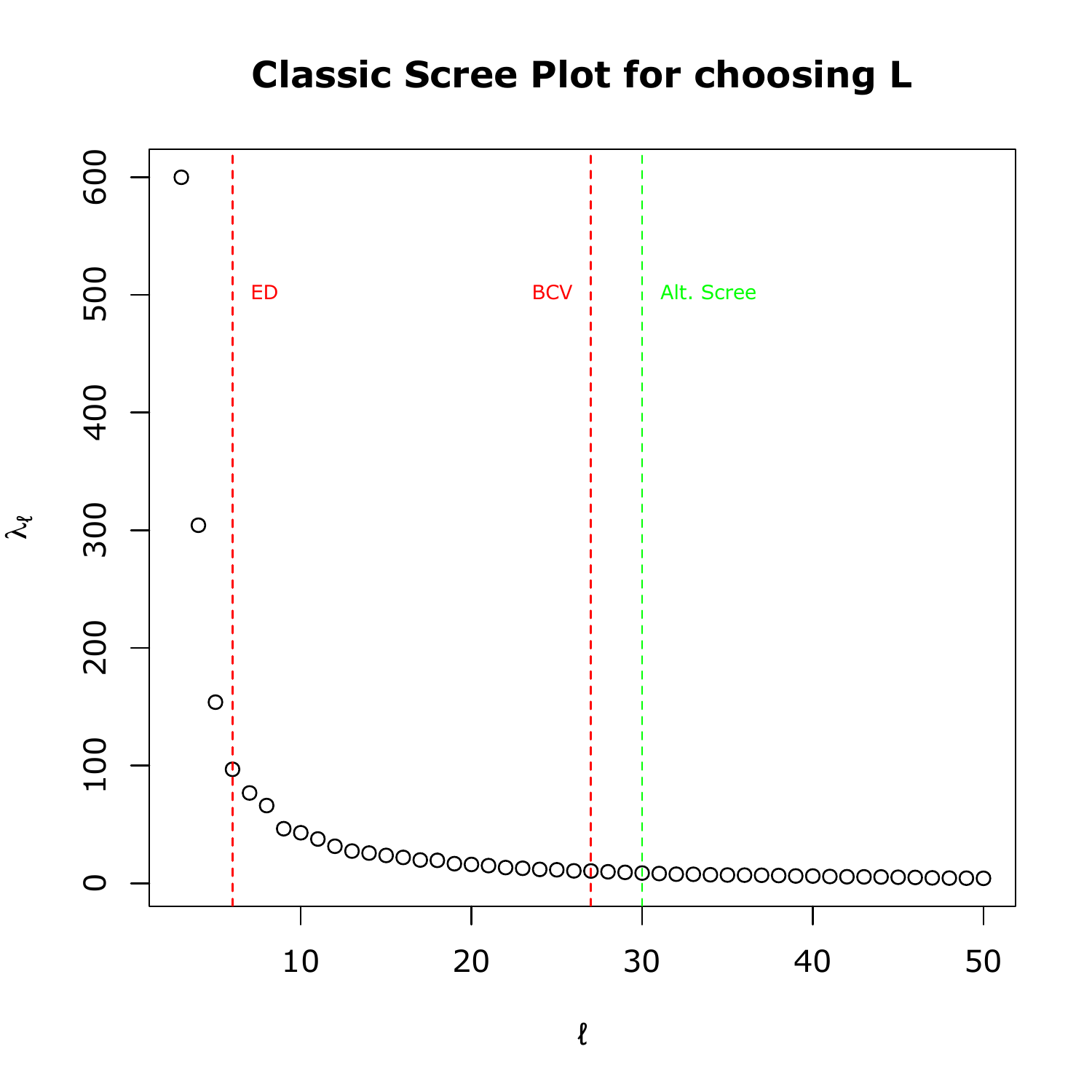}	
\caption{Alternative (left) and classic (right) Scree plots to estimate the number of factors for the spatial data example. The estimates $\hat{L}_{\text{ED}}$ and $\hat{L}_{\text{BCV}}$ are indicated in red.}
\label{fig:canadaLhat}
\end{center}
\end{figure}

In order to illustrate the fits of these models more clearly, we only show the first two months of the year for two specific stations in Figure~\ref{fig:canadafit}. The factor model follows the raw data very closely. The correlation matrices of FPC and PB residuals show strong and unsystematic correlations, indicating that the signal has not been accurately extracted (see Figure~\ref{fig:canadacor}). For the factor model, the empirical residual correlation matrix appears diagonal. Still, our independence test for the residuals rejects in all cases very clearly, see Table~\ref{tab:testcanada}. Observing the averaged periodogram ordinates we see that the factor model approach yields a pattern reminiscent of an autoregressive process (Figure~\ref{fig:canadaacf}). Analogously to Section~\ref{s:temporal}, we have estimated AR(2) processes for each residual curve and show the associated spectral densities for ten stations in Figure~\ref{fig:canadaacf}. In this instance, the results are not nearly as well-behaved as in the previous section. This stems from the fact that we are now considering different stations rather than different years for the same station. The residual structure may vary wildly in between the different stations, as one can see in Figure~\ref{fig:canadaacf}. Still, it is reasonable to assume that there is some mild temporal dependence between consecutive observations, which makes autoregressive error processes rather plausible.

\begin{figure}[!ht]
\begin{center}
\includegraphics[width=12cm]{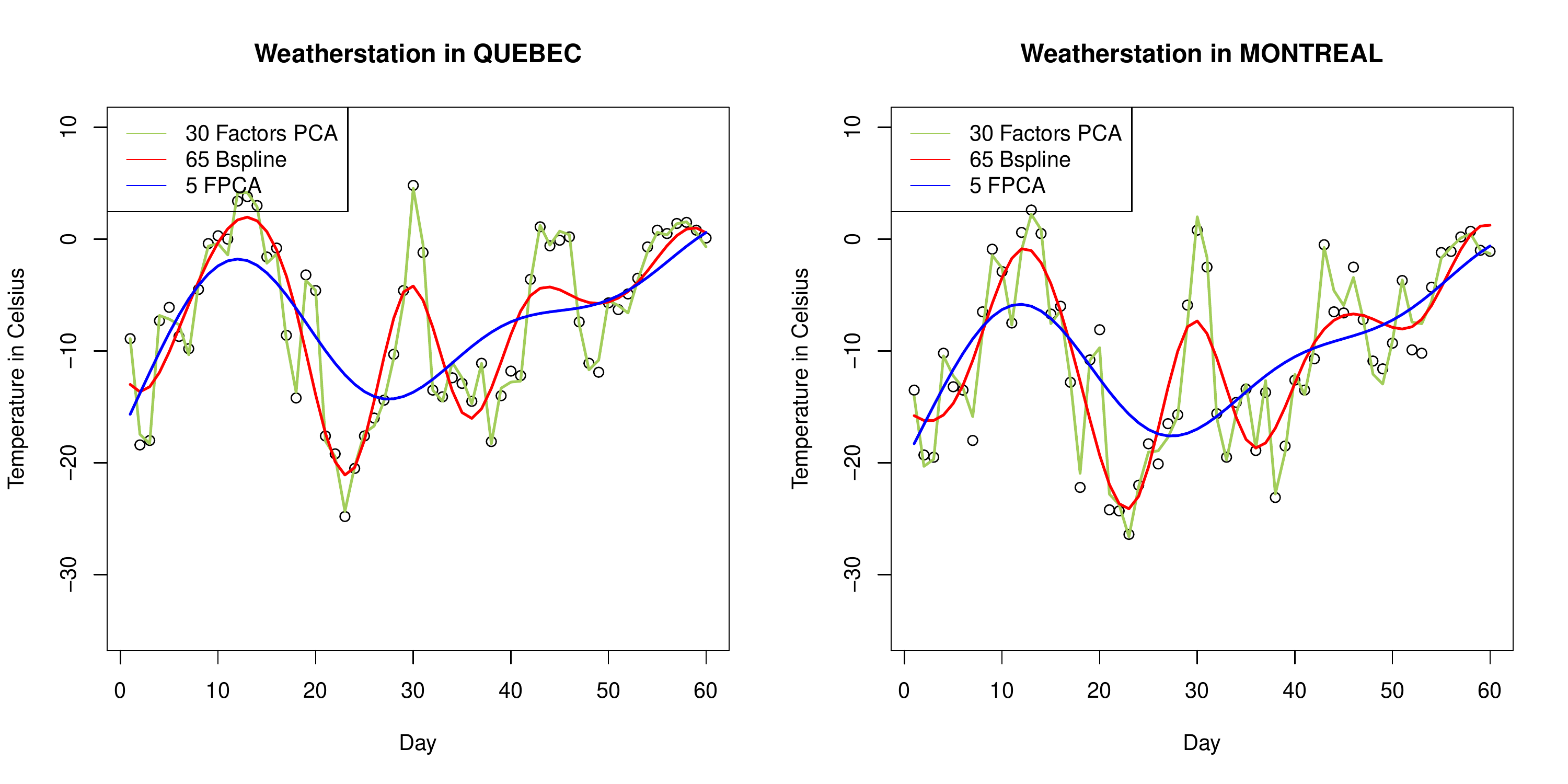}
\caption{Annual temperature curve model fit in 2013 for the stations \textit{Quebec} and \textit{Montreal}. Dots represent original observations.}
\label{fig:canadafit}
\end{center}
\end{figure}

\begin{figure}[!ht]
\begin{center}
\includegraphics[width=4.5cm]{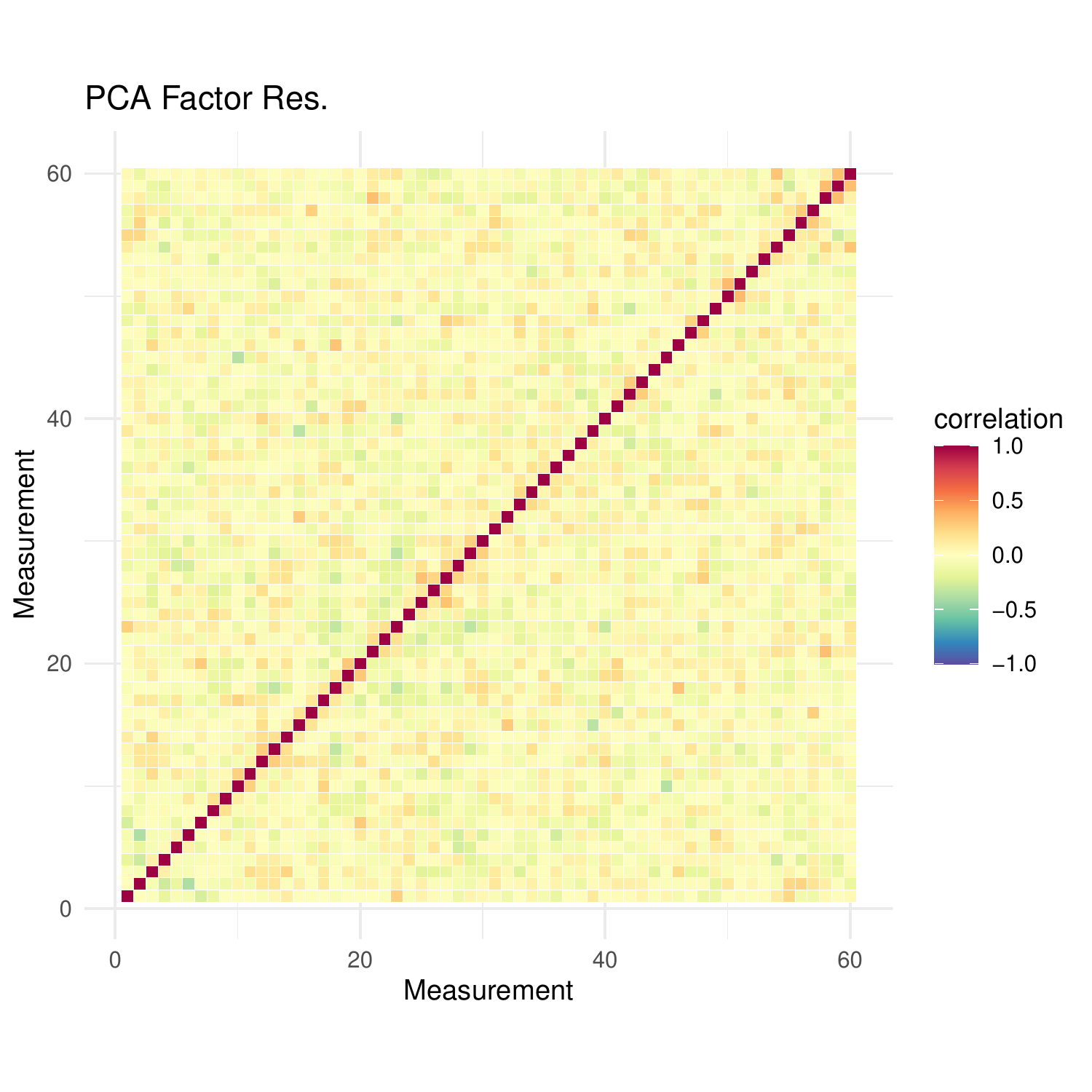}	
\includegraphics[width=4.5cm]{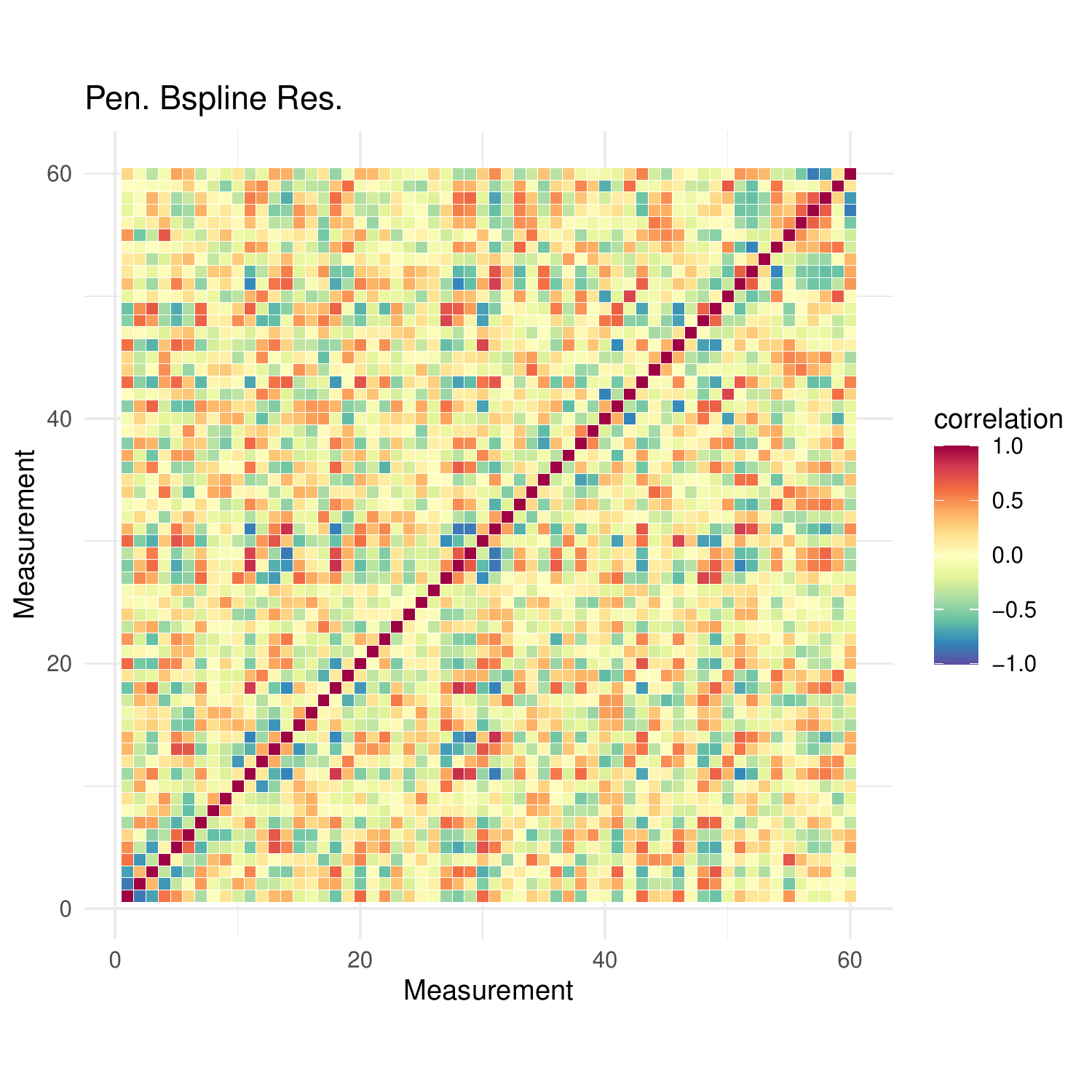}
\includegraphics[width=4.5cm]{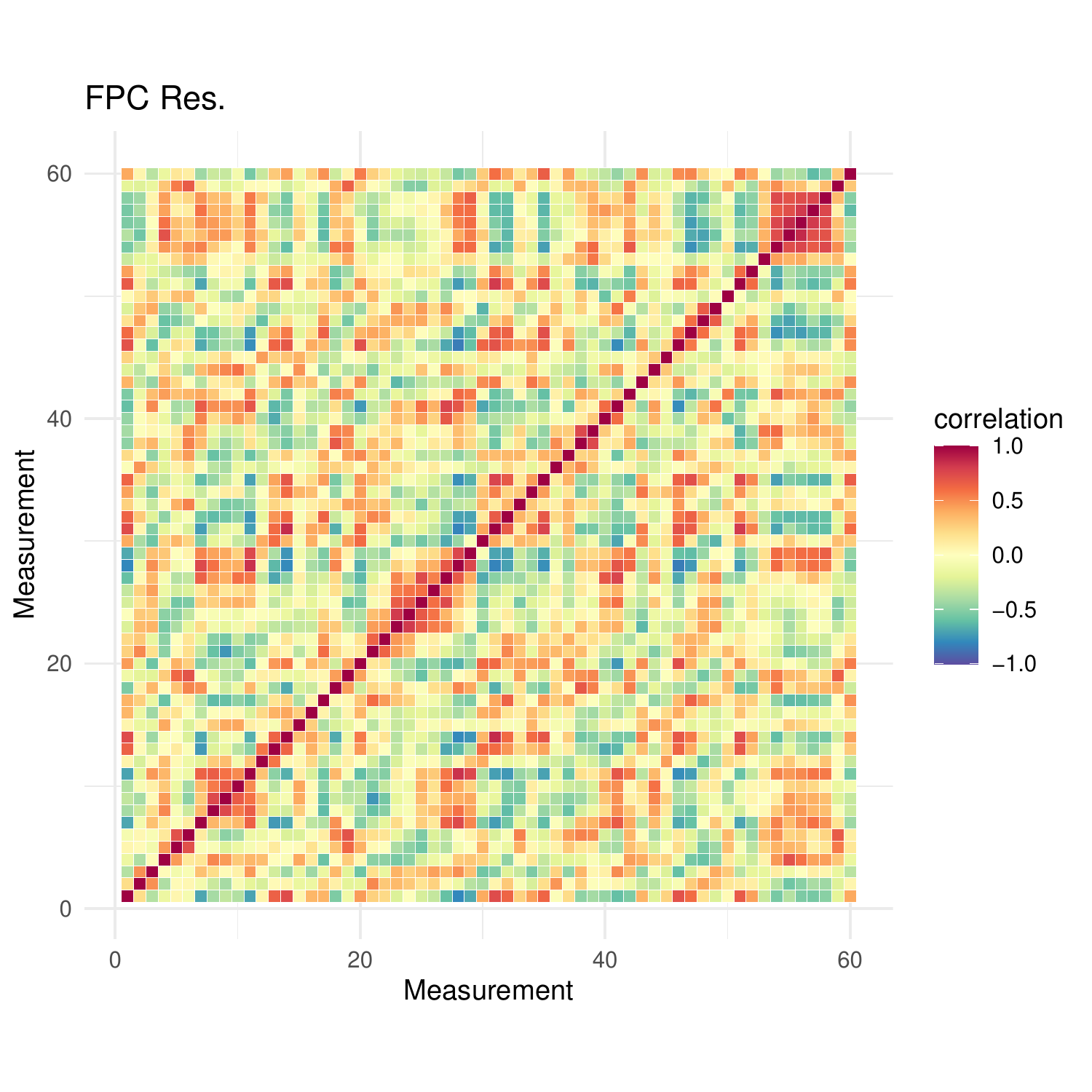}
\caption{Heat maps of the empirical correlation matrices for the PCA (left), PB (middle) and FPCA (right) approach for the spatial Canadian Weather Station Data.}
\label{fig:canadacor}
\end{center}
\end{figure}

\begin{figure}[!ht]
\begin{center}
\includegraphics[width=7.5cm]{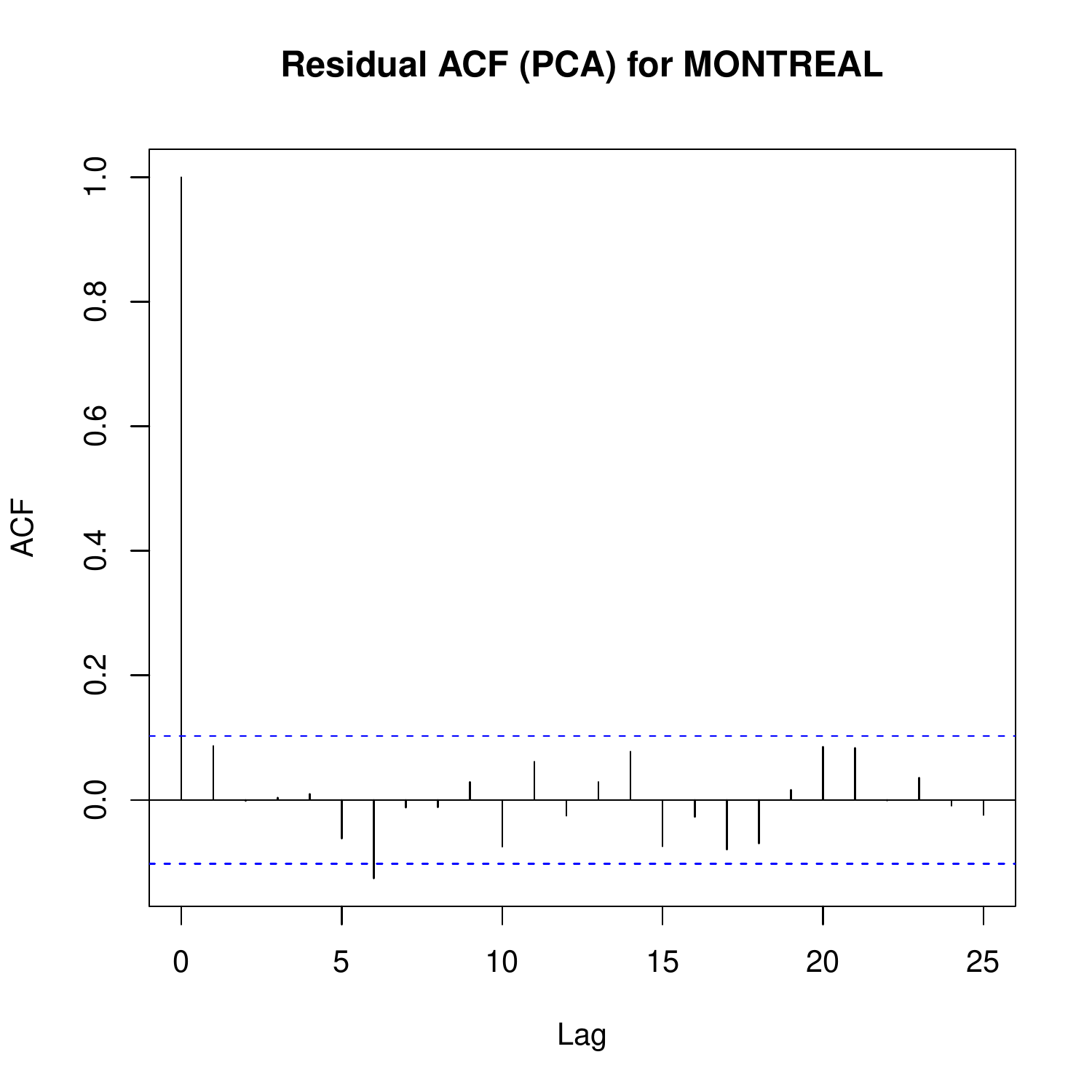}	
\includegraphics[width=7.5cm]{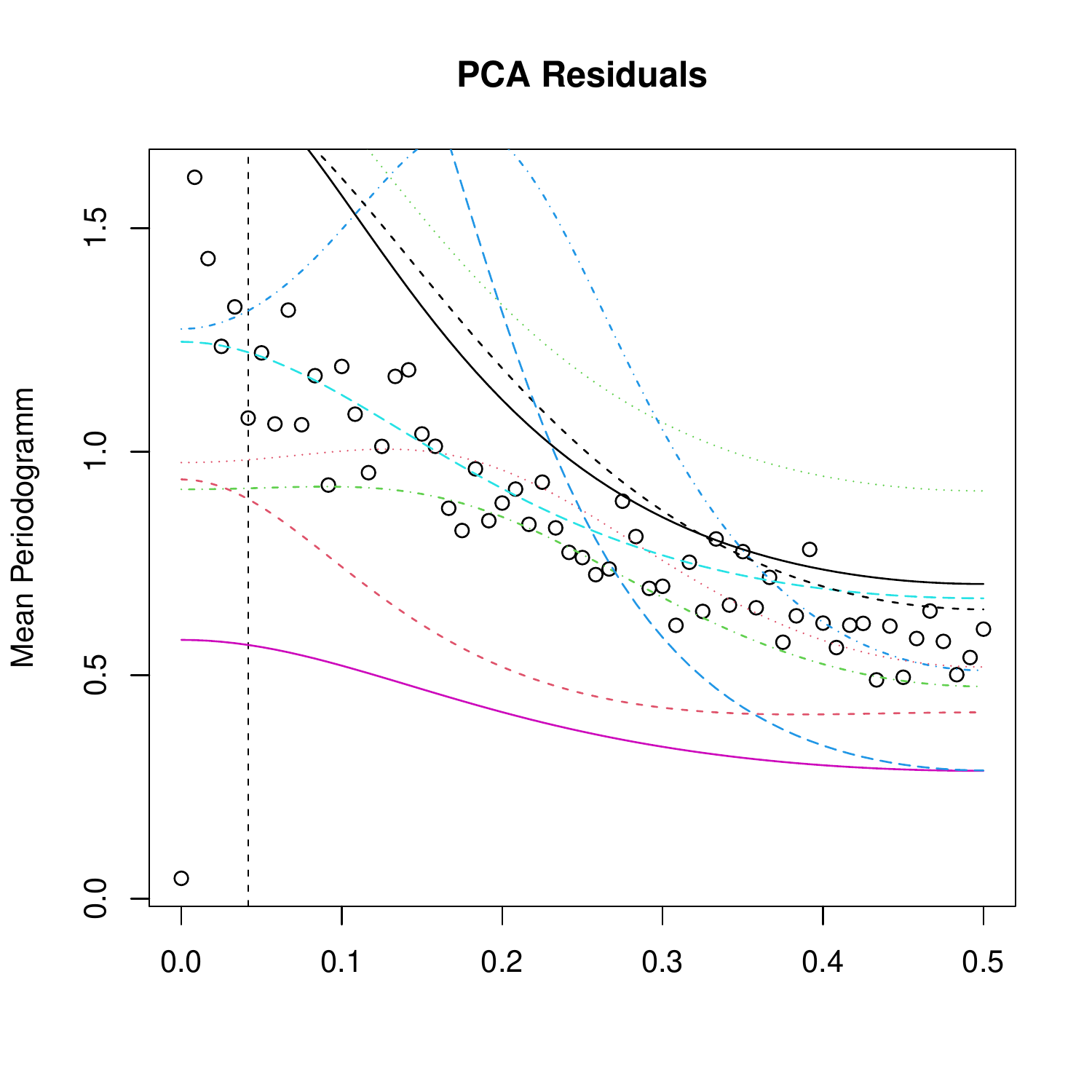}	
\caption{Autocovariance of the residual vectors in the station \textit{Montreal} for the factor model approach (left), and averaged periodograms $\xi=T^{-1} \sum_{t=1}^T I_{\hat{U}_t}(\bm\theta)$ (right) from the spatial temperature data. The vertical dotted line indicates the $10\%$ cutoff. Curves on the right indicate the estimated spectral density of estimated AR(2) processes for $10$ factor model residual curves.
}
\label{fig:canadaacf}
\end{center}
\end{figure}

\begin{table}[H]
\centering
\begin{tabular}{r|c|c|c}
  & PCA & PB & FPC \\ 
  \hline
$\hat\Lambda_{\text{inf}}$ & 107.26 & 12306.19 & 51776.87 \\ 
  \end{tabular}
\caption{Teststatistic for the independent noise test on the residuals of the spatial Canadian Weather Data.} 
\label{tab:testcanada}
\end{table}

The two examples illustrate the sophistication of factor models particularly well, as they truly manage to distinguish between systematic signal and noise in a very efficient way.

\section{Conclusion }
\label{s:outlook}
In this paper we give a multivariate perspective to the modelling of discretely observed functional data. We outline that such data follow some approximate factor models which are playing an important role in macroeconomics. This perspective yields ready to use methods to estimate the latent signal without requiring smoothness of the curves. We show that this approach works extremely well on simulated data and leads to interesting results on real data. Moreover, this paper offers some tools for analysing the model residuals. Typically those are assumed to be iid, but very often no residual analysis is done in order to justify this strong assumption. A theoretical foundation of the proposed estimation method is provided in our companion paper,  \citet{hormann:jammoul:2021}. 
\section*{Acknowledgement}

We thank Jeff Goldsmith and Sonja Greven for a very helpful discussion on the  \texttt{refund} package which we used to implement the FPCA method.

\bibliography{paper}
\appendix 

\section{Appendix}\label{s:appendix}
\subsection{Technical assumptions}\label{s:assumptions}

\begin{assumption}\label{a:noise}
The noise process $(U_t)$ is i.i.d.\ zero mean and independent of the signals~$(X_t)$. The processes $(U_{ti}\colon 1\leq i\leq p)$ are stationary and Gaussian with covariance function $\gamma^U(h) = \text{Cov}(U_{t,(i+h)},U_{ti})$, such that
$
\sum_{h\in\mathbb{Z}}|\gamma^U(h)|\leq C_U<\infty.
$
\end{assumption}

\begin{assumption}\label{a:signal}
(a) The process $(X_t\colon t\geq 1)$ is zero mean and $L^4$-$m$-approximable. (b) The curves $X_t=(X_t(s)\colon s\in [0,1])$ define fourth order random processes (i.e.\ $
\sup_{s\in[0,1]} EX_1^4(s)\leq C_X<\infty$)
with a continuous covariance kernel. (c) It holds that $E\sup_{s\in [0,1]} X^2_1(s)<\infty$. (d)~Observations $X_t$ lie in some $L$-dimensional function space, where $L=L(T)$ may diverge with $T\to\infty$.
\end{assumption}

\begin{assumption}\label{a:pcs}
For the eigenfunctions $\varphi_\ell$  it holds that $\max_{1\leq k,\ell\leq L}\left|p^{-1}\sum_{i=1}^p\varphi_k(s_i)\varphi_\ell(s_i)\right|=O(1)$ as $T\to\infty$.
\end{assumption}

\subsection{Proofs}\label{s:proofs}
We begin with an elementary lemma.
\begin{lemma}\label{l:cont}
Let us denote by $\hat X_t(s)$, $s\in [0,1]$, the interpolation of the estimates $\hat{X}_t(s_{i})$ as defined in \eqref{e:pcaapproach} and let $\omega^f(\delta)=\sup_{s,s^\pr \in [0,1]\colon |s-s^\pr|\leq\delta} | f(s)-f(s^\pr) |$ be the modulus of continuity of a function $f\colon [0,1]\to \mathbb{R}$. Then with $\delta=\max_{1\leq i \leq p-1} |s_{i+1}-s_i|$ we have
$$
\sup_{s \in \lbrack 0,1 \rbrack} | X_t(s) - \hat{X}_t(s) |\leq 2\omega^{X_t}(\delta)+\max_{1\leq i\leq p} |X_t(s_i)-\hat{X}_t(s_{i})|.
$$
\end{lemma}

The lemma shows that the approximation error of the full curve can be decomposed into the modulus of continuity of the functional data and the approximation error on the observation grid. 
The proof of Lemma~\ref{l:cont} can be easily seen and will thus be omitted.
%\begin{proof}[Proof of Lemma~\ref{l:cont}]
%Let $\tilde{X}_t(s)$ the interpolation of the unobservable $X_t(s_i)$. Then we consider
%\begin{align*}
%\sup_{s \in \lbrack 0,1 \rbrack} | X_t(s) - \hat{X}_t(s) | &\leq \sup_{s \in \lbrack 0,1 \rbrack} | X_t(s) - \tilde{X}_t(s) | + \sup_{s \in \lbrack 0,1 \rbrack} | \tilde{X}_t(s) - \hat{X}_t(s) | 
%=: A_1 + A_2.
%\end{align*}
%Note that $A_2 = \max_{i=1, \ldots,p} | X(s_i) - \hat{X}(s_i) |$ holds. Moreover, 
%\begin{align*}
%A_1 &\leq \max_{i=1,\ldots,p} \sup_{s \in \lbrack s_i, s_{i+1} \rbrack} | X_t(s) - X_t(s_i) | + \Big | \frac{s - s_i}{s_{i+1} - s_i} \Big | |X_t(s_{i+1}) - X_t(s_i)|,
%\end{align*}
%and now the conclusion is immediate.
%\end{proof}

\begin{proof}[Proof of Theorem~\ref{thm:cont}]
The first part of the proof follows immediately from Theorem~1 in \citet{hormann:jammoul:2021}, where it is shown that under the Assumptions~\ref{a:noise}--\ref{a:pcs} in the Appendix we have $$\max_{1\leq i\leq p} |X_t(s_i)-\hat{X}_t(s_{i})| = O_P\left( \frac{1}{T^{1/4}} + \frac{T^{1/4}}{\sqrt{p}} \right).$$
For the modulus of continuity $\omega^{X_t}(\delta)$ we may conclude with Markov's inequality that
\begin{align*}
P(\omega^{X_t}(\delta) > \kappa  \delta^\alpha ) &\leq  EM_t/\kappa.
\end{align*}
Thus we see that $\omega^{X_t}(\delta) = O_P(\delta^{\alpha})$ and the result follows using Lemma~\ref{l:cont}.
\end{proof}

\begin{proof}[Proof of Theorem~\ref{thm:eigenfun}]
We decompose the $\|\varphi_\ell-\tilde\varphi_\ell\|$ into three pieces. To this end, we define the empirical covariance operator $\hat\Gamma^X$ of the fully observed $X_1, \ldots, X_T$ and its eigenfunctions $\hat{\varphi}_\ell$. Let $X^\star_t(s) := X_t(s_i)$ for $s\in \lbrack s_i, s_{i+1})$ be a discretized version of the fully observed data and let the associated empirical covariance operator be denoted by $\hat{\Gamma}^{X^{\star}}$ and its eigenfunctions by $\hat{\varphi}_\ell^\star$. Finally, let us define the empirical covariance matrix $\hat{\Sigma}^X = T^{-1}XX^\prime$, where $X=(X_1(\bs),\ldots, X_T(\bs))$ and its associated eigenvectors $\hat{\psi}_\ell^X$. Consider 
\begin{equation}\label{e:0}
\|\varphi_\ell-\tilde\varphi_\ell\| \leq \Vert \varphi_\ell - \hat{\varphi}_\ell \Vert + \Vert \hat{\varphi}_\ell - \hat{\varphi}_\ell^\star \Vert + \Vert \hat{\varphi}_\ell^\star - \tilde{\varphi}_\ell \Vert.
\end{equation}
We may deduce from Weyl's theorem that
\begin{align}
\|\varphi_\ell-\hat\varphi_\ell\|&\leq\frac{2\sqrt{2}}{\alpha_\ell}\|\Gamma^X-\hat\Gamma^X\|,\label{e:1}\\
 \|\hat\varphi_\ell-\hat{\varphi}_\ell^\star\|&\leq\frac{2\sqrt{2}}{\hat\alpha_\ell}\Vert \hat{\Gamma}^X - \hat{\Gamma}^{X^{\star}} \Vert,\label{e:2}
\end{align}
where $\hat\alpha_\ell=\min\{\hat\lambda_\ell-\hat\lambda_{\ell+1},\hat\lambda_{\ell-1}-\hat\lambda_\ell\}$ and where $\hat\lambda_\ell$ are the empirical eigenvalues of the fully observed data.
From \citet{hoermann2010} it follows under Assumption~\ref{a:signal}~(a) that \eqref{e:1} is $O_P\left(T^{-1/2}\right)$ and 
that $\hat\alpha_\ell\to\alpha_\ell>0$, as $T\to\infty$.
Note that when $a(s,t)$ is the kernel of the bounded linear operator $A$, then $\|A\|^2\leq\int_0^1\int_0^1a^2(t,s)ds dt$. Hence
$$
\Vert \hat{\Gamma}^X - \hat{\Gamma}^{X^{\star}} \Vert^2\leq \int_0^1\int_0^1\left(\frac{1}{T}\sum_{t=1}^T(X_t(r)X_t(s)-X^\star_t(r)X^\star_t(s))\right)^2 dr ds.
$$
In the proof of Lemma~1 in \citet{hormann:jammoul:2021} it is shown that \eqref{a:cont} implies that the right hand side is $O_P\left(p^{-1}\right)$ as $p\to\infty$. We hence conclude that \eqref{e:2} is $O_P\left(p^{-1/2}\right)$.

In the final step, we have to move from the functional setting to the matrix setting. It can be readily seen that the eigenvector $\hat \psi_\ell^X$ of $\hat{\Sigma}^X$ satisfies $\sqrt{p}[\hat \psi_\ell^X]_i=\hat\varphi_\ell^y(s)$ for $s\in [(i-1)/p,i/p)$. 
Thus, we may rewrite the last term in \eqref{e:0} as  
$$
\|\hat\varphi_\ell^y-\tilde\varphi_\ell\|^2=(\hat \psi_\ell^X-\hat\psi_\ell^y)^\prime (\hat \psi_\ell^X-\hat\psi_\ell^y).
$$
Again by Weyl's theorem the right hand side is $O_P(\frac{1}{\hat\beta_\ell}\|\hat\Sigma^y-\hat\Sigma^X\|)$, where $\hat\beta_\ell=\min\{\hat\gamma_\ell-\hat\gamma_{\ell+1},\hat\gamma_{\ell-1}-\hat\gamma_\ell\}$.  Lemma~1 in \citet{hormann:jammoul:2021} implies that under Assumptions \ref{a:noise} and \ref{a:signal}(a) and (b) we have $$\hat\gamma_\ell-\hat\gamma_{\ell+1}\sim p(\lambda_\ell- \lambda_{\ell+1}).$$ Moreover, it is shown in this lemma that
$$
\|\hat\Sigma^y-\hat\Sigma^X\|=\begin{cases}O_P(\sqrt{p}),\quad \text{if $p/T\to\gamma\in [0,\infty)$};\\
O_P(p/\sqrt{T}),\quad\text{if $p/T\to\infty.$}\end{cases}
$$
Combining all bounds yields the desired convergence.
\end{proof}

\begin{proof}[Proof of Proposition~\ref{p:test}]
Suppose that the $Z=(Z_1,\ldots, Z_p)^\prime$ has iid components with $EZ_1=0$ and $EZ_1^2=\sigma^2$ and $EZ_1^4<\infty$. Denote $\kappa := EZ_1^4 - 3$. Then it is well known that for any fundamental frequency  we  have that $EI_Z(\theta_\ell)=\sigma^2$ and 
$$
\mathrm{Cov}(I_Z(\theta_\ell),I_Z(\theta_{\ell^\prime}))=
\begin{cases}
\sigma^4\kappa/p+\sigma^4\quad\text{if $\ell=\ell^\prime$};\\
\sigma^4\kappa/p\quad\text{else}.
\end{cases}
$$
We thus have that
the random vectors  $V_t = I_{U_t}(\bm{\theta}) - \sigma^2 1_f$ are  iid, zero-mean and $\Sigma:=\text{Var}(V_t) =  \sigma^4 (I_f + \frac{\kappa}{p} \mathbbm{1}_f) \in \mathbb{R}^{f \times f}$ holds. 
 Consider the centering matrix $P_f := I_f - f^{-1}\mathbbm{1}_f$ ($\mathbbm{1}_f$ is the matrix with entries equal to 1) and note that 
 \begin{align*}
TS^2_\xi & = (f-1)^{-1} \left\Vert T^{-1/2} \sum_{t=1}^T P_f V_t  \right\Vert^2.
\end{align*}
We also note that $P_f \Sigma = \sigma^4 P_f$ and recall the well known fact that $P_f$ has $f-1$ non-zero eigenvalues which are all equal to 1. If $Q$ denotes the orthogonal matrix which has in its columns the related eigenvectors, then 
\begin{align*}
TS^2_\xi & = (f-1)^{-1} \left\Vert T^{-1/2} \sum_{t=1}^T Q^\prime P_f V_t  \right\Vert^2 = (f-1)^{-1} \left\Vert T^{-1/2} \sum_{t=1}^T W_t  \right\Vert^2,
\end{align*}
where $(W_t^\prime, 0)^\prime := Q^\prime P_f V_t$. The vector $W_t$ is zero-mean and $\text{Var}(W_t) = \sigma^4 I_{f-1}$.
By the central limit theorem the expression inside the norm converges to a normally distributed vector with variance $\sigma^4 I_{f-1}$. The weak convergence of $\Lambda_\text{inf}$ then follows by the continuous mapping theorem and Slutzky's lemma. 

For growing $f$ we consider the variable $\Lambda_\text{inf}= (TS^2_\xi/\sigma^4 - 1)\sqrt{(f-1)/2}$ and we wish to compare its distribution to the normal distribution, with its distribution function denoted by $\Phi(z)$. To this end let $Z \sim N(\mathbf{0},  I_{f-1})$ be a $(f-1)$-variate standard normal random vector. For any $z\in\mathbb{R}$ we get by  the central limit theorem that
$$
\bigg|P\bigg(\Big(\frac{1}{(f-1)}\Vert Z\Vert^2 - 1\Big)\sqrt{\frac{(f-1)}{2}} \leq z\bigg)-\Phi(z)\bigg|\to 0.
$$
Hence, it suffices to show that for all real $z$ we have
$$
\left\vert P(\Lambda_\text{inf} \leq z) - P\bigg(\Big(\frac{1}{(f-1)}\Vert Z \Vert^2 - 1\Big)\sqrt{\frac{(f-1)}{2}} \leq z\bigg) \right\vert \to 0.
$$
By Slutzky's lemma we can replace $\hat\sigma^4$ in the definition of $\Lambda_\text{inf}$ by $\sigma^4$. 
With $\tilde z=\sqrt{(\frac{z\sqrt{2}}{\sqrt{f-1}} + 1)(f-1)}$ we hence need to show that
\begin{align}
\left\vert P\Big(\big\| T^{-1/2} \sum_{t=1}^T W_t /\sigma^2\big\| \leq \tilde{z}\Big) - P(\Vert Z\Vert \leq \tilde{z}) \right\vert.\label{e:normapp}
\end{align}
If we can show that $E|W_{ti}|^4$ are uniformly bounded (in $t$ and $i$), then by Corollary~3.1 in \citet{fang:koike:2021} we get that for any $\tilde z$ the term \eqref{e:normapp} is bounded by
$$
C\left(T^{-1/8}+(f/T)^{1/6}\right),
$$
for some constant $C$ which is independent of $T$ and $f$.
 Thus we can guarantee convergence if $f/T \to 0$. 
 
We want to show that 
 $\max_{1\leq i\leq f-1}\max_{1\leq t\leq T}\text{E}|W_{ti}|^4 < C$, where $C$ does not depend on the dimension parameters $f$, $p$ and $T$. It holds that $W_{ti} = v_i^\prime V_t$, where $v_i$ denotes the $i$-th column of the matrix $Q$ and is thus an eigenvector of $P_f$ belonging to a non-zero eigenvalue. It can be easily checked that for $f\geq 3$ the $v_i$ can be written as $(0,\ldots, 1/\sqrt{2},0,\ldots,0,-1/\sqrt{2})^\prime$, with non-zero entries at the $i$-th and the last coordinate. Since we assume iid noise $(U_t$,  $t\geq 1)$, it follows that $(W_{ti}$, $t\geq 1)$ are iid as well and thus the expectations do not depend on $t$. Hence, let us consider $\text{E}|v_i^\prime V|^4$, where  $V=(I_{U}(\theta_{j_1}),\ldots, I_{U}(\theta_{j_f}))^\prime-\sigma^2 1_f$, with $\{j_1,\ldots, j_f\}=\mathcal{F}$ and $U=(u_1,\ldots, u_p)^\prime\sim U_1$.  We assume without loss of generality that $\sigma^2=1$. Then the $k$-th component of $V$ is given by
\begin{align*}
& V_k^c+V_k^s:=\frac{1}{p}\left(\sum_{r=1}^p u_{r}\cos(\theta_{j_k} r)\right)^2-1/2+\frac{1}{p}\left(\sum_{r=1}^p u_{r}\sin(\theta_{j_k} r)\right)^2-1/2.
\end{align*}
We have 
\begin{align*}
E(v^\prime V^c+v^\prime V^s)^4&\leq 16 \left(E(v^\prime V^c)^4+E(v^\prime V^s)^4\right)\\
&\leq 64 \left(E(V_i^c)^4+E(V_f^c)^4+E(V_i^s)^4+E(V_f^s)^4\right).
\end{align*}
All the terms on the right can be bounded in the same way. Let us consider $E(V_k^c)^4$. Noting that $\theta_{j_k} r=\theta_r j_k$ and $\sum_{r=1}^p \cos^2(\theta_r j_k)=p/2$ it can be written as 
\begin{align*}
E(V_k^c)^4&=\frac{1}{p^4}E\left(\left(\sum_{r=1}^p u_{r}\cos(\theta_r j_k)\right)^2-p/2\right)^4\\
&\leq \frac{1}{p^4}E \left(\sum_{r=1}^p u_{r}\cos(\theta_r j_k)\right)^8. 
%&=\frac{1}{p^4}E\left(\sum_{r=1}^p\sum_{s=1}^p \left(u_r u_s-\delta_{rs}\right)\cos(\theta_r j_k)\cos(\theta_s j_k)\right)^4\\
%&=\frac{1}{p^4}\sum_{r_1,\ldots, r_4=1}^p\sum_{s_1,\ldots, s_4=1}^p
%\left(\prod_{i=1}^4  \cos(\theta_{r_i} j_k)\cos(\theta_{s_i} j_k)\right) E\left[\prod_{i=1}^4\left(u_{r_i} u_{s_i}-\delta_{r_i s_i}\right)\right].
\end{align*}
The last inequality follows from the fact, that $E(X-EX)^4\leq EX^4$ when $X$ is a positive random variable. Now apply the Rosenthal inequality (see e.g.\ \cite{petrov}).
\end{proof}
 
\begin{proof}[Proof of Proposition~\ref{p:alt}]
For the proof it suffices to show that  $P(S_\xi^2>\delta/2)\to 1$, $T\to\infty$. Now we have
$$S_\xi^2 = S_{  g}^2 + S_{\xi -  g}^2 + 2S_{ g ,\xi -  g} \geq  S_{ g}^2-2|S_{g,\xi -  g}|$$
where $S_{\xi -  g}^2$ is defined analogously to $S_\xi^2$ and where 
\begin{align*}
S_{ g ,\xi -  g} &= \frac{1}{f-1} \sum_{j=1}^f ( g (\theta_{\ell_j}) - \bar{g} ) ( \xi_j -  g (\theta_{\ell_j}) - (\bar{\xi} - \bar{g} )) \\
&= \frac{1}{f-1} \sum_{j=1}^f ( g (\theta_{\ell_j}) -  \bar{g} ) ( \xi_j -  g (\theta_{\ell_j})).
\end{align*}
By \eqref{e:spectnonconst} $S_{ g}^2 >  \delta$ and it remains to show that $|S_{ g ,\xi - g} |\to 0$ in probability for $p\to\infty$. It is easy to see that $S_{g}^2\leq 2\left(\sum_{h\in\mathbb{Z}}|\gamma_U(h)|\right)^2<\infty$ and by Markov's inequality we have $|S_{ g ,\xi - g} |^2\leq S_g^2 S_{\xi-g}^2$. Hence the claim follows if we can show that $S_{\xi-g}^2\to 0$ in probability. 

To this end we recall that 
\begin{equation}\label{e:uniformspectral}
\sup_{\theta \in \lbrack -\pi, \pi \rbrack} \vert E I_{U_t}(\theta) -  g (\theta) \vert \to 0\quad (p\to\infty).
\end{equation}
(See e.g.\ Proposition 10.3.1 in \citet{brockwell:davis:1991}.)
Hence, when $p$ is large enough, we have
\begin{align*}
& P\left( \max_{\ell_j \in \mathcal{F}} \left\vert \xi_j -   g (\theta_{\ell_j}) \right\vert > \varepsilon \right)\leq \sum_{\ell \in \mathcal{F}}P \left( \left\vert \frac{1}{T} \sum_{t=1}^T (I_{U_t}(\theta_\ell) -  g  (\theta_\ell)) \right\vert > \varepsilon \right)\\
\leq &\; f \max_{\ell \in \mathcal{F}} P \left( \left\vert \frac{1}{T} \sum_{t=1}^T (I_{U_t}(\theta_\ell) -EI_{U_t}(\theta_\ell)) \right\vert + \frac{1}{T} \sum_{t=1}^T \left \vert EI_{U_t}(\theta_\ell)) -   g  (\theta_\ell) \right\vert > \varepsilon \right)\\
\leq &\; f \max_{\ell \in \mathcal{F}}  P \left( \left\vert \frac{1}{T} \sum_{t=1}^T (I_{U_t}(\theta_\ell) -EI_{U_t}(\theta_\ell)) \right\vert > \varepsilon/2 \right) \leq \frac{4 f}{\varepsilon^2} \max_{\ell \in \mathcal{F}} E\left( \frac{1}{T} \sum_{t=1}^T (I_{U_t}(\theta_\ell) -EI_{U_t}(\theta_\ell)) \right)^2 \\
\leq &\frac{4f}{\varepsilon^2 T^2} \max_{\ell \in \mathcal{F}} \sum_{t=1}^T \sum_{t^\prime=1}^T \text{Cov}(I_{U_t}(\theta_\ell), I_{U_{t^\prime}}(\theta_\ell)) = \frac{4f}{\varepsilon^2T}  \max_{\ell \in \mathcal{F}} \text{Var}I_{U_t}(\theta_\ell)\to 0.
\end{align*}
\end{proof}
\end{document}